\documentclass[a4paper]{scrartcl}
\usepackage{amssymb}
\usepackage{amsmath}
\usepackage{amsthm}
\usepackage{latexsym}
\usepackage{enumerate}
\usepackage{graphicx,color}
\usepackage{dsfont}
\usepackage{float}
\usepackage[textwidth=17mm]{todonotes}
\usepackage{placeins}
\usepackage[title]{appendix}
\usepackage{etoolbox}
\patchcmd{\appendices}{\quad}{: }{}{}
\usepackage[normalem]{ulem}
\usepackage{lmodern}

\DeclareFontFamily{OT1}{pzc}{}
\DeclareFontShape{OT1}{pzc}{m}{it}{<-> s * [1.10] pzcmi7t}{}
\DeclareMathAlphabet{\mathpzc}{OT1}{pzc}{m}{it}

\usepackage{xcolor}
\usepackage{enumitem}
\usepackage{microtype}

\usepackage[round,authoryear]{natbib}
\newcommand*{\doi}[1]{\href{http://dx.doi.org/#1}{doi: #1}}
\usepackage[autostyle=true]{csquotes} 

\makeatletter
\newcommand*\bigcdot{\mathpalette\bigcdot@{.5}}
\newcommand*\bigcdot@[2]{\mathbin{\vcenter{\hbox{\scalebox{#2}{$\m@th#1\bullet$}}}}}
\makeatother

\numberwithin{equation}{section}

\numberwithin{figure}{section}

\theoremstyle{plain}
\newtheorem{theorem}{Theorem}[section]
\newtheorem{proposition}[theorem]{Proposition}

\newtheorem{lemma}[theorem]{Lemma}

\theoremstyle{definition}

\newtheorem{assumption}[theorem]{Assumption}

\newtheorem{example}[theorem]{Example}
\newtheorem{remark}[theorem]{Remark}

\usepackage{tikz}
\usetikzlibrary{arrows,positioning,calc}
\usetikzlibrary{decorations.pathreplacing}
\tikzset{
    >=stealth',
    punkt/.style={
           rectangle,
           rounded corners,
           draw=black, thick,
           text width=5em,
           minimum height=2em,
           text centered},
    punktm/.style={
           rectangle,
           rounded corners,
           draw=black, thick,
           text width=5.5em,
           minimum height=2em,
           text centered},
    punktl/.style={
           rectangle,
           rounded corners,
           draw=black, thick,
           text width=7em,
           minimum height=2em,
           text centered},
    punktll/.style={
           rectangle,
           rounded corners,
           draw=black, thick,
           text width=7.5em,
           minimum height=2em,
           text centered},           
    pil/.style={
           ->,
           shorten <=4pt,
       shorten >=4pt
    },
    pildotted/.style={
           ->,
           shorten <=4pt,
           shorten >=4pt,
  dotted,
  },
    punktf/.style={
           rectangle,
           text width=4.0em,
           minimum height=1.5em,
           text centered},
    punktfTop/.style={
           rectangle,
           text width=4.0em,
           minimum height=1.5em,
           text centered,
           append after command={
               [thick,shorten >=0.2bp, shorten <=0.2bp]
               (\tikzlastnode.north west)edge(\tikzlastnode.north east)
}
    },
    punktfBot/.style={
           rectangle,
           text width=4.0em,
           minimum height=1.5em,
           text centered,
           append after command={
               [thick,shorten >=0.2bp, shorten <=0.2bp]
               (\tikzlastnode.south west)edge(\tikzlastnode.south east)
            }
    }
}

\usepackage{pgfplots}
\pgfplotsset{width=7cm,compat=1.18}

\newcommand{\indep}{\perp \!\!\! \perp}

\colorlet{lightred}{red!20}
\colorlet{lightgray}{gray!25}
\definecolor{platinum}{RGB}{229,228,226}
\definecolor{ash}{RGB}{178,190,181}
\definecolor{lightblue}{HTML}{b3cde0}

\newcommand\xqed[1]{%
  \leavevmode\unskip\penalty9999 \hbox{}\nobreak\hfill
  \quad\hbox{#1}}
\newcommand\demormk{\xqed{$\triangledown$}}
\newcommand\demoex{\xqed{$\circ$}}
\newcommand\demoas{\xqed{$\diamond$}}

\makeatletter
\newcommand{\pushright}[1]{\ifmeasuring@#1\else\omit\hfill$\displaystyle#1$\fi\ignorespaces}
\newcommand{\pushleft}[1]{\ifmeasuring@#1\else\omit$\displaystyle#1$\hfill\fi\ignorespaces}
\makeatother

\usepackage{hyperref}

\makeatletter
\DeclareFontFamily{OMX}{MnSymbolE}{}
\DeclareSymbolFont{MnLargeSymbols}{OMX}{MnSymbolE}{m}{n}
\SetSymbolFont{MnLargeSymbols}{bold}{OMX}{MnSymbolE}{b}{n}
\DeclareFontShape{OMX}{MnSymbolE}{m}{n}{
    <-6>  MnSymbolE5
   <6-7>  MnSymbolE6
   <7-8>  MnSymbolE7
   <8-9>  MnSymbolE8
   <9-10> MnSymbolE9
  <10-12> MnSymbolE10
  <12->   MnSymbolE12
}{}
\DeclareFontShape{OMX}{MnSymbolE}{b}{n}{
    <-6>  MnSymbolE-Bold5
   <6-7>  MnSymbolE-Bold6
   <7-8>  MnSymbolE-Bold7
   <8-9>  MnSymbolE-Bold8
   <9-10> MnSymbolE-Bold9
  <10-12> MnSymbolE-Bold10
  <12->   MnSymbolE-Bold12
}{}

\let\llangle\@undefined
\let\rrangle\@undefined
\DeclareMathDelimiter{\llangle}{\mathopen}%
                     {MnLargeSymbols}{'164}{MnLargeSymbols}{'164}
\DeclareMathDelimiter{\rrangle}{\mathclose}%
                     {MnLargeSymbols}{'171}{MnLargeSymbols}{'171}
\makeatother

\newlength{\leftstackrelawd}
\newlength{\leftstackrelbwd}
\def\leftstackrel#1#2{\settowidth{\leftstackrelawd}%
{${{}^{#1}}$}\settowidth{\leftstackrelbwd}{$#2$}%
\addtolength{\leftstackrelawd}{-\leftstackrelbwd}%
\leavevmode\ifthenelse{\lengthtest{\leftstackrelawd>0pt}}%
{\kern-.5\leftstackrelawd}{}\mathrel{\mathop{#2}\limits^{#1}}}

\newcommand*\diff{\mathop{}\!\mathrm{d}}

\newcommand{\stkout}[1]{\ifmmode\text{\sout{\ensuremath{#1}}}\else\sout{#1}\fi}

\usepackage{authblk}

\definecolor{ForestGreen}{RGB}{34,139,34}

\title{A multistate approach to disability insurance reserving with information delays} 

\author[1,2,$\star$]{Oliver Lunding Sandqvist}

\affil[1]{\footnotesize PFA Pension, Sundkrogsgade 4, DK-2100 Copenhagen \O, Denmark.}
\affil[2]{\footnotesize Department of Mathematical Sciences, University of Copenhagen, Universitetsparken 5, DK-2100 Copenhagen \O, Denmark.}
\affil[$\star$]{\footnotesize Corresponding author. E-mail: \href{mailto:oliver.s@math.ku.dk}{oliver.s@math.ku.dk}.}

\date{\vspace{-8mm}}

\begin{document}

\maketitle

\begin{abstract}

Disability insurance claims are often affected by lengthy reporting delays and adjudication processes. The classic multistate life insurance modeling framework is ill-suited to handle such information delays since the cash flow and available information can no longer be based on the biometric multistate process determining the contractual payments. We propose a new individual reserving model for disability insurance schemes which describes the claim evolution in real-time. Under suitable independence assumptions between the available information and the underlying biometric multistate process, we show that these new reserves may be calculated as natural modifications of the classic reserves. For estimation of the model constituents, we employ the procedure proposed in~\citet{Buchardt.etal:2023b}. A real data application illustrates how the proposed estimation and reserving methods can be implemented and compares the resulting reserves with simpler reserving approaches.
\end{abstract}

\vspace{5mm}


\vspace{5mm}


\vspace{2mm}


\vspace{5mm}

\section{Introduction} \label{sec:Introduction}

Insurance contracts are defined in terms of the times at which insured events occur. Disability benefits, for example, depend on when an individual becomes disabled, recovers, or dies. In practice, however, these events are not observed by the insurer immediately. Disability claims may be reported only after a delay, and even after reporting, additional time may be required for adjudication (i.e., determining whether the claim satisfies the contractual conditions for benefit payments). Consequently, the information available to the insurer at a given point in time generally differs from the underlying insurance process. This discrepancy gives rise to reserves for claims that have occurred but remain unreported or unsettled and therefore constitutes an important source of insurance risk.

While reporting and settlement delays have long been recognized in non-life insurance reserving, they are much less studied in life and health insurance. Traditional multistate life insurance models typically assume that the insurer continuously observes the insured's state and can therefore condition reserves on complete information. This assumption is often reasonable for mortality products but is less appropriate for disability insurance, where substantial reporting and adjudication delays frequently occur. As a result, the reserving problem falls outside the classical multistate framework whenever the insurer must value future liabilities based on incomplete information.

A general framework for addressing this issue was recently proposed by~\citet{Buchardt.etal:2023a}, who distinguish between a \textit{valid time process} describing when events actually occur and a \textit{transaction time process} describing when these events become known to the insurer. This framework provides a mathematically consistent separation between the underlying insurance process and the insurer's information process. However, the relationship between valid time and transaction time reserves depends strongly on the chosen transaction time model. While~\citet{Buchardt.etal:2023a} established the general framework and illustrated it in a simplified example, explicit reserve expressions for realistic disability insurance models with both reporting delays and claim adjudication have not yet been derived.

The objective of this paper is to develop such a model. We formulate a transaction time disability insurance model that incorporates both delayed reporting and delayed claim adjudication while preserving the semi-Markov structure commonly used in disability insurance. Under suitable conditional independence assumptions, we derive explicit expressions for transaction time reserves and show how they naturally decompose into covered-but-not-incurred (CBNI), incurred-but-not-reported (IBNR), and reported-but-not-settled (RBNS) components. The resulting formulas provide an interpretable connection between the information available to the insurer and the classic multistate reserves.

The theoretical results are complemented by a practical estimation framework. We show how the required quantities can be estimated using recently developed methods for multistate models observed subject to reporting delays and incomplete event adjudication, thereby providing a complete framework from statistical estimation to reserve calculation.

The remainder of the paper is organized as follows. Section~\ref{sec:literatureReview} reviews the relevant literature and positions the present contribution within the multistate life insurance, individual non-life reserving, and statistical estimation literatures. Section~\ref{sec:setup}  introduces the disability insurance model in both valid time and transaction time. Section~\ref{sec:reserving} derives explicit reserve expressions, which constitute the paper's principal contribution. Estimation is discussed in Section~\ref{sec:estimation}. Section~\ref{sec:application} contains an illustrative data application demonstrating how the proposed reserve formulas can be implemented in practice. Section~\ref{sec:conclusion} concludes. Lengthy proofs are deferred to Appendix~\ref{sec:StrongMarkovProof}. The supplementary material contains additional discussions on the modeling assumptions, the extension to stochastic interest rates, and additional details on the estimation procedure.

\section{Literature review} \label{sec:literatureReview}

We provide a concise review of the literature most relevant to this paper and the gap it addresses. A broader review can be found in Section 4 of the more accessible paper~\citet{Furrer:Sandqvist:2026}, which examines recent developments and remaining challenges in disability insurance and seeks to bridge the gap between academic research and practical actuarial challenges.

\subsection{Multistate life insurance reserving}

Multistate models constitute the standard framework for modeling life and health insurance contracts involving transitions between several health states. Early contributions include~\citet{Janssen:1966}, who introduced semi-Markov models for disability insurance, and~\citet{Hoem:1972}, who generalized the approach to time-inhomogeneous settings and established its theoretical basis for actuarial and demographic applications. These ideas have subsequently become well established in actuarial science through, among others, \citet{Haberman:Pitacco:1998},~\citet{Helwich:2008},~\citet{Christiansen:2012}, and~\citet{Buchardt.etal:2015}.

The popularity of semi-Markov models in disability insurance stems from several attractive features. First, they naturally allow transition intensities, such as recovery and mortality rates, to depend on the duration since disability onset, which is well known to be important in practice. Second, contractual benefit structures frequently depend on disability duration through waiting periods or duration-dependent payments, which can be represented directly within the semi-Markov framework. Finally, multistate models provide a parsimonious representation of insurance contracts while retaining a probabilistic structure that leads to tractable reserve calculations.

A common assumption throughout this literature is that the insurer continuously observes the insured's state. This assumption is frequently violated in disability insurance due to reporting delays and claim adjudication. Recently,~\citet{Buchardt.etal:2023a} proposed a general framework that can accommodate such information delays. However, operationalizing this framework to obtain explicit reserve formulas for realistic disability insurance models remains unexplored. The present paper develops such a model. 

The derivation of the reserves is based on stochastic analysis that falls outside of existing results and techniques, because we are led to analyze the biometric state process stopped at a random time that is not a stopping time with respect to the filtration of interest, namely the filtration generated by the biometric state process. Such complications did not arise in the simple model from Example 5.8 of~\citet{Buchardt.etal:2023a}, which did not consider reporting and adjudication delays, but only uncertainty about the type of disability (work-related or non-work-related). In that setting, a simple conditional independence assumption was sufficient. The treatment of the mathematical complexities arising from incorporating reporting and adjudication delays into the model is one of the main contributions of the paper.

\subsection{Individual reserving in non-life insurance}

The problem of incomplete information has long been recognized in non-life insurance reserving. Classical aggregate reserving methods, such as chain-ladder techniques~\citep{Mack:1993,Mack:1999}, can accommodate IBNR and RBNS at the portfolio level, but provide limited information about individual claims. Since disability claims frequently lead to several decades of benefits, there is generally no complete run-offs in the data, so structural assumptions have to be imposed in order to extrapolate the future expenses for late development years. The most natural way to do this is via individual claims reserving. For discussions on some general disadvantages of aggregate reserving, see~\citet{Antonio:Plat:2014} and~\citet{Lopez.etal:2018}. 

The field of individual claims reserving in non-life insurance is growing rapidly,  so we refer to the survey~\citet{Boumezoued:Devineau:2017} for a historical overview, and focus on models that accommodate both IBNR and RBNS at the individual policy level. \citet{Crevecoeur.etal:2022b} seems to be the only such paper, and the authors examine both reserving and estimation. They consider a Poisson model for claims, where IBNR is accommodated by estimating the claim intensity and reporting probabilities via an EM-algorithm, and RBNS is accommodated by modeling the conditional distribution of the full claim development, consisting of all payments and auxiliary characteristics of the claim, given the historical development. Reserves are computed via Monte-Carlo methods. Most papers only consider RBNS (e.g.,~\citet{Antonio.etal:2016},~\citet{Lopez.etal:2019},~\citet{Okine.etal:2022},~\citet{Avanzi.etal:2026}), while the remaining papers either do not use an exposure measure (e.g.,~\citet{Wuthrich:2018},~\citet{Delong.etal:2022},~\citet{Vanegas.etal:2025}), only use an aggregate exposure measure (e.g.,~\citet{Antonio:Plat:2014},~\citet{Badescu.etal:2016}), or only consider IBNR (e.g.,~\citet{Bucher:Rosenstock:2024}). Having an exposure measure allows CBNI and IBNR reserves to adjust automatically as the portfolio changes in size. When exposure is measured at the policy level, reserves also adapt to changes in the portfolio's covariate mix.

Although the non-life reserving literature provides important insights into delayed observation and claim settlement, the structural modeling choices and assumptions that are employed are not natural for disability insurance. The main difference is that the contractual payments in disability insurance can be characterized rather simply via transitions between health states, while non-life coverages have less structure, so the distribution of repeated claim occurrences and payments generally has to be learned directly from the data. Having to learn the full claim development dynamics results in a larger number of model elements, and thus a greater risk of misspecification and estimation error. 

Our approach requires additional conditional independence assumptions between the observed information and the underlying biometric state process, but in return, we only need two extra model elements in addition to what is usually modeled in the multistate approach, namely the reporting delay distribution and the adjudication probabilities. Furthermore, we derive analytically tractable expressions for the transaction time reserves that can be evaluated with computational complexity comparable to that of standard valid time reserves. This eliminates the need for Monte Carlo simulation, which is ubiquitous in non-life individual claims reserving.

\subsection{Estimation for delayed and censored multistate data}

The statistical estimation of multistate models under incomplete observation has received considerable attention.~\citet{Lagakos.etal:1978},~\citet{Andersen.etal:1993},~\citet{Dabrowska:1995},~\citet{Hougaard:2000}, and~\citet{Spitoni.etal:2012} developed estimation procedures together with asymptotic theory for semi-Markov processes observed under various censoring mechanisms.

Although these methods accommodate incomplete observation arising from censoring, they do not account for reporting delays and incomplete event adjudication, both of which are characteristic features of disability insurance data. This limitation was recently addressed by~\citet{Buchardt.etal:2023b}, who proposed estimation procedures for multistate models observed subject to reporting delays and incomplete event adjudication. Their approach estimates both transition hazards, reporting delay distributions, and adjudication probabilities.

The reserving problem considered in the present paper is mathematically more demanding than the corresponding estimation problem, and therefore requires additional conditional independence assumptions that are not needed by~\citet{Buchardt.etal:2023b}. For estimation, it is sufficient to characterize the distribution of the observed multistate data, since this directly gives rise to a likelihood-based loss function. In contrast, reserving within the framework of~\citet{Buchardt.etal:2023a} requires characterizing the conditional distribution of the underlying valid time multistate process given the observed data.

Combined with the estimation methodology of~\citet{Buchardt.etal:2023b}, the reserve formulas developed in the present paper provide a unified framework for modeling, estimating, and reserving disability insurance liabilities when observations are subject to reporting and adjudication delays.

\section{Setup} \label{sec:setup}

\subsection{Disability insurance in valid time} \label{sec:ValidTimeSetup}
Let $(\Omega,\mathcal{F},(\mathcal{F}_t)_{t \geq 0},\mathbb{P})$ be a filtered background probability space. The biometric state of the insured is governed by a non-explosive pure jump process $Y : \Omega \times \mathbb{R}_+ \mapsto \mathcal{J}$ on a finite state space $\mathcal{J} = \{1,2,...,J\}$ for $J \in \mathbb{N}$ with deterministic initial state $y_0$. Denote by $N$ the corresponding multivariate counting process with components $N_{jk} : \Omega \times \mathbb{R}_+ \mapsto \mathbb{N}_0$ $(j,k\in \mathcal{J}, k\neq j)$ given by
\begin{equation*}
N_{jk}(t)=\# \{ s\in (0,t] : Y_{s-}=j, \ Y_s=k \}.
\end{equation*}
For $\mathcal{A} \subseteq \mathcal{J}$, let $\tau_\mathcal{A} : \Omega \rightarrow \overline{\mathbb{R}}_+$ be the first hitting time of $\mathcal{A}$ such that $\tau_\mathcal{A} = \inf\{t \geq 0 : Y_t \in \mathcal{A}\}$. The information generated by $Y$ is represented by the filtration $\mathcal{F}^Y_t = \sigma(Y_s , s \leq t)$. We shall also need the future information $\mathcal{F}^{t,Y} = \sigma(Y_s , s \geq t)$. Let $U : \Omega \times \mathbb{R}_+ \mapsto \mathbb{R}_+$ be the duration in the current state,
$$U_t = t-\sup \{ s \in (0,t] : Y_s \neq Y_t \}.$$

A life insurance contract between the insured and the insurer is stipulated by the specification of the accumulated cash flow $B : \Omega \times \mathbb{R}_+ \mapsto \mathbb{R}$ representing the accumulated benefits less premiums. We refer to this as the valid time cash flow or the contractual payments, and assume it is on the usual semi-Markov form,
\begin{equation*}
    B(\diff t) = \sum_{j=1}^J 1_{(Y_{t-}=j)} B_{j,t-U_{t-}}(\diff t) + \sum_{\substack{j,k=1 \\ j \neq k} }^J b_{jk}(t,U_{t-}) N_{jk}(\diff t), \quad B(0) \in \mathbb{R}, 
\end{equation*}
where $B_{j,w} : \mathbb{R}_+ \mapsto \mathbb{R}$ $(j \in \mathcal{J}, w \geq 0)$ are measurable, càdlàg and of finite variation, and $b_{jk}(t,u)$ $(j,k \in \mathcal{J}, j \neq k)$ are measurable and bounded. Since $t-U_{t-}$ is piecewise constant, the above expression is well-defined. We bundle all the processes that determine the payments into $X_t=(t,Y_t,U_t)$. Note that $\mathcal{F}^X_t = \mathcal{F}^Y_t$ since $U$ is constructed from the history of $Y$. Like~\citet{Buchardt.etal:2023a}, we name the model for $X$ and $B$ the valid time model. In this paper, we assume that the biometric state process $Y$ takes values in the state space $\mathcal{J}$ depicted in Figure~\ref{fig:Ystatespace}.

\begin{figure}[H]
	\centering
	\scalebox{0.75}{
	      \begin{tikzpicture}[node distance=4em, auto]
	\node[punkt] (g) {active};;
	\node[punktll, right=of g] (i1) {disabled $1$};
    \node[below=1em of i1] (i2) { $\vdots$ };
    \node[punktll, below=1.25em of i2] (in) { disabled $m$ };
    \node[punktl, right=of i1] (r) {reactivated};
	\node[punkt, below =3.5em of in] (dead) {dead}; 
    \node[below =0.5em of in] (dummysouth) {}; 
     \node[right=0.5em of i1] (dummyeast) {};
     \node[left=0.5em of i1] (dummywest) {}; 
     \node[anchor=north east, at=(g.north east)] {$a$};
     \node[anchor=north east, at=(i1.north east)] {$i_1$};
     \node[anchor=north east, at=(in.north east)] {$i_m$};
     \node[anchor=north east, at=(r.north east)] {$r$};
     \node[anchor=north east, at=(dead.north east)] {$d$};
     \node[right=0pt] at ($(i1.north east)+(-0.05,0.2)$) {$\mathcal{I}$};
	\path 
		($(g.south east)$) edge [pil, bend right = 40] ($(dead.west)$)
        ($(g.east)$) edge [pildotted, thick] ($(dummywest.west)$)
        ($(dummyeast.east)$) edge [pildotted, thick] ($(r.west)$)
        ($(r.south west)$) edge [pil, bend left = 40] ($(dead.east)$)
        ($(dummysouth.south)$) edge [pildotted, thick] ($(dead.north)$)
	; 
    \draw[thick, dotted] ($(i1.north west)+(-0.5,0.5)$) rectangle ($(in.south east)+(0.5,-0.5)$);
 
    \end{tikzpicture}
    }
    \caption{The state process $Y$ takes values in $\mathcal{J}=\{a,i_1,\dots,i_m,r,d\}$, being an illness-death model with $m$ disabled states $\mathcal{I}=\{i_1,\dots,i_m\}$ and a separate reactivated state. To reduce clutter, all transitions to and from $\mathcal{I}$ are illustrated as single dotted arrows. Transition between the disabled states is not possible.} 
    \label{fig:Ystatespace}
\end{figure}

\noindent Note that we have here labeled the states with letters instead of integers to stay consistent with the actuarial literature. This should cause no confusion in what follows. We assume $y_0 = a$ since only non-disabled are offered the insurance. We assume that $Y$ is a semi-Markov process with measurable transition hazards $\mu_{jk}: \mathbb{R}^2_+ \mapsto \mathbb{R}_+$ $(j,k \in \mathcal{J}, j \neq k)$ which are Lebesgue-integrable on compact subsets of $\mathbb{R}_+^2$ so that the intensity process for $N_{jk}$ is given by $\lambda_{jk}(t) = 1_{(Y_{t-} = j)}\mu_{jk}(t,U_{t-})$. That $Y$ is semi-Markov implies that $X$ is Markov. The assumption regarding the existence of transition hazards (as opposed to cumulative transition hazards) could be removed using the techniques of~\citet{Jacobsen:2006} or~\citet{Helwich:2008} and it would similarly not be difficult to allow for an uncountable number of disabled states e.g.\ $\mathcal{I}=(0,1]$ representing the degree of lost earning capacity. The choice and implications of the chosen state space are discussed in the supplementary material.

In order to formulate the transaction time model in the next section, it is convenient to introduce some marked point process notation. In general, all our processes are assumed to be constructed according to the canonical approach of~\citet{Jacobsen:2006}, which among other things implies a specific regular conditional distribution used in the conditional distributions and conditional expectations. We note that $X$ takes values in a Borel-space which we denote $(E,\mathcal{E})$. Write $\nabla$ for the irrelevant mark and $\overline{E} = E \cup \{\nabla\}$. Let 
\begin{align*}
    K = \{ ( (t_n)_{n \in \mathbb{N}} , (x_n)_{n \in \mathbb{N}}  ) \in \overline{\mathbb{R}}_+^{\mathbb{N}} \times \overline{E}^{\mathbb{N}} : & \: t_1 \leq t_2 \leq \dots \uparrow \infty, t_n < t_{n+1} \: \textnormal{if} \: t_n < \infty, \\
    & \: \textnormal{and} \: x_n \in E \: \textnormal{iff} \: t_n < \infty \}
\end{align*}
denote the space of sequences of jump times and jump marks and let this be equipped with the $\sigma$-algebra $\mathcal{K}$ generated by the coordinate projections 
$$T_n^\circ ((t_k)_{k \in \mathbb{N}},(x_k)_{k \in \mathbb{N}}) = t_n, \quad X_n^\circ ((t_k)_{k \in \mathbb{N}},(x_k)_{k \in \mathbb{N}}) = x_n$$
for $n \in \mathbb{N}$. Let the stochastic process $H : \Omega \times \mathbb{R}_+ \mapsto K$ be the marked point process history of $X$. The value $H_t$ consists of the ordered sequences of jump times $(\tau_{\{j\}} \times 1_{(\tau_{\{j\}} \leq t)})_{j \in \mathcal{J}}$ and corresponding jump marks $(X_{\tau_{\{j\}}} \times  1_{(\tau_{\{j\}} \leq t)})_{j \in \mathcal{J}}$ followed by a sequence of $\infty$ and $\nabla$ respectively. Note that this representation of the jump times and jump marks only holds when the model is hierarchical. Since $X$ is a piecewise deterministic process, there exists a measurable function $f : K \times \mathbb{R}_+ \mapsto E$ with the property that $X_t = f_{H_t}(t)$.

\begin{example} (Illustration of valid time notation.) \label{ex:ValidTimeNotation}\\
Consider an insured who is initially active, i.e.\ $Y_0 = a$. Suppose the insured becomes disabled with disability type $i_1$ at time $t_1$, recovers at time $t_2 > t_1$, and dies at time $t_3>t_2$. In this case, the valid time state process is
\[
Y_t=
\begin{cases}
a, & t<t_1,\\
i_1, & t_1\le t<t_2,\\
r, & t_2\le t<t_3,\\
d, & t\ge t_3.
\end{cases}
\]
The associated counting processes satisfy
\[
N_{ai_1}(t)=1(t\ge t_1), \qquad
N_{i_1r}(t)=1(t\ge t_2), \qquad
N_{rd}(t)=1(t\ge t_3),
\]
while all remaining components of $N$ are identically zero. The duration process is
\[
U_t=
\begin{cases}
t, & t<t_1,\\
t-t_1, & t_1\le t<t_2,\\
t-t_2, & t_2\le t<t_3,\\
t-t_3, & t\ge t_3.
\end{cases}
\]
Before $t_1$, the marked point process history $H_t$ contains no jumps, so 
$$H_t = ((\infty,\infty,\dots),(\nabla,\nabla,\dots)).$$
Immediately after $t_1$, it records the disability event at time $t_1$, meaning 
$$H_t = ((t_1,\infty,\infty\dots),(i_1,\nabla,\nabla,\dots)).$$
Immediately after $t_2$, it additionally contains the reactivation event at time
$t_2$, leading to 
$$H_t = ((t_1,t_2,\infty,\infty\dots),(i_1,r,\nabla,\nabla,\dots)).$$
Finally, after $t_3$ it also contains the death event, so
$$H_t = ((t_1,t_2,t_3,\infty,\infty\dots),(i_1,r,d,\nabla,\nabla,\dots)),$$
which is the final timeline of the insured in valid time. It is clear that the current value $X_t$ can be reconstructed from $H_t$ through a suitable mapping $f$.
\demoex
\end{example}

\subsection{Disability insurance in transaction time} \label{subsec:disabilityrealtime}

As was pointed out in~\citet{Buchardt.etal:2023a}, it may sometimes be unreasonable to assume that the insurer has observed $\mathcal{F}^X_t$ at time $t$, since there can be reporting and processing delays for disability claims. In this case, we also cannot assume that $B(t)$ has been paid out at time $t$. Consequently, we introduce a stochastic process $\mathcal{Z} : \Omega \times \mathbb{R}_+ \mapsto S$, where $(S,\mathcal{S})$ is a Borel-space, which generates the insurer's available information $\mathcal{F}_t^\mathcal{Z}$. We furthermore introduce the stochastic process $\mathcal{B} : \Omega \times \mathbb{R}_+ \mapsto \mathbb{R}$ modeling the accumulated observed payments which by construction will be $\mathcal{F}^\mathcal{Z}$-adapted, measurable, càdlàg and of finite variation. We refer to it as the transaction time cash flow. The model for $\mathcal{Z}$ and $\mathcal{B}$, which we now specify, is called the transaction time model.

\subsubsection*{Information}
As a first coordinate of $\mathcal{Z}$ we define the right-continuous pure jump process $Z^{(1)} : \Omega \times \mathbb{R}_+ \mapsto \mathcal{J}^{(1)}$ taking values in the state space $\mathcal{J}^{(1)}=\{1,2,3,4,5\}$ depicted in Figure~\ref{figure:Zstatespace}.
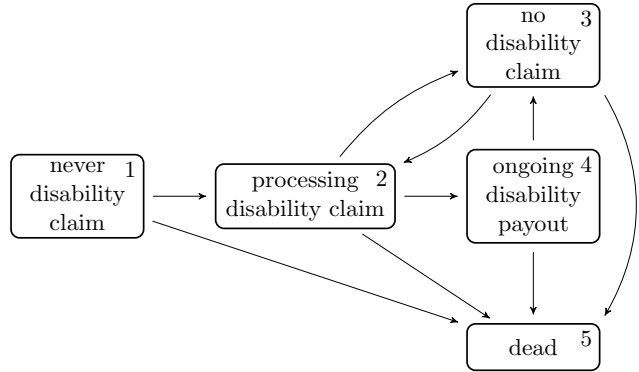
\begin{figure}[H] 
	\centering
	\scalebox{0.8}{
    \begin{tikzpicture}[node distance=2em and 0em]
    \node[punkt] at (-5.75, -2.5)    (0)   {never disability claim};
    \node[anchor=north east, at=(0.north east)]{$1$};  
    
    \node[punktl] at (-2, -2.5)    (1)   {processing disability claim}; \node[anchor=north east, at=(1.north east)]{$2$};    
    
    \node[punkt] at (1.75, 0)      (2)   {no disability claim};
    \node[anchor=north east, at=(2.north east)]{$3$};   
    
    \node[punkt] at (1.75, -2.5)    (3)   {ongoing disability payout};
    \node[anchor=north east, at=(3.north east)]{$4$};   
    
    \node[punkt] at (1.75, -5) (4)   {dead};
    \node[anchor=north east, at=(4.north east)]{$5$};   

    \path
        (0)     edge[pil]            node [above] {} (1)
        (1)     edge[pil, bend left = 15]            node [above] {} (2)
        (2)     edge[pil, bend left = 15]            node [above] {} (1)
        (1)     edge[pil]            node [above] {} (3)
        (3)     edge[pil]            node [above] {} (2)

        (0)     edge[pil]            node [above] {} (4)
        (1)     edge[pil]            node [above] {} (4)
        (2.south east)     edge[pil, bend left = 30]            node [above] {} (4.north east)
        (3)     edge[pil]            node [above] {} (4)
        ; 
\end{tikzpicture}}
\caption{State space $\mathcal{J}^{(1)}$ for the process $Z^{(1)}$.}
\label{figure:Zstatespace}
\end{figure}
\noindent The process $Z^{(1)}$ represents the state of the claim settlement and it holds that $Z_0^{(1)}=1$. We introduce another coordinate of $\mathcal{Z}$ denoted $Z^{(2)} : \Omega \times \mathbb{R}_+ \mapsto \mathbb{R}_+$ which represents the time of the disability event as reported by the insured in connection with a claim. We require $t \mapsto Z^{(2)}_t$ to be increasing and piecewise constant, and that its value can only increase upon a jump of $Z^{(1)}$ into state $2$. Furthermore, we require $Z^{(2)}_t \leq t$ and that $Z^{(2)}_t$ stays constant after a nonzero amount of disability benefits have been awarded. How disability benefits are awarded is formalized later in this section. The interpretation is that when a disability claim is reported, the insurer also reports at which past time the disability occurred. Furthermore, different disability claims are allowed, but only until one of the claims is awarded, and the claims must furthermore always be reported in the same order as their chronological ordering. We let $Z^{(2)}_0=0$ as a convention. We also introduce a coordinate $Z^{(3)} : \Omega \times \mathbb{R}_{+} \mapsto \mathcal{I}$ which represents
the disability type that is reported in connection with a claim, assume that it only changes
when $Z^{(2)}$ changes, and set $Z^{(3)}_0=i_1$ as a convention. Finally, denote the counting processes related to $Z^{(1)}$ by $N^{(1)}_{jk} : \Omega \times \mathbb{R}_+ \mapsto \mathbb{N}_0$ $(j,k \in \mathcal{J}^{(1)}, j \neq k)$ and denote by $T_{\{j\}}= \inf\{s \in [0,\infty) : Z^{(1)}_s = j\}$ the first hitting time of state $j$ by $Z^{(1)}$.

In addition to observing $Z=(Z^{(1)},Z^{(2)} ,Z^{(3)})$, the insurer observes what is being awarded to the insured; for example, whether a jump from state $2$ to state $3$ of $Z^{(1)}$ was accompanied by a payout of disability benefits in the form of \textit{backpay} or not. The term backpay refers to a payout of overdue payments that have been delayed by reporting and processing delays, and such payments appear in the transaction time cash flow $\mathcal{B}$ constructed later in this section. What is awarded to the insured is described through the bi-temporal stochastic process $H : \Omega \times \mathbb{R}^2_+ \mapsto K$, where $H^t_s$ is interpreted as the value of $H_s$ based on the information available at time $t$. We sometimes refer to $t$ as the observational time and $s$ as the historical time. We also introduce $X : \Omega \times \mathbb{R}^2_+ \mapsto E$ given by $X^t_s = (s,Y^t_s,U^t_s) = f_{H^t_s}(s)$, which is interpreted as the value of $X_s$ based on the available information at time $t$. Similarly, introduce the bi-temporal counting processes $N_{jk} : \Omega \times \mathbb{R}^2_+ \mapsto \mathbb{N}_0$ which we denote $N^t_{jk}(s)$ with analogous interpretation. In total, we let $\mathcal{Z}_t = (Z_t,H^t_t)$. 

To specify a model for $H^t_s$, we introduce an auxiliary stochastic process $G : \Omega \times \mathbb{R}_+ \mapsto \mathbb{R}_+$, where $G_t$ marks the beginning of the period where the insured would be eligible for additional disability when standing at time $t$, and which is given by
\begin{align*}
    \diff G_t &= \diff Z^{(2)}_{t} + 1_{(Z^{(1)}_t = 4)} \diff t +\sum_{k \in \{3,4,5\} } \delta^{2k}_t \diff N_{2k}^{(1)}(t), \hspace{0.5cm} G_0=0,
\end{align*}
for stochastic processes $\delta^{2k} : \Omega \times \mathbb{R}_+ \mapsto \mathbb{R}_+$ with $k \in \{3,4,5\}$ satisfying $\delta_t^{24} = t-G_{t-}$ and $\delta^{2k}_t \in [0,t-G_{t-})$ when $k \in \{3,5\}$. The fact that $Z^{(2)}_t$ was required to be constant after a non-zero amount of disability claims have been awarded corresponds to saying that it stays constant after time $t$ if $G_t > Z^{(2)}_t$. The interpretation of the specification of $G$ is that the insured is disabled if they have an ongoing disability payout (in other words:\ the insurance company is not able to retract disability benefits that have been paid out), which is captured by $G$ increasing with a rate of $1$ in state $4$ as well as $G$ jumping to the value $t$ if a jump from state $2$ to state $4$ occurs at time $t$. If the insured was eligible for additional disability benefits but is no longer disabled at the time of payout, this is captured by an increase in $G$ upon a jump from state $2$ to state $3$ or to state $5$. 

From $G$, we can create other processes of interest such as $W : \Omega \times \mathbb{R}_+ \mapsto \mathbb{R}_+$, being the number of time units the insured has been eligible for disability benefits using the current information, which is given by
$$W_t = G_t - Z^{(2)}_t.$$
We use $W$ instead of writing expressions in terms of $G$ whenever it eases interpretation. Some immediate properties are that $G_t$ and $W_t$ are increasing and $W_t \leq G_t \leq t$. 

 We now specify $H^t_s$. The bi-temporal process $H^t_s$ contains $(T_{\{5\}},d)$ on $(T_{\{5\}} \leq s, T_{\{5\}} \leq t)$, $(Z^{(2)}_t, Z^{(3)}_t)$ on $(Z^{(2)}_t \leq s, 0 < W_t)$, and $(G_t,r)$ on $(G_t \leq s, 0 < W_t, G_t < t, G_t < T_{\{5\}})$. Thus, events enter $H^t_s$ when the corresponding jump time exceeds $s$ and some $t$-related criterion is satisfied: Death has to have occurred before time $t$ for the death event to be included, the insured has to have been deemed eligible for disability benefits for a nonzero amount of time by time $t$ for the disability event to be included, and if it additionally holds that the payout of disability benefits has been stopped at time $t$ and this wasn't due to death then the reactivation event is also included. As detailed later in this section, the transaction time cash flow $\mathcal{B}$ is constructed such that it is always in accordance with $H^t_s$. Therefore, $H^t_s$ determines the payments and its specification is consequently essential to how the proposed transaction time model works.  
 
 To complete the modeling setup for the observations, we need to specify how $\mathcal{Z}$ is related to $X$. Note that $t \mapsto H^t_s$ is constant between jumps of $Z^{(1)}$, and assume $T_{\{5\}}$ is finite almost surely and that $H^t_s = H_s$ for $t \geq T_{\{5\}}$. In other words, $Z^{(1)}$ is the driver of new information arriving and the valid time process $X_s$ is the limit of the transaction time process $X^t_s$ when the observation time $t$ tends to infinity. It follows that the \textit{basic bi-temporal structure assumptions} introduced in~\citet{Buchardt.etal:2023a} are satisfied, and we may hence use their results. Note that with this specification of the transaction time model, once a disability claim has triggered some benefits, everything that happens afterward relates to this disability and no other disabilities can be reported. This is similar in spirit to how it was assumed that at most one disability could occur in the valid time model depicted in Figure~\ref{fig:Ystatespace}. 

\begin{example} (Illustration of transaction time notation.) \\
Consider the valid time trajectory from Example~\ref{ex:ValidTimeNotation}. Suppose
the disability occurring at time $t_1$ is reported at time
$t_4>t_1$, disability benefits are awarded at time $t_5>t_4$, the
disability payments are stopped at time $t_2 > t_5$, and the insured dies at time $t_3$. At time $t<t_4$, the disability has not been reported, and
\[
Z^{(1)}_t=1, \qquad
Z^{(2)}_t=0, \qquad
G_t=W_t=0, \qquad
H^t_t = ((\infty,\infty,\dots),(\nabla,\nabla,\dots)).
\]
Although the disability has already occurred for
$t_1\le t<t_4$, it is not yet reflected in the insurer's information. At time $t_4$, the disability claim is reported, and the claim settlement
process jumps into state $2$. Thus, for $t_4 \leq t<t_5$
\[
Z^{(1)}_{t}=2,\qquad
Z^{(2)}_{t}=t_1,\qquad
Z^{(3)}_{t}=i_1,
\]
while $G_{t}=W_{t}=0$ and $H^t_t = ((\infty,\infty,\dots),(\nabla,\nabla,\dots))$. The insurer now knows the reported
disability time and type, but has not yet awarded disability benefits. 

At time $t_5$, disability benefits are awarded, and the claim settlement
process jumps into state $4$ as the insured has not yet recovered at time $t_5$. Consequently, for $t_5 \leq t < t_2$,
\[
G_t=t,\qquad
W_t=t-t_1, \qquad
H^{t}_{t} = ((t_1,\infty,\infty\dots),(i_1,\nabla,\nabla,\dots)).
\]
Note that $W_{t_5}$ corresponds to the length of the disability period that is awarded as backpay at time $t_5$. At time $t_3$, the payments are stopped as the insured is judged to have recovered, and the claim settlement
process enters state $3$, signifying that there is no active application/reapplication for disability benefits. At this point, $G$ stops increasing, so for $t_2 \leq t < t_3$,
\[
G_t=t_3,\qquad
W_t=t_3-t_1, \qquad
H^{t}_{t} = ((t_1,t_3,\infty,\infty\dots),(i_1,r,\nabla,\nabla,\dots)).
\]
Finally, when the insured dies at time $t_3$, this is immediately reported to the insurer and the claim settlement process enters
state $5$. After this, the observed history coincides with the
valid time history, i.e.\ $H^t_s = H_s$ for all $s \geq 0$ when $t \geq t_3$.
\demoex
\end{example}

\begin{remark}(Granularity of the information.) \\
The above transaction time model presumes that the actuary has access to relatively granular information. It might be that the actuary does not receive information about reported claims, but is only notified about payouts. In that case, it is possible to work with alternative transaction time models such as the $Z^{(1)}$ model depicted in Figure~\ref{figure:AlternativeZstatespace}.
\begin{figure}[H] 
	\centering
	\scalebox{0.8}{
    \begin{tikzpicture}[node distance=2em and 0em]
    \node[punktm] at (0, -2.5)    (0)   {no event}; \node[anchor=north east, at=(0.north east)]{$2$};    
    \node[punktm] at (4, -2.5)    (1)   {event 1};   
    \node[anchor=north east, at=(1.north east)]{$3$};  
    \node[punktm] at (8, -2.5)    (2)   {event 2};   
    \node[anchor=north east, at=(2.north east)]{$4$};
    \node[] at (11, -2.5)    (3)   {...};   
    \node[anchor=north east, at=(3.north east)]{};
    \node[punktm, below = 15mm of 0] (4)   {dead};
    \node[anchor=north east, at=(4.north east)]{$1$};   

    \path
        (0)     edge[pil]            node [above] {} (1)
        (1)     edge[pil]            node [above] {} (2)
        (2)     edge[pil]            node [above] {} (3)
        (0)     edge[pil]            node [above] {} (4) 
        (1)     edge[pil]            node [above] {} (4) 
        (2)     edge[pil]            node [above] {} (4) 
        ; 
\end{tikzpicture}}
\caption{Coarser state space for the process $Z^{(1)}$.}
\label{figure:AlternativeZstatespace}
\end{figure}
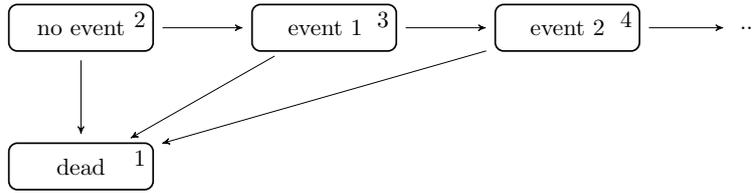
\noindent Here the different events could represent starting running payments, stopping running payments, or awarding backpay. It would be natural to let $Z^{(2)} : \Omega \times \mathbb{R}_+ \mapsto \{0,1\}$ be an indicator of whether there is running payments in the current state, and allow for a stochastic amount (including zero) of disability benefits to be awarded when jumping from one state to the next. The methods presented in this paper for the granular case can be adapted to handle this coarser case as well. Our methods do not apply if individual data is not available. \demormk
\end{remark}

\subsubsection*{Payments}
We now specify how the transaction time cash flow $\mathcal{B}$ is related to the valid time cash flow $B$. 
We first introduce the time value of money. A detailed treatment may be found in~\citet{Norberg:1990}. Let $\kappa : \mathbb{R}_+ \mapsto \mathbb{R}_+$ be some deterministic strictly positive càdlàg accumulation function with initial value $\kappa(0)=1$. A common choice is $\kappa(t)=\exp\left(\int_{(0,t]} r(v) \diff v \right)$ for some deterministic integrable function $r : \mathbb{R}_+ \mapsto \mathbb{R}$ called the force of interest. The corresponding discount function is $t \mapsto 1/\kappa(t)$. Introduce the auxiliary stochastic process $\mathcal{B}' : \Omega \times \mathbb{R}_+ \mapsto \mathbb{R}$ satisfying 
\begin{align*}
    \mathcal{B}'(t) =  \sum_{j=1}^J \int_{[0,t]} \frac{\kappa(t)}{\kappa(s)} 1_{(Y^t_{s-}=j)} B_{j,s-U^t_{s-}}(\diff s) +  \sum_{\substack{j,k=1 \\ j \neq k} }^J \int_{[0,t]} \frac{\kappa(t)}{\kappa(s)} b_{jk}(s,U^t_{s-}) N^t_{jk}(\diff s)
\end{align*}
which is the time $t$ value of the payments generated by $(X^t_s)_{s \leq t}$. We then specify the transaction time cash flow $\mathcal{B} : \Omega \times \mathbb{R}_+ \mapsto \mathbb{R}$ such that
\begin{align} \label{eq:cashflow}
    \int_{[0,t]} \frac{\kappa(t)}{\kappa(s)} \mathcal{B}(\diff s) = \mathcal{B}'(t).
\end{align} 
The payments $\mathcal{B}$ are constructed such that the accumulated payments in real-time are always congruent with the most recent marked point process history.  The way discounting is incorporated is closely connected to the notion of no arbitrage; the insurance company should not be able to keep the additional interest that they would earn by delaying the payout of benefits, so payments are accumulated with interest from the relevant historical time to the current observational time.

An explicit construction of a $\mathcal{B}$ satisfying Equation~\eqref{eq:cashflow} and further discussions are given in~\citet{Buchardt.etal:2023a}. That their definition of $\mathcal{B}$ satisfies Equation~\eqref{eq:cashflow} follows by using their Proposition 5.3 upon dividing with $\kappa(t)$ on both sides and taking the difference between evaluating in $0$ and $t$ to see
\begin{align*}
    \int_{(0,t]} \frac{1}{\kappa(s)} \mathcal{B}(\diff s) = \frac{1}{\kappa(t)}\mathcal{B}'(t)-\mathcal{B}(0)
\end{align*}
and then isolating $\mathcal{B}'(t)$. 

In the supplementary material, we extend the results of Section~\ref{sec:reserving} to stochastic interest rates that are independent of the valid and transaction time models and furthermore provide a simple prognostic tool for validating the non-financial parts of the model which may be of interest for practitioners. While these results are important for practical applications, they are relatively straightforward to derive and are somewhat orthogonal to the rest of the paper. Consequently, they have been deferred to the appendix.

\section{Reserving} \label{sec:reserving}

We now introduce the valid time and transaction time reserves, which are the focal point of the remainder of the paper. The present value in the valid time model $P : \Omega \times \mathbb{R}_+ \mapsto \mathbb{R}$ and the present value in the transaction time model $\mathcal{P} : \Omega \times \mathbb{R}_+ \mapsto \mathbb{R}$ are defined as 
\begin{equation*}
    P(t) = \int_{(t,\infty)} \frac{\kappa(t)}{\kappa(s)} B(\diff s), \quad \mathcal{P}(t) = \int_{(t,\infty)} \frac{\kappa(t)}{\kappa(s)} \mathcal{B}(\diff s),
\end{equation*}
respectively. We assume $P(t)$ and $\mathcal{P}(t)$ are integrable for any $t$. The corresponding valid time reserve $V : \Omega \times \mathbb{R}_+ \mapsto \mathbb{R}$ and transaction time reserve $\mathcal{V} : \Omega \times \mathbb{R}_+ \mapsto \mathbb{R}$ are then defined as 
\begin{equation*}
    V(t) = \mathbb{E}[P(t) \mid \mathcal{F}^X_t], \quad \mathcal{V}(t) = \mathbb{E}[\mathcal{P}(t) \mid \mathcal{F}_t^\mathcal{Z}].
\end{equation*}
Since $P(t)$ is $\mathcal{F}^{t,X}$-measurable and $X$ is Markov, we get $V(t) = \mathbb{E}[P(t) \mid X_t]$ almost surely. We also introduce the state-wise reserves $V_j : \mathbb{R}_+^2 \mapsto \mathbb{R}$ given by $V_j(t,u) = \mathbb{E}[P(t) \mid X_t = (t,j,u)]$ $(t \geq 0, j \in \mathcal{J}, u \geq 0)$, which are measurable functions satisfying $V_{Y_t}(t,U_t)=V(t)$ almost surely. The choice of $V_j(t,u)$ is not unique, but we follow the convention from the literature, which is to choose the version where the transition probabilities satisfy the Chapmann-Kolmogorov equations surely, confer with~\citet{Jacobsen:2006} p.\ 158-159 for the construction of such a version.  In most applications, the specific choice of the state-wise reserves will not matter, since they will always be evaluated in $X_t$. This is not the case in our application, so we make this choice explicit. The choice is also important for the path properties of $V(t)$, see e.g.~\citet{Christiansen:Furrer:2021}. Our choice of state-wise reserves agrees with that of~\citet{Christiansen:Furrer:2021}, confer with the proof of Theorem 5.10 in~\citet{Buchardt.etal:2023a}.

For $K \in \mathcal{A}$ we write 
\begin{align*}
    V_j(t,u ; K) = \mathbb{E}[P(t) \mid X_t = (t,j,u), K]
\end{align*}
for the reserve when we additionally condition on the event $K$. It is defined as $V_j(t,u ; K) = \mathbb{E}[1_{K} P(t) \mid X_t = (t,j,u)]/\mathbb{P}(K \mid X_t = (t,j,u))$ with the convention $0/0=0$. This equals $\mathbb{E}[P(t) \mid X_t = (t,j,u), 1_K]$ on the event $K$.

\subsubsection*{Categorization}
Using the setup introduced in Section~\ref{subsec:disabilityrealtime}, we can categorize the reserve at different points in time according to the usual claims reserving terminology known from e.g.~\citet{Norberg:1993}. 
\begin{itemize}
    \item On $(\textnormal{Settled}_t)=(Z^{(1)}_t = 5)$, the claim and reserve are classified as \textit{settled}, because even if some payments may remain in the dead state, the total claim size is known exactly. 
    \item On $(\textnormal{RBNS}_t)=(Z^{(1)}_t \in \{2,3,4\})$, the claim and reserve are \textit{reported-but-not-settled}.  
    \item On $(\textnormal{CBNR}_t) = (Z^{(1)}_t = 1)$, the claim and reserve are \textit{covered-but-not-reported}.
\end{itemize}
In the latter case, we will have an IBNR contribution for policies where the disability has already occurred, and a \textit{covered-but-not-incurred} (CBNI) contribution for policies where no disability has occurred yet. We therefore also introduce the events $(\textnormal{IBNR}_t) = (\textnormal{CBNR}_t, \tau_{\mathcal{I}} \leq t)$ and $(\textnormal{CBNI}_t) = (\textnormal{CBNR}_t, \tau_{\mathcal{I}} > t)$ which are however not known from the information available at time $t$. We also define $(\textnormal{RBNSr}_t) = (\textnormal{RBNS}_t,W_t>0)$ and $(\textnormal{RBNSi}_t) = (\textnormal{RBNS}_t,W_t=0)$, being RBNS where benefits have and have not been awarded respectively. In the former case, the time of disability is known and the time of reactivation is not fully determined, while in the latter case, both are not fully determined. The latter case is sometimes referred to as reported-but-not-paid in the literature, see e.g.~\citet{Bettonville.etal:2021}, and the former could consequently be called paid-but-not-settled, but these terms are not used in the current paper.

\subsubsection*{Independence}

For the reserves to be tractable, we need some restriction on the conditional distribution of $X$ given $\mathcal{F}^\mathcal{Z}_t$. Conditional independence criteria provide a natural way to impose such restrictions. We find it desirable to assume that the transaction time information only affects the distribution of future values $X$ by affecting the probability that a certain valid time outcome was the true realization, and thus provides no additional information if the true valid time outcome was known. This leads to tractable reserves and is often a reasonable assumption, and even if it is not, the violation can often be remedied by extending the valid time model, as discussed in the supplementary material. Remark~\ref{remark:Violation} discusses an informal approach to ascertaining whether the assumptions hold. For discussions on how to accommodate other violations of the setup, see the supplementary material. When the insured is dead, all is known, and so no independence assumption is needed. Formally, we thus assume:
\begin{assumption} (Influence of transaction time information.) \label{assumption:Independence} \\
On $(\textnormal{CBNR}_t)$
\begin{align*}
    \sigma((X_s)_{s \geq G_t}) \indep \mathcal{F}^\mathcal{Z}_t \mid 1_{(\tau_{\{d\}} \leq t)},  X_{\tau_{\mathcal{I}}} 1_{(\tau_{\mathcal{I}} \leq t)}.
\end{align*}
On $(\textnormal{RBNS}_t)$
\begin{align*}
    \sigma((X_s)_{s \geq G_t}) \indep \mathcal{F}^\mathcal{Z}_t \mid 1_{(\tau_{\{d\}} \leq t)},X_{G_t}. 
    \\[-7mm]\tag*{\demoas}
\end{align*}
\end{assumption}
\noindent Here $X_{G_t}$ is understood as the composite stochastic variable $\omega \mapsto (X_{G_t(\omega)})(\omega)$. Note that on both events, it holds that $\tau_{\{d\}} > t$ and we could thus have replaced $1_{(\tau_{\{d\}} \leq t)}$ by the event $(\tau_{\{d\}} > t)$ in the conditioning. This assumption states that the distribution of the valid time behavior of two subjects after time $G_t$ is exchangeable whenever the values of the variables in the conditioning agree for these subjects, no matter the rest of the transaction time information. The effect of this assumption is that the transaction time information $\mathcal{F}^\mathcal{Z}_t$ only affects the distribution of the variables entering in the conditioning and not the rest of the valid time process. 

As Assumption~\ref{assumption:Independence} stands, there is still some transaction time information remaining via $X_{G_t}$ in the second case; for example $(X_{G_t} = (G_t,i,0))$ for $i \in \mathcal{I}$ implies that the disability starting at time $G_t$ is not awarded at time $t$. What is needed to remove this final piece of transaction time information is a strong Markov-type property at the random time $G_t$. This situation is non-standard since the random time where the process is stopped is not a stopping time, see however~\citet{Yackel:1968} where a random time change with a non-stopping time is used to obtain a Markov process from a semi-Markov process. The situation is also similar to how left-truncated processes are usually studied conditional on some event having occurred prior, where the event is measurable with respect to an enlarged filtration stopped at the left-truncation time, but might not be measurable with respect to the self-exciting filtration, see for example Section III.3 of~\citet{Andersen.etal:1993}. To obtain the strong Markov property at $G_t$, we impose Assumption~\ref{assumption:IndependenceG}.
\begin{assumption} (Conditional independence of stopped valid time process.) \label{assumption:IndependenceG} \\
$\forall v,t \geq 0$:
$$\mathcal{F}^{v,X} \indep \mathcal{F}^X_v \vee \sigma(X_{G_t}) \mid 1_{(\tau_{\{d\}} \leq t)}, X_v$$
on $(G_t \leq v, \textnormal{RBNS}_t)$. \demoas
\end{assumption}
\noindent This is analogous to how the Markov property allows $(X_s)_{s \leq v}$ to be replaced by $X_v$ in the the conditional distribution of $(X_s)_{s \geq v}$. Here we, however, need to keep the knowledge that death has not occurred. Under this assumption, we get the following strong Markov property.
\begin{lemma} (Strong Markov property at $G_t$.) \label{lemma:StrongMarkov} \\
Under Assumption~\ref{assumption:IndependenceG},
    \begin{align*}
        \mathbb{P}((X_s)_{s \geq G_t} \in \hspace{-0.025cm} \cdot \mid (X_s)_{s \leq G_t}, 1_{(\tau_{\{d\}} \leq t)} ) =  \left.\mathbb{P}((X_s)_{s \geq x_1} \in \hspace{-0.025cm} \cdot \mid X_{x_1}=x, 1_{(\tau_{\{d\}} \leq t)})\right\vert_{x = X_{G_t}}
    \end{align*}
    on $(\textnormal{RBNS}_t)$ where $x=(x_1,x_2,x_3)$.
\end{lemma}
\noindent The proof of Lemma~\ref{lemma:StrongMarkov} is long and is hence deferred to the appendix. Lemma~\ref{lemma:StrongMarkov} states that there is no extra knowledge gained about the distribution of $X$ knowing that a path $(X_s)_{s \leq x_1}$ came from a transaction time realization with $x=X_{G_t}$ compared with just conditioning on $(X_s)_{s \leq x_1}$ as long as knowledge about survival until time $t$ is retained.

\subsection{CBNR reserve} \label{subsec:CBNR}

We first consider being on $(\textnormal{CBNR}_t)$ and calculate the transaction time reserve. This is only of interest if $\mathbb{P}(\textnormal{CBNR}_t)>0$, so we assume that this is the case. Introduce the \textit{IBNR-factor} $I : \mathbb{R}^2_+ \times \mathcal{I} \mapsto [0,1]$ defined as 
$$I_i(s,t) = \mathbb{P}(\textnormal{CBNR}_t \mid \tau_{\mathcal{I}} = s, Y_{\tau_{\mathcal{I}}}=i).$$ 
Note that this probability can also be expressed as the probability that the delay between the disability event and the first disability claim is larger than $t-s$. Introduce also the transition probabilities 
$$p_{jk}(s,t,u,z) = \mathbb{P}(Y_t=k, U_t \leq z \mid Y_s = j, U_s = u).$$  
The reserve for the CBNR case is given in Theorem~\ref{thm:CBNRReserve}.
\begin{theorem} (CBNR reserve.) \label{thm:CBNRReserve} \\
On $(\textnormal{CBNR}_t)$, we have
\begin{align*}
    \mathcal{V}(t) &= V_a(t,t) \times \mathbb{P}(\tau_{\mathcal{I}} > t \mid \textnormal{CBNR}_t) \\
    & \quad + \sum_{i \in \mathcal{I}} \int_{(0,t]} \bigg( \frac{\kappa(t)}{\kappa(s)} V_i(s,0;(\tau_{\{d\}} > t))+\frac{\kappa(t)}{\kappa(s)}b_{ai}(s,s)-\int_{(s,t]} \frac{\kappa(t)}{\kappa(v)} B_{a,0}(\diff v) \bigg) \\ & \qquad \qquad \; \: \times \frac{p_{aa}(0,s,0,\infty)}{\mathbb{P}(\textnormal{CBNR}_t)} I_i(s,t) \mu_{ai}(s,s) \diff s.
\end{align*}
Furthermore,
\begin{align*}
    \mathbb{P}(\tau_{\mathcal{I}} > t \mid \textnormal{CBNR}_t) &= 1-  \int_{(0,t] \times \mathcal{I}} \frac{I_i(s,t)}{\mathbb{P}(\textnormal{CBNR}_t)} \: (\tau_{\mathcal{I}},Y_{\tau_{\mathcal{I}}})(\mathbb{P})(\diff s, \diff i)
\end{align*}
and
\begin{align*}
    \mathbb{P}(\textnormal{CBNR}_t) = \frac{\int_{(t,\infty) \times \mathcal{I}} I_i(s,t) \: (\tau_{\mathcal{I}},Y_{\tau_{\mathcal{I}}})(\mathbb{P})(\diff s, \diff i) }{1-\mathbb{P}(\tau_{\mathcal{I}} = \infty \mid \tau_{\mathcal{I}} > t, \tau_{\{d\}} > t)}+\int_{(0,t] \times \mathcal{I}} I_i(s,t) \: (\tau_{\mathcal{I}},Y_{\tau_{\mathcal{I}}})(\mathbb{P})(\diff s, \diff i).
\end{align*}
\end{theorem}

\noindent Inserting $(\tau_{\mathcal{I}},Y_{\tau_{\mathcal{I}}})(\mathbb{P})(\diff s, \diff i) = p_{aa}(0,s,0,\infty)\mu_{ai}(s,s) \diff s$ for $s \in (0,\infty)$ and $i \in \mathcal{I}$ as well as substituting
\begin{align*}
    1-\mathbb{P}(\tau_{\mathcal{I}} = \infty \mid \tau_{\mathcal{I}} > t, \tau_{\{d\}} > t) &= \sum_{i \in \mathcal{I}}\int_{(t,\infty)} p_{aa}(t,s,t,\infty) \mu_{ai}(s,s) \diff s,
\end{align*}
and using Remark~\ref{remark:ConditionalSemiMarkov} to calculate $V_i(s,0;(\tau_{\{d\}} > t))$, we see that the reserve in Theorem~\ref{thm:CBNRReserve} is computable using only the usual valid time hazards and the new model element $I_i(s,t)$ for $s<\infty$. 

\begin{proof}
    Write
\begin{align*}
\begin{split}
    \mathcal{V}(t)   &= \mathbb{E}[ 1_{(\tau_{\mathcal{I}} > t)} \mathcal{P}(t)+1_{(\tau_{\mathcal{I}} \leq t)} \mathcal{P}(t) \mid \mathcal{F}_t^\mathcal{Z}] \\ 
    &= \mathbb{E}[\mathcal{P}(t) \mid \mathcal{F}_t^\mathcal{Z}, \tau_{\mathcal{I}} > t] \mathbb{P}(\tau_{\mathcal{I}} > t \mid \mathcal{F}_t^\mathcal{Z})  + \mathbb{E}[1_{(\tau_{\mathcal{I}} \leq t)}  \mathcal{P}(t) \mid \mathcal{F}_t^\mathcal{Z}].
\end{split}
\end{align*}
The first term is the CBNI reserve and the second term is the IBNR reserve.

We start by treating the CBNI reserve. Note that on $(\textnormal{CBNR}_t, \tau_{\mathcal{I}} > t)$, it holds that $\mathcal{P}(t)=P(t)$ by Theorem 5.4 of~\citet{Buchardt.etal:2023a}, since it then holds that $(X_s)_{0 \leq s \leq t} = (X_s^t)_{0 \leq s \leq t}$. This leads to 
\begin{align*}
1_{(\textnormal{CBNR}_t)} \mathbb{E}[\mathcal{P}(t) \mid \mathcal{F}_t^\mathcal{Z}, \tau_{\mathcal{I}} > t] = 1_{(\textnormal{CBNR}_t)} \mathbb{E}[P(t) \mid \tau_{\{d\}} > t, \tau_{\mathcal{I}} > t ] = 1_{(\textnormal{CBNR}_t)} V_a(t,t),
\end{align*}
by the first part of Assumption~\ref{assumption:Independence} and using that $P(t)$ is $\sigma((X_s)_{s \geq G_t})$-measurable since $G_t \leq t$.

For the IBNR reserve, we find by Theorem 5.4 of~\citet{Buchardt.etal:2023a} that on $(\textnormal{CBNR}_t, \tau_{\mathcal{I}}  \leq t)$ it holds 
\begin{align*}
    \mathcal{P}(t) &= P(t)+\int_{[0,t]} \frac{\kappa(t)}{\kappa(s)} (B-\mathcal{B})(\diff s) \\
    &= \frac{\kappa(t)}{\kappa(\tau_{\mathcal{I}}-)}P(\tau_{\mathcal{I}}-)-\int_{[\tau_{\mathcal{I}},t]} \frac{\kappa(t)}{\kappa(s)} B_{a,0}(\diff s) \\
    &= \frac{\kappa(t)}{\kappa(\tau_{\mathcal{I}})}P(\tau_{\mathcal{I}})+\frac{\kappa(t)}{\kappa(\tau_{\mathcal{I}})}b_{a Y_{\tau_{\mathcal{I}}}}(\tau_{\mathcal{I}},\tau_{\mathcal{I}})-\int_{(\tau_{\mathcal{I}},t]} \frac{\kappa(t)}{\kappa(s)} B_{a,0}(\diff s)
\end{align*}
since the payments are at least equal until $\tau_{\mathcal{I}}$ on this event.
Hence 
\begin{align*}
   &1_{(\textnormal{CBNR}_t)} \mathbb{E}[\mathcal{P}(t) 1_{(\tau_{\mathcal{I}} \leq t)} \mid \mathcal{F}_t^\mathcal{Z}] \\
   &=  1_{(\textnormal{CBNR}_t)} \mathbb{E}\bigg[  1_{(\tau_{\mathcal{I}} \leq t)} \bigg( \frac{\kappa(t)}{\kappa(\tau_{\mathcal{I}})}P(\tau_{\mathcal{I}})+\frac{\kappa(t)}{\kappa(\tau_{\mathcal{I}})}b_{a Y_{\tau_{\mathcal{I}}}}(\tau_{\mathcal{I}},\tau_{\mathcal{I}})-\int_{(\tau_{\mathcal{I}},t]} \frac{\kappa(t)}{\kappa(s)} B_{a,0}(\diff s) \bigg) \: \Big\vert \: \mathcal{F}_t^\mathcal{Z} \bigg] \\
   &=  1_{(\textnormal{CBNR}_t)} \mathbb{E}\bigg[  1_{(\tau_{\mathcal{I}} \leq t)} \bigg( \frac{\kappa(t)}{\kappa(\tau_{\mathcal{I}})} \mathbb{E}\big[P(\tau_{\mathcal{I}}) \mid X_{\tau_{\mathcal{I}}} \vee \mathcal{F}^\mathcal{Z}_t \big] \\
   & \qquad \qquad \qquad \qquad \qquad \; +\frac{\kappa(t)}{\kappa(\tau_{\mathcal{I}})}b_{a Y_{\tau_{\mathcal{I}}}}(\tau_{\mathcal{I}},\tau_{\mathcal{I}})-\int_{(\tau_{\mathcal{I}},t]} \frac{\kappa(t)}{\kappa(s)} B_{a,0}(\diff s) \bigg) \: \Big\vert \: \mathcal{F}_t^\mathcal{Z} \bigg]
\end{align*}
by the tower property. Now note that on $(\textnormal{CBNR}_t, \tau_{\mathcal{I}} \leq t)$, we have
\begin{align*}
    \mathbb{E}\big[P(\tau_{\mathcal{I}}) \mid X_{\tau_{\mathcal{I}}} \vee \mathcal{F}^\mathcal{Z}_t \big] &= \mathbb{E}\big[P(\tau_{\mathcal{I}}) \mid X_{\tau_{\mathcal{I}}}, \tau_{\{d\}}> t \big]\\
    &= \left.\frac{\mathbb{E}\big[P(s) 1_{(\tau_{\{d\}}> t)} \mid X_s = (s,i,0) \big]}{\mathbb{P}(\tau_{\{d\}}> t \mid X_s = (s,i,0) )}\right\vert_{s=\tau_{\mathcal{I}}, i=Y_{\tau_{\mathcal{I}}}}\\
    &= V_{Y_{\tau_{\mathcal{I}}}}(\tau_{\mathcal{I}},0 ; (\tau_{\{d\}}> t))
\end{align*}
by the first part of Assumption~\ref{assumption:Independence} and the usual strong Markov property, see Theorem 7.5.1 of~\citet{Jacobsen:2006}, using $(\tau_{\{d\}}> t) \in \sigma(X_t)$. Hence we obtain
\begin{align*}
   &\mathbb{E}[\mathcal{P}(t) 1_{(\tau_{\mathcal{I}} \leq t)} \mid \mathcal{F}_t^\mathcal{Z}] \\
   & = \int_{(0,t] \times \mathcal{I}} \bigg( \frac{\kappa(t)}{\kappa(s)} V_i(s,0;(\tau_{\{d\}} > t))+\frac{\kappa(t)}{\kappa(s)}b_{ai}(s,s)-\int_{(s,t]} \frac{\kappa(t)}{\kappa(v)} B_{a,0}(\diff v) \bigg) ( \tau_{\mathcal{I}}, Y_{\tau_{\mathcal{I}}} \mid \mathcal{F}^\mathcal{Z}_t)(\mathbb{P})(\diff s, \diff i)
\end{align*}
on $(\textnormal{CBNR}_t)$.
Note that on $(\textnormal{CBNR}_t)$, we have $(\tau_{\mathcal{I}},Y_{\tau_{\mathcal{I}}} \mid \mathcal{F}^\mathcal{Z}_t)(\mathbb{P})(\diff s,\diff i) = (\tau_{\mathcal{I}},Y_{\tau_{\mathcal{I}}} \mid \textnormal{CBNR}_t)(\mathbb{P})(\diff s, \diff i)$. 
Using Bayes' theorem, see for example Theorem 1.31 in~\citet{Schervish:1995}, we can write  
    \begin{align*}
    (\tau_{\mathcal{I}},Y_{\tau_{\mathcal{I}}} \mid \textnormal{CBNR}_t)(\mathbb{P})(\diff s, \diff i)  &= \mathbb{P}(\textnormal{CBNR}_t \mid \tau_{\mathcal{I}}=s,Y_{\tau_{\mathcal{I}}}=i) \frac{(\tau_{\mathcal{I}},Y_{\tau_{\mathcal{I}}})(\mathbb{P})(\diff s, \diff i)}{\mathbb{P}(\textnormal{CBNR}_t)}  \\
    &= I_i(s,t) \mu_{ai}(s,s) \frac{p_{aa}(0,s,0,\infty)}{\mathbb{P}(\textnormal{CBNR}_t)}  \diff s
    \end{align*}
for $s \in [0,\infty)$ and $i \in \mathcal{I}$. This also implies
\begin{align*}
    \mathbb{P}(\tau_{\mathcal{I}} > t \mid \textnormal{CBNR}_t) &= 1-  \sum_{i \in \mathcal{I}} \int_{(0,t]} I_i(s,t) \mu_{ai}(s,s) \frac{p_{aa}(0,s,0,\infty)}{\mathbb{P}(\textnormal{CBNR}_t)} \diff s. 
\end{align*}
For the final part, note
\begin{align*}
    \mathbb{P}(\textnormal{CBNR}_t) = \mathbb{E}[\mathbb{P}(\textnormal{CBNR}_t \mid \tau_{\mathcal{I}},Y_{\tau_{\mathcal{I}}}) ] = \int_{(0,\infty) \times \mathcal{I}} I_i(s,t) \: (\tau_{\mathcal{I}},Y_{\tau_{\mathcal{I}}})(\mathbb{P})(\diff s,\diff i) + \mathbb{P}(\textnormal{CBNR}_t, \tau_{\mathcal{I}} = \infty ).
\end{align*}
We have
\begin{align*}
    \mathbb{P}(\textnormal{CBNR}_t, \tau_{\mathcal{I}} = \infty ) &=   \mathbb{P}(\tau_{\mathcal{I}} = \infty \mid \textnormal{CBNR}_t) \mathbb{P}(\textnormal{CBNR}_t) \\
    &= \mathbb{P}(\tau_{\mathcal{I}} = \infty \mid \tau_{\mathcal{I}} > t, \tau_{\{d\}} > t) \mathbb{P}(\tau_{\mathcal{I}} > t \mid \textnormal{CBNR}_t) \mathbb{P}(\textnormal{CBNR}_t)
\end{align*}
using the first part of Assumption~\ref{assumption:Independence} in the second equality. Inserting these expressions, isolating for $\mathbb{P}(\textnormal{CBNR}_t)$, and simplifying gives
\begin{align*}
    \mathbb{P}(\textnormal{CBNR}_t) = \frac{\int_{(t,\infty) \times \mathcal{I}} I_i(s,t) \: (\tau_{\mathcal{I}},Y_{\tau_{\mathcal{I}}})(\mathbb{P})(\diff s, \diff i) }{1-\mathbb{P}(\tau_{\mathcal{I}} = \infty \mid \tau_{\mathcal{I}} > t, \tau_{\{d\}} > t)}+\int_{(0,t] \times \mathcal{I}} I_i(s,t) \: (\tau_{\mathcal{I}},Y_{\tau_{\mathcal{I}}})(\mathbb{P})(\diff s, \diff i).
\end{align*}
Collecting the results, we obtain the statement of the theorem.
\end{proof}

\begin{remark} (Relation to non-life insurance Poisson models.) \\
    For the IBNR term, the time $s$ disability rate $\mu_{ai}(s,s)$ has to be multiplied by the IBNR-factor $\mathbb{P}(\textnormal{CBNR}_t \mid \tau_{\mathcal{I}} = s, Y_{\tau_{\mathcal{I}}}=i)$ similarly to the Poisson process model in \citet{Norberg:1999}. Heuristically, an insurer has to hold a disability reserve for the expected number of disabilities $\mu_{ai}(s,s) \diff s$ at a prior time $s$ times the proportion of insured that have yet to report their claim by time $t$, which is $\mathbb{P}(\textnormal{CBNR}_t \mid \tau_{\mathcal{I}} = s, Y_{\tau_{\mathcal{I}}}=i)$. The extra factor $p_{aa}(0,s,0,\infty)/\mathbb{P}(\textnormal{CBNR}_t)$ adjusts for the fact that there can be at most one disability occurrence in this model as opposed to a Poisson process model where there can be several occurrences. \demormk
\end{remark}

\begin{remark} (Conditional semi-Markov model.) \label{remark:ConditionalSemiMarkov} \vspace{0.15cm} \\
In Section 6 of \citet{Hoem:1972}, the author obtains an expression for the transition probabilities and hazards of a semi-Markov process conditional on not having entered a specific absorbing part of the state space before a given time. For our purposes, choose the transient states to be $\mathcal{J} \backslash \{d\}$. Calculating the transition probabilities that appear in $V_i(\: \cdot,\cdot \: ;(\tau_{\{d\}}> t))$ can then be done as usual upon switching to the hazards $\mu_{\ell j}(s,u;t)$, where 
$$\mu_{\ell j}(s,u;t)=\mu_{\ell j}(s,u)\frac{ \sum_{k \in \mathcal{J} \backslash \{d\} } p_{jk}(s,t,0,\infty)  }{\sum_{k \in \mathcal{J} \backslash \{d\}} p_{\ell k}(s,t,u,\infty) }$$
for $s < t$ and 
$$\mu_{\ell j}(s,u;t) = \mu_{\ell j}(s,u)$$
for $s \geq t$. 
Consequently, the hazard of jumping to states where there is a higher probability of remaining in $\mathcal{J} \backslash \{d\}$ is increased and the hazard is decreased when there is a lower probability of remaining in $\mathcal{J} \backslash \{d\}$. After time $t$, the conditioning provides no additional information, and the hazards are equal to the hazards in the unconditional model. \demormk
\end{remark}

\subsection{RBNS reserve} \label{subsec:RBNS}
We now consider being on $(\textnormal{RBNS}_t)$. The reserve for the RBNSi case is given in Theorem~\ref{thm:RBNSiReserve}.
\begin{theorem} (RBNSi reserve.) \label{thm:RBNSiReserve} \\
We have
\begin{align*}
    \mathcal{V}(t) &=  \frac{\kappa(t)}{\kappa(G_t)} V_a(G_t,G_t ; (\tau_{\{d\}} > t)) \times (1-\mathbb{P}(X_{G_t} = (G_t,i,0) \mid \mathcal{F}_t^\mathcal{Z})) \\
    & \quad + \bigg( \frac{\kappa(t)}{\kappa(G_t)}V_{i}(G_t,0;(\tau_{\{d\}} > t)) + \frac{\kappa(t)}{\kappa(G_t)}b_{a i}(G_t,G_t) \bigg) \times \mathbb{P}(X_{G_t} = (G_t,i,0) \mid \mathcal{F}_t^\mathcal{Z}) \\
    & \quad- \int_{(G_t,t]} \frac{\kappa(t)}{\kappa(s)} B_{a,0}(\diff s)
\end{align*}
on $(\textnormal{RBNSi}_t)$ with $i=Z^{(3)}_t$.
\end{theorem}
\noindent Utilizing Remark~\ref{remark:ConditionalSemiMarkov} to calculate $V_i(G_t,0;(\tau_{\{d\}} > t))$, and analogously $V_a(G_t,G_t;(\tau_{\{d\}} > t))$, we see that the RBNSi reserve may be calculated using the usual valid time model and the new model element $\mathbb{P}(X_{G_t} = (G_t,Z^{(3)}_t,0) \mid \mathcal{F}_t^\mathcal{Z})$ which we name the \textit{adjudication probability}. This gives the probability that the reported disability claim will ultimately be awarded.

\begin{proof}
    By similar calculations to the IBNR-case, we find
\begin{align*}
    \mathcal{P}(t) &= \frac{\kappa(t)}{\kappa(G_t-)}P(G_t-)- \int_{[G_t,t]} \frac{\kappa(t)}{\kappa(s)} \mathcal{B}(\diff s).
\end{align*}
The latter term is $\mathcal{F}^\mathcal{Z}_t$-measurable and on $(\textnormal{RBNSi}_t)$ satisfies
$$ \int_{[G_t,t]} \frac{\kappa(t)}{\kappa(s)} \mathcal{B}(\diff s) = \int_{[G_t,t]} \frac{\kappa(t)}{\kappa(s)} B_{a,0}(\diff s)$$
so the interesting part is $\kappa(t)/\kappa(G_t-) \times P(G_t-)$. On $(\textnormal{RBNSi}_t)$ we have
$$\frac{\kappa(t)}{\kappa(G_t-)}P(G_t-) = \frac{\kappa(t)}{\kappa(G_t)}P(G_t)+1_{(X_{G_t} = (G_t,Z^{(3)}_t,0))}\frac{\kappa(t)}{\kappa(G_t)}b_{a Z^{(3)}_t}(G_t,G_t)+\frac{\kappa(t)}{\kappa(G_t)}B_{a,0}(\{G_t\}),$$
which implies
\begin{align*}
    \mathcal{V}(t) &= \mathbb{E}\left[ \frac{\kappa(t)}{\kappa(G_t)}P(G_t)+1_{(X_{G_t} = (G_t,Z^{(3)}_t,0))} \frac{\kappa(t)}{\kappa(G_t)}b_{a Z^{(3)}_t}(G_t,G_t) \mid \mathcal{F}^\mathcal{Z}_t \right] - \int_{(G_t,t]} \frac{\kappa(t)}{\kappa(s)} B_{a,0}(\diff s) \\
    &= \frac{\kappa(t)}{\kappa(G_t)}\mathbb{E}[P(G_t) \mid \mathcal{F}_t^\mathcal{Z}, X_{G_t} = (G_t,a,G_t)] \mathbb{P}(X_{G_t} = (G_t,a,G_t) \mid \mathcal{F}_t^\mathcal{Z}) \\
    & \quad + \frac{\kappa(t)}{\kappa(G_t)}\mathbb{E}[P(G_t) \mid \mathcal{F}_t^\mathcal{Z}, X_{G_t} = (G_t,Z^{(3)}_t,0)] \mathbb{P}(X_{G_t} = (G_t,Z^{(3)}_t,0) \mid \mathcal{F}_t^\mathcal{Z}) \\
    & \quad +\frac{\kappa(t)}{\kappa(G_t)}b_{a Z^{(3)}_t}(G_t,G_t)\mathbb{P}(X_{G_t} = (G_t,Z^{(3)}_t,0) \mid \mathcal{F}_t^\mathcal{Z}) - \int_{(G_t,t]} \frac{\kappa(t)}{\kappa(s)} B_{a,0}(\diff s)
\end{align*}
on $(\textnormal{RBNSi}_t)$. By the second part of Assumption~\ref{assumption:Independence} and Lemma~\ref{lemma:StrongMarkov}, we get on $(\textnormal{RBNSi}_t)$ that
\begin{align*}
    \mathbb{E}[P(G_t) \mid \mathcal{F}_t^\mathcal{Z}, X_{G_t}=(G_t,a,G_t)] &= \mathbb{E}[P(G_t) \mid   \tau_{\{d\}} > t, X_{G_t}=(G_t,a,G_t)] \\
    &= V_a(G_t,G_t ; (\tau_{\{d\}} > t))
\end{align*}
and 
\begin{align*}
    \mathbb{E}[P(G_t) \mid \mathcal{F}_t^\mathcal{Z}, X_{G_t}=(G_t,Z^{(3)}_t,0) ] &= \mathbb{E}[P(G_t) \mid \tau_{\{d\}} > t, X_{G_t}=(G_t,Z^{(3)}_t,0)] \\
    &= V_{Z^{(3)}_t}(G_t,0;(\tau_{\{d\}} > t))
\end{align*}
using that all of $(X_s)_{s \leq G_t}$ is known in the conditioning for both cases.
Collecting the results, we arrive at the desired expression.
\end{proof}

\noindent The reserve for the RBNSr case is given in Theorem~\ref{thm:RBNSrReserve}.
\begin{theorem} (RBNSr reserve.) \label{thm:RBNSrReserve} \\
We have
\begin{align*}
    \mathcal{V}(t) &=  V_r(t,t-G_t) \times (1-\mathbb{P}(X_{G_t} = (G_t,i,W_t) \mid \mathcal{F}^\mathcal{Z}_t)) \\
    & \quad +\bigg(\frac{\kappa(t)}{\kappa(G_t)} V_{i}(G_t,W_t;(\tau_{\{d\}} > t)) -1_{(Z^{(1)}_t \neq 4)}\Big(\frac{\kappa(t)}{\kappa(G_t)} b_{i r}(G_t,W_t) + \int_{(G_t,t]} \frac{\kappa(t)}{\kappa(s)} B_{r,G_t}(\diff s) \Big) \bigg) \\
    & \qquad \times \mathbb{P}(X_{G_t} = (G_t,i,W_t) \mid \mathcal{F}^\mathcal{Z}_t)
\end{align*}
on $(\textnormal{RBNSr}_t)$ with $i=Z^{(3)}_t$.
\end{theorem}
\noindent Similarly to Theorem~\ref{thm:RBNSiReserve}, we can use Remark~\ref{remark:ConditionalSemiMarkov} and the adjudication probability $\mathbb{P}(X_{G_t} = (G_t,Z^{(3)}_t,W_t) \mid \mathcal{F}_t^\mathcal{Z})$ to calculate the RBNSr transaction time reserve. Note also that if $Z^{(1)}_t = 4$ then $\mathbb{P}(X_{G_t} = (G_t,Z^{(3)}_t,W_t) \mid \mathcal{F}^\mathcal{Z}_t) = 1$ and furthermore $G_t=t$, $Z^{(3)}_t=Y_t$, and $W_t=U_t$. Thus, the expression collapses to $\mathcal{V}(t)=V_{Y_t}(t,U_t)$ which is the classic valid time disability reserve. 

\begin{proof}
As in the RBNSi-case, we find
\begin{align*}
    \mathcal{P}(t) &= \frac{\kappa(t)}{\kappa(G_t-)}P(G_t-)- \int_{[G_t,t]} \frac{\kappa(t)}{\kappa(s)} \mathcal{B}(\diff s).
\end{align*}
The latter term is $\mathcal{F}^\mathcal{Z}_t$-measurable and on $(\textnormal{RBNSr}_t)$ satisfies {\fontsize{10.5}{10.5} $$\int_{[G_t,t]} \frac{\kappa(t)}{\kappa(s)} \mathcal{B}(\diff s) =  \frac{\kappa(t)}{\kappa(G_t)} B_{Z^{(3)}_t,G_t-W_t}(\{G_t\})+1_{(Z^{(1)}_t \neq 4)} \Big( \frac{\kappa(t)}{\kappa(G_t)} b_{Z^{(3)}_t r}(G_t,W_t)+\int_{(G_t,t]} \frac{\kappa(t)}{\kappa(s)} B_{r,G_t}(\diff s)\Big)$$}

\noindent so the interesting part is $\kappa(t)/\kappa(G_t-) \times P(G_t-)$. Proceeding in a similar manner as before, we see on $(\textnormal{RBNSr}_t)$ that
$$\frac{\kappa(t)}{\kappa(G_t-)}P(G_t-) = \frac{\kappa(t)}{\kappa(G_t)}P(G_t)+1_{(X_{G_t}=(G_t,r,0))}\frac{\kappa(t)}{\kappa(G_t)}b_{Z^{(3)}_t r}(G_t,W_t)+\frac{\kappa(t)}{\kappa(G_t)}B_{Z^{(3)}_t,G_t-W_t}(\{G_t\}).$$
Now on $(\textnormal{RBNSr}_t)$ we have
\begin{align*}
    \mathcal{V}(t)  &= \frac{\kappa(t)}{\kappa(G_t)} \mathbb{E}\left[ P(G_t)\mid \mathcal{F}^\mathcal{Z}_t , X_{G_t} = (G_t,Z^{(3)}_t,W_t) \right]\mathbb{P}(X_{G_t} = (G_t,Z^{(3)}_t,W_t) \mid \mathcal{F}^\mathcal{Z}_t) \\
    & \quad +\frac{\kappa(t)}{\kappa(G_t)} \mathbb{E}\left[ P(G_t)\mid \mathcal{F}^\mathcal{Z}_t , X_{G_t}=(G_t,r,0) \right]\mathbb{P}(X_{G_t}=(G_t,r,0) \mid \mathcal{F}^\mathcal{Z}_t)\\
    & \quad -1_{(Z^{(1)}_t \neq 4)}\bigg(\frac{\kappa(t)}{\kappa(G_t)}b_{Z^{(3)}_t r}(G_t,W_t)\mathbb{P}(X_{G_t}=(G_t,Z^{(3)}_t,W_t) \mid \mathcal{F}^\mathcal{Z}_t) + \int_{(G_t,t]} \frac{\kappa(t)}{\kappa(s)} B_{r,G_t}(\diff s)\bigg).
\end{align*}
By the second part of Assumption~\ref{assumption:Independence} and Lemma~\ref{lemma:StrongMarkov}, we see on $(\textnormal{RBNSr}_t)$ that
\begin{align*}
    \mathbb{E}\left[ P(G_t)\mid \mathcal{F}^\mathcal{Z}_t , X_{G_t} = (G_t,Z^{(3)}_t,W_t) \right] &= \mathbb{E}\left[ P(G_t)\mid \tau_{\{d\}}>t, X_{G_t} = (G_t,Z^{(3)}_t,W_t) \right] \\
    &= V_{Z^{(3)}_t}(G_t,W_t;(\tau_{\{d\}} > t))
\end{align*}
 and
 \begin{align*}
     \mathbb{E}\left[ P(G_t)\mid \mathcal{F}^\mathcal{Z}_t , X_{G_t} = (G_t,r,0) \right] &= \mathbb{E}\left[ P(G_t)\mid \tau_{\{d\}}>t, X_{G_t} = (G_t,r,0) \right] \\
     &= V_r(G_t,0;(\tau_{\{d\}} > t)) \\
    &= \frac{\kappa(G_t)}{\kappa(t)}V_r(t,t-G_t) + \int_{(G_t,t]} \frac{\kappa(G_t)}{\kappa(s)} B_{r,G_t}(\diff s)
 \end{align*}
using that $(X_s)_{s \leq G_t}$ is known in the former case, and towering on $(X_s)_{s < G_t}$ in the latter case, observing that the expectation given $(X_s)_{s \leq G_t}$ only depends on $X_{G_t}$ and not $(X_s)_{s < G_t}$. Putting everything together, we arrive at the claimed result.    
\end{proof}

\noindent The implications of Assumptions~\ref{assumption:Independence} and~\ref{assumption:IndependenceG} are most naturally understood by examining the resulting CBNR and RBNS reserves; the transaction time information affects the probability that a given valid time event has occurred via the IBNR-factor and the adjudication probability, but it does not enter into the statewise valid time reserves. This implies that the reporting and adjudication delays have no prognostic value for predicting reactivation in valid time. 

\begin{remark} (Violation of the independence assumptions.) \label{remark:Violation} \\
We primarily care about violations of the assumptions if such violations lead to meaningful biases in the transaction time reserves. If the proposed transaction time reserves have been implemented, a natural approach would be to validate them (e.g., using run-off plots described in the supplementary material or the reserve dynamics derived in~\citet{Buchardt.etal:2023a}) within different strata. Relevant strata could be short and long reporting delays for the RBNSi and RBNSr reserves, or insured with and without a history of reapplications for the RBNSr reserves.

If the proposed model hasn't already been implemented, then an alternative is to look for empirical violations of the conditional independence assumptions using old data. The data should be sufficiently old so that (parts of) the valid time trajectories are known with high certainty. With such data, identify subjects with similar valid time information (entering in the conditioning of the independence assumptions) but differing transaction time information, and look at whether the empirical distribution of the future trajectory of $X$ differs for these groups. A more systematic approach would be to use standard conditional independence tests to formally test Assumption~\ref{assumption:Independence} and~\ref{assumption:IndependenceG} based on the old data. However, conditional independence testing is fundamentally a hard problem~\citep{Shah:Peters:2020}, so the simple empirical approaches may be preferable.

\demormk
\end{remark}

\section{Estimation} \label{sec:estimation}

To compute the transaction time reserves, we need to estimate the valid time transition hazards, the IBNR-factor, and the adjudication probabilities. The IBNR-factor and adjudication probabilities are new model elements, and we hence need to find a suitable way to estimate these. In addition, standard estimation procedures also do not apply for the valid time transition hazards since the data is contaminated by reporting delays and incomplete event adjudication. 

For simplicity, we limit the discussion to the situation where there is at most one reported disability claim in the sense that $Z^{(2)}$ can increase only once. In this case, the delay between the disability time $\tau_{\mathcal{I}}$ and the time of the first reported disability $T_{\{2\}}$ equals the reporting delay of the disability event, the latter being the difference between $\tau_{\mathcal{I}}$ and the last time where a disability is reported. This is sufficient for the data application in Section~\ref{sec:application} and makes the statistical problem a special case of the one studied in~\citet{Buchardt.etal:2023b}. The methods however easily generalize to the case where several distinct disabilities may be reported, in which case both the IBNR-factor and disability reporting delay distribution would need to be estimated in order to compute the IBNR reserve and estimate the disability hazard, respectively. 

The data structure in~\citet{Buchardt.etal:2023b} consists of events that are reported with a delay and which may be confirmed or annulled upon adjudication. Their proposed estimation algorithm is a two-step procedure, where the first step is to estimate the adjudication probabilities and the reporting delay distribution, and the second step uses these to estimate the valid time hazards while correcting for contamination. Due to the assumed model for $\mathcal{Z}$, the adjudication probability $\mathbb{P}(X_{G_t} = (G_t,Z^{(3)}_t,W_t) \mid \mathcal{F}_t^\mathcal{Z})$ can be calculated as an absorption probability for a suitable multistate model as in~\citet{Buchardt.etal:2023b}. We henceforth refer to the transition hazards in the adjudication multistate model as adjudication hazards. Furthermore, estimating the IBNR-factor is equivalent to estimating the disability reporting delay distribution since
\begin{align*}
    I_i(s,t) = \mathbb{P}(\textnormal{CBNR}_t \mid \tau_{\mathcal{I}} = s, Y_{\tau_{\mathcal{I}}} = i) = \mathbb{P}( T_{\{2\}}-\tau_{\mathcal{I}} > t-s \mid \tau_{\mathcal{I}} = s, Y_{\tau_{\mathcal{I}}} = i).
\end{align*}
Imposing a parametric model for all three model elements hence results in an estimation problem that can be handled using~\citet{Buchardt.etal:2023b}. More details are given in the supplementary material.

The estimator of the valid time hazards described in the supplementary material and employed in the data application of Section~\ref{sec:application} corresponds to the \textit{Poisson approximation} from~\citet{Buchardt.etal:2023b} rather than the full estimator described in Section 3.3 of that paper.  The approximation has the advantage of being simpler to implement and allows for the estimation of the different valid time hazards separately. Because the data employed in Section~\ref{sec:application} only contains deaths recorded during the adjudication period, it is not possible to run the full estimation procedure of~\citet{Buchardt.etal:2023b}. The use of the Poisson approximation rather than the full estimator is thus due to data limitations rather than a methodological choice. Mortality rates are nevertheless needed in order to calculate the reserves in the data application, so in Section~\ref{sec:application} we employ the hazards given in Table~\ref{table:deathHazards}, which are inspired by those employed by a large Danish life insurance company, and published by the Danish Financial Supervisory Authority (FSA).

\begin{remark} (Restricting the information used for adjudications.) \label{remark:AdjHazardInfo} \\
     Note that~\citet{Buchardt.etal:2023b} allows us to reduce the information that the adjudication hazards depend on e.g.\ such that the reactivation adjudication hazards are not conditional on the disability reporting delay. In this paper, the adjudication probabilities are however defined conditional on the full transaction time filtration $\mathcal{F}^\mathcal{Z}$, so to compute these, the adjudication hazards must also depend on all of $\mathcal{F}^\mathcal{Z}$. It is of course still possible to impose structural assumptions on the adjudication hazards such that they only depend on parts of $\mathcal{F}^\mathcal{Z}$. \demormk
\end{remark}

\begin{remark}(Estimation and independence assumptions.) \\
    The independence imposed via Assumption~\ref{assumption:Independence} and~\ref{assumption:IndependenceG} is not used in the estimation procedure of~\citet{Buchardt.etal:2023b}, and we hence still obtains consistent and asymptotically normal estimators if these assumptions do not hold. The reserves themselves are of course still incorrectly specified if the assumptions are violated, even if the individual model elements are estimated consistently. We conjecture that it is possible to derive estimators that are more efficient than the ones suggested here by exploiting these independence assumptions. \demormk
\end{remark}

\FloatBarrier

\section{Illustrative case study} \label{sec:application}

\noindent  The objective of this section is to illustrate the practical implementation of the proposed reserving framework rather than to provide a comprehensive empirical validation. Because the current proof-of-concept implementation is slow, the reserve calculations are carried out for a random sample of insured individuals. Moreover, the available data contain information from a single valuation date, which precludes validation against subsequent claim developments. Consequently, the numerical results should be interpreted as an illustration of the methodology and of the magnitude of the differences between reserving approaches using the estimated parameters, rather than as evidence that the proposed method outperforms competing approaches.

\subsection{Estimation}

We calculate transaction time reserves at time $\eta=\textnormal{01/09/2019}$ for a subset of the LEC-DK19 (Loss of Earning Capacity -- Denmark 2019) data set, which was introduced in~\citet{Buchardt.etal:2023b}. The data includes information on disability exposure and occurrences, reactivation exposure and occurrences, reporting delays for disabilities, and adjudication exposure and occurrences related to both disabilities and reactivations. The data window $[0,\eta]$ is $[\textnormal{31/01/2015},\textnormal{01/09/2019}]$. Available covariates are gender and age. 

We note that the biometric data conforms with the valid time state space from Figure~\ref{fig:Ystatespace} with a single disabled state $i_1$ which we henceforth refer to as $i$ for notational convenience. Furthermore, the adjudication data conforms with the adjudication multistate model described in the supplementary material. In the following we estimate the relevant model elements using the method proposed in Section~\ref{sec:estimation} and subsequently calculate the reserves using the results from Section~\ref{sec:reserving}. The implementation is written in \texttt{R}~\citep{R:2023} and is available on GitHub (\url{https://github.com/oliversandqvist/Web-appendix-disability-reserving}).

For estimation of the hazards, we let all the covariates enter in a linear predictor with log link and assume that all hazards are variationally independent such that there is no overlap in parameters between different hazards. The disability hazard is regressed on age, gender, and calendar time, while the reactivation hazard is regressed on the same covariates but also the duration as disabled. As noted in the supplementary material, the data does not permit a reasonable estimate of the death hazards, and we thus simply employ the death hazards from Table~\ref{table:deathHazards} when calculating reserves. 

The adjudication hazards for an $\textnormal{RBNSi}$ claim are regressed on age, gender, duration since the disability event, duration since the disability event was reported, and whether or not the claim has been (temporarily) rejected previously. The adjudication hazards for an $\textnormal{RBNSr}$ claim are regressed on age, gender, duration since the disability event, and duration since the reactivation event. As noted in Remark~\ref{remark:AdjHazardInfo}, these hazards should depend on all of $\mathcal{F}^\mathcal{Z}$, so the fact that we for example do not regress the reactivation adjudication hazards on the disability reporting delay should be understood as the implicit assumption that the true value of that regressor is known to be zero. 

For the disability reporting delay, we impose a Weibull proportional reverse time hazard distribution which has the distribution function $t \mapsto (1-\exp(-(\lambda t)^k))^{\exp(W^T \beta)}$ for covariates $W$ and parameters $(\lambda,k,\beta)$. For covariates, we use age at disability onset and gender. The data does not contain observations of reactivation reporting delays, but we take this to be an artifact of the data rather than a violation of Assumption~\ref{assumption:IndependenceG} and hence proceed as if this had not been the case, recalling that the reactivation reporting delay distribution is only needed for estimation and not for reserving. 

With these specifications, we note that the estimation procedure becomes identical to the one in Section 6 of~\citet{Buchardt.etal:2023b} and we may hence use their estimates; confer with Section 6 and Section G of the supplementary material in~\citet{Buchardt.etal:2023b} for the parameter values. We however keep the calendar time effect fixed at its value at $\eta$ so as to not overextrapolate the observed calendar time trend when calculating the prospective reserves.

\begin{table}[ht!]
\centering
{%
\begin{tabular}{lcccc}
 &  Age  & Male & Female & $\min\{\textnormal{Duration},5\}$ \\[3pt]
$\mu_{ad}$, $\mu_{rd}$ & 0.09  & -9.50 & -9.80 & -     \\
$\mu_{id}$ & 0.09 & -6.40 & -6.80 & -0.25
\end{tabular}}
\caption{Death hazards based on estimates published by the Danish FSA.}
\label{table:deathHazards}
\end{table}

\subsection{Reserving}

\noindent For reserving, we sample 100 random insured at time $\eta$ and compute reserves for each of the categories CBNR, RBNSi, and RBNSr. The data contains around 250,000 insured that are in the portfolio at time $\eta$, but it would take a long time to compute the reserve for everyone using our proof-of-concept \texttt{R} implementation. The terms that take the longest to compute are valid time active reserves and IBNR reserves which for a single insured take a couple of seconds to evaluate on a regular laptop. Insurance companies have access to optimized calculation kernels and greater computing power with some presently relying on semi-Markov models for their reserves, so this is not a limitation of our proposed approach but rather of this specific implementation. Furthermore, it is possible to speed up computations by finding suitable approximations of, for example, the IBNR reserve given in Theorem~\ref{thm:CBNRReserve}. For CBNR and RBNSr, we sample the 100 insured without replacement, but since there are only 59 insured in the RBNSi category at time $\eta$, we sample these with replacement. 

The transaction time reserves are compared with a naive approach where $X^\eta_\eta$ is plugged into the valid time reserves. The naive approach thus leads to reserves that are sufficient to cover disabilities that have an ongoing payout at time $\eta$ and disabilities that occur after time $\eta$ but ignores IBNR-claims, claims that are under adjudication, and possible reapplications. 

The transaction time reserves are further compared with a simple approach where ad-hoc adjustments of the valid time reserves are made to adjust for IBNR and incomplete adjudication. The simple approach takes the CBNR reserve to be a valid time active reserve where the coverage period is extended with the average disability reporting delay. The average delay $d$ is found to be around 0.53 years when the average is based on disabilities that occurred at least two years before time $\eta$ to limit the effect of right-truncation. The heuristic is that the insurance company should cover disabilities that arrive up to $d$ years after the end of the coverage period since these disabilities occurred within the coverage period if the reporting delay was deterministically equal to $d$. The RBNSi reserve is chosen to be a valid time disability reserve with duration $0$ since most claims get awarded. For the RBNSr reserve, a valid time disability reserve is used if there are ongoing disability payments and a valid time active reserve is used otherwise to accommodate those who will apply for additional benefits in the future. The benchmark methods are intentionally simple and are included to provide interpretable reference points for the proposed methodology rather than to represent state-of-the-art reserving models.

The same hazards are used in all of the approaches so the differences between the results only reflect the reserving methodologies and not the estimation. Since the data does not contain information about benefit type or size we set $B(\diff t) = 1_{(Y_{t-}=i)} 1_{(t-U_{t-} \leq \eta+3)} 1_{(a + t \leq 67)} \diff t$, where $a$ is the age at time $0$, corresponding to a unit disability annuity until retirement at age $67$ with a coverage period of 3 years. Note that this specification of the cash flow also implies that the insured are covered for disabilities occurring before time $\eta$. We finally assume a constant force of interest $r \equiv 0.02$. The reserves are calculated by plugging in the estimated model elements into Theorem~\ref{thm:CBNRReserve},~\ref{thm:RBNSiReserve}, and~\ref{thm:RBNSrReserve}. The state-wise valid time reserves entering into these expressions are calculated by solving Thiele's differential equation iteratively over the states by exploiting the hierarchical structure of $\mathcal{J}$.

\subsection{Results}

The results are given in Table~\ref{table:reserves}. The purpose of these comparisons is to illustrate the practical implications of the proposed reserve formulas relative to two simple benchmark approaches under a common fitted model, rather than to establish empirical superiority. The settled category is not depicted since the reserve is zero in all cases, recalling that the contractual payments correspond to a disability annuity and settlement corresponds to death. The largest relative difference arises for the RBNSi reserve. The naive method under-reserves since it ignores the reported disability and reserves as if the insured was active because no disability benefits have been awarded yet. Surprisingly, the proposed method also leads to larger reserves than the simple method which always reserves the valid time disability reserve. This happens because the adjudication probabilities are close to one (they have an average of around 0.9) and there are many older insured where the effect of conditioning on not having died and reserving from time $G_t$ instead of $t$ leads to moderately larger reserves. 

The second largest relative difference is seen for the CBNR reserve. Here, the probability weighting of the active reserve is observed to be very close to one, so the main difference between the proposed and naive method is the IBNR contribution. Thus, the naive method under-reserves as it neglects this term. The simple method much closer to the proposed method as expected, but there is still a sizeable relative difference. 

The smallest relative difference for the non-settled cases is seen for the RBNSr reserve. This is because, in the data set, there are considerably more insured receiving running benefits than are reactivated. In fact, only four reactivated insured were sampled and none of them had applied for additional benefits. As the population mix shifts toward a higher proportion of reactivated subjects, the difference between the proposed and naive method is likely to grow.

\begin{table}[ht!]
\centering
{%
\begin{tabular}{lccc}
 &  $(\textnormal{CBNR}_\eta)$  & $(\textnormal{RBNSi}_\eta)$ & $(\textnormal{RBNSr}_\eta)$   \\[3pt]
Proposed method & 9.29 & 502.44 & 961.68   \\
Simple method & 8.98  & 481.20 & 945.66   \\
Naive method & 7.73  & 6.23 & 945.35   \\
Simple method difference &  0.32 (3.53\%) &  21.23 (4.41\%)  & 16.02 (1.69\%)  \\
Naive method difference &  1.56 (20.19\%)  &  496.21 (7961.48\%)  & 16.33 (1.73\%)  
\end{tabular}}
\caption{Reserves for 100 randomly sampled insured from each of the claims settlement categories except the settled category where the reserve is identically zero.}
\label{table:reserves}
\end{table}

To illustrate the potential portfolio-level implications, we approximate the portfolio reserve by assuming that the average reserve among the 100 sampled insured is representative of the corresponding claims-settlement category in the full portfolio. This approximation is intended solely as an illustrative scaling exercise and should not be interpreted as an estimate of the true portfolio reserve. Multiplying the average reserve in a given category by the total number of insured in that category, we obtain the results depicted in Figure~\ref{fig:Barchart}. Comparing with Table~\ref{table:reserves}, we see that despite the RBNSi category having the largest absolute and relative differences for the 100 sampled insured, the difference on the portfolio level is comparable with that of the RBNSr category, which had the smallest relative difference for the 100 sampled insured. This is because there are considerably more insured in the RBNSr category at time $\eta$. Similarly, the CBNR category, which showed the smallest absolute difference in Table~\ref{table:reserves}, leads to the largest difference on the portfolio level due to this category being by far the largest. In total, the naive and simple method leads to portfolio reserves that are around 11.1\% and 2.7\% smaller than the proposed method, respectively. 

As an illustration of the potential magnitude of the differences implied by the proposed methodology, consider that in 2019, the annual reports of two large Danish insurers showed portfolio reserves for health and disability insurance obligations amounting to 9,351 and 17,606 million DKK respectively. Therefore, a 2.7\% increase in these reserves would equate to approximately 250 million DKK for the former and 470 million DKK for the latter.
\begin{figure}[H]
    \centering
    \includegraphics[scale=0.75]{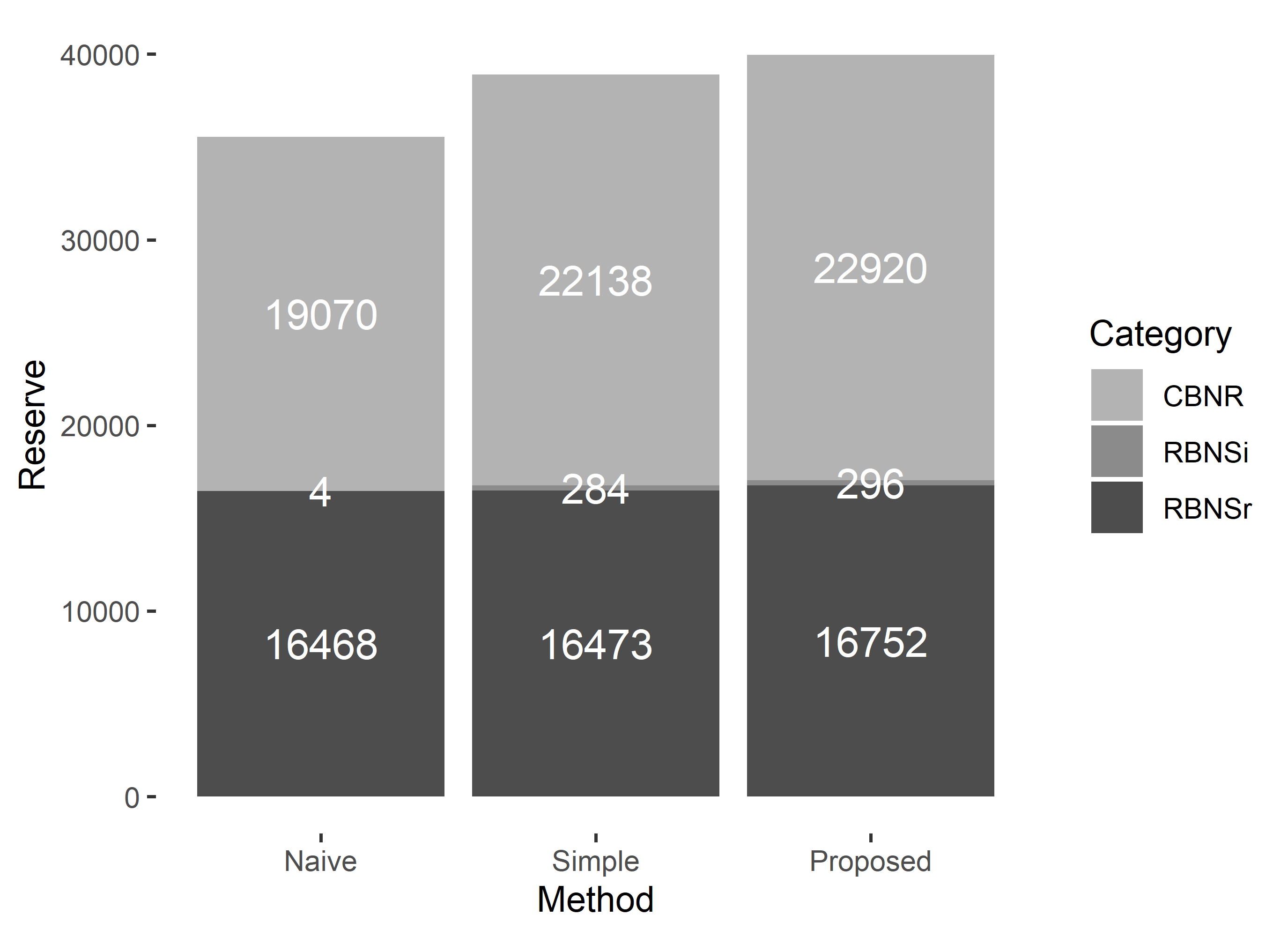}
    \caption{Approximated portfolio-level reserves (of unit disability annuities) decomposed by claims-settlement category (CBNR, RBNSi, and RBNSr) across the Naive, Simple, and Proposed reserving methods. The values at portfolio level are approximated by scaling up the average reserve of the 100 randomly sampled individuals in each category by the total number of policyholders in that category across the full portfolio.}
    \label{fig:Barchart}
\end{figure}

\noindent It would be highly relevant to compare the different reserves with observed claim developments to see which best describe the data. This is however not possible with the current data since the disability and reactivation occurrences and exposures are available for a single valuation date only, and their values at different valuation dates also cannot be inferred from the adjudication data since the same id does not refer to the same insured across the individual data tables. This is a major limitation of the present data application and is left as a topic for future research.

\section{Conclusion} \label{sec:conclusion}

This paper develops an individual reserving model for disability insurance in the presence of information delays caused by reporting delays and adjudication processes. We have introduced suitable conditional independence assumptions that lead to tractable and interpretable reserves, which may be calculated using the usual valid time hazards, the IBNR-factor, and the adjudication probabilities. The reserves are tailored to the features of disability insurance schemes by accommodating reporting delays and adjudication processing while preserving the advantages offered by valid time multistate models, namely that contractual payments are an a priori known function of a multistate process whose intertemporal distribution is well-understood. It is shown how the estimation procedure from~\citet{Buchardt.etal:2023b} may be used to estimate the model constituents of the reserves. Finally, an application using real data illustrates that the proposed reserve formulas can be implemented with estimated model components. A comprehensive empirical validation requires longitudinal data containing claim developments over multiple valuation dates, but this is left for future research.

\section*{Acknowledgments and declarations of interest}

This research has partly been funded by the Innovation Fund Denmark (IFD) under File No.\ 1044-00144B. The author declares no conflicts of interest. I would like to thank Kristian Buchardt and my supervisor Christian Furrer for many fruitful discussions. I would also like to thank them both for their helpful comments on an earlier version of the manuscript. 

\newpage

\begin{appendices}

\section{Proof of Lemma~\ref{lemma:StrongMarkov}} \label{sec:StrongMarkovProof}

To state and prove Proposition~\ref{prop:strongMarkov}, which immediately implies Lemma~\ref{lemma:StrongMarkov}, we first introduce some marked point process notation. Let $(\mathcal{M},\mathcal{H},(\mathcal{H}_t)_{t \geq 0})$ be the canonical space of counting measures, see Section 2 and 4.2 in~\citet{Jacobsen:2006}. Let $\mu \in \mathcal{M}$ be the underlying random counting measure for $X$. Write $(T_n,X_n)$ for the jump times and marks of $\mu$. For a given random counting measure $m$, we let $\tau_n(m)$ and $\eta_n(m)$ be the $n$'th jump time and jump mark respectively implying $T_n=\tau_n(\mu)$ and $X_n = \eta_n(\mu)$. For any measurable random time $R$, define the truncated measure $\theta_R \mu : (R < \infty) \rightarrow \mathcal{M}$ by
$$\theta_R \mu (\omega) = \sum_{\substack{ n : R(\omega) < T_n(\omega) < \infty }} \varepsilon_{(T_n(\omega),X_n(\omega))} $$
where $\varepsilon_{(t,x)}$ is the Dirac measure in $(t,x)$. We write $T_{R,n} = \tau_n(\theta_R \mu)$ and $X_{R,n} = \eta_n(\theta_R \mu)$ for the jump times and jump marks determining $\theta_R \mu$. Thus, if $R(\omega) < \infty$,
$$\theta_R \mu (\omega) = \sum_{\substack{ n : T_{R,n}(\omega) < \infty }} \varepsilon_{(T_{R,n}(\omega),X_{R,n}(\omega))} $$
with $R(\omega) < T_{R,1}(\omega) \leq T_{R,2}(\omega)  \leq \dots$.  Also let $T_{R,0}=R$ and $X_{R,0}=X_R$. We denote by $Q^x$ the distribution of $\theta_{x_1} \mu$ given $X_{x_1} = x$ constructed as the time-inhomogeneous case in~\citet{Jacobsen:2006} p.\ 157-158. Note that this actually corresponds to $Q^{x_1,x}$ in the notation of~\citet{Jacobsen:2006}. Even though $X$ is time-homogeneous, this is usually not the case when we condition on $1_{(\tau_{\{d\}} \leq t)}$ which is why we employ the time-inhomogeneous construction. The function governing the behavior of $X$ between jumps is denoted $\phi_{v s}$, meaning that if $v \leq s$ and no jumps occurred in $(v,s]$, we have $X_s = \phi_{vs}(X_v)$. The jump time and jump mark Markov-kernels for $Q^{x}$ are denoted $F_{t_n,y_n}(v)$ and $r_{t}(\phi_{t_n t}(y_n),C)$ respectively. The interpretation is that $F_{t_n,y_n}(v)$ gives the probability that the next jump has occurred by time $v$ given that the previous jump happened at time $t_n$ with mark $y_n$, while $r_{t}(y,C)$ gives the probability that an event occurring at time $t$ from state $y$ ends up in the set $C$.  Note $F_{\infty,\nabla}(v) = 0$ for any $v \in [0,\infty)$ and 
$r_{\infty}(y,C) = 1_{(\nabla \in C)}$.

We let $t$ be fixed but arbitrary. Let $\tilde{Q}^x$ be constructed as $Q^x$ but according to the modified Markov kernels where we additionally condition on $1_{(\tau_{\{d\}} \leq t)}$, hence now being stochastic. We could, of course, have removed this additional stochasticity by replacing the indicator with the event $(\tau_{\{d\}} > t)$ in the conditioning on the relevant event $(\textnormal{RBNS}_t)$. We however stick with this construction since it more easily generalizes to other cases where we want to keep additional information that is not deterministically known on the relevant event. It holds that $\tilde{Q}^x(H) = \mathbb{P}(\theta_{x_1} \mu \in H \mid X_{x_1}=x, 1_{(\tau_{\{d\}} \leq t)})$ almost surely for any $H \in \mathcal{H}$. Note that for every possible outcome of the conditioning information, the transition kernels of $\tilde{Q}^x$ stay on the Markov form:\ either the conditioning is superfluous because it relates to an event that occurred before the previous jump time, or it is future-measurable and we can hence use the formula for conditional probabilities in the jump time and jump mark kernels, use the Markov property, and then use the formula for conditional probabilities in reverse. This shows that the transition kernels indeed only depend on $(t_n,y_n)$ and not on $t_1,...,t_{n-1}$ and $y_1,...,y_{n-1}$. They will be denoted $\Tilde{r}_{t}(y,C)$ and $\tilde{F}_{t_n,y_n}(v)$. For notational convenience, we write $\mathds{1}(t) = 1_{(\textnormal{RBNS}_t)}$. We now state and prove Proposition~\ref{prop:strongMarkov}, which immediately implies the statement in Lemma~\ref{lemma:StrongMarkov}.

\begin{proposition} (Strong Markov type property at $G_t$.) \label{prop:strongMarkov} \\
Under Assumption~\ref{assumption:IndependenceG}, we have
\begin{align} \label{eq:strongMarkovResult}
    \mathds{1}(t)\mathbb{P}(\theta_{G_t} \mu \in \cdot \mid (X_s)_{s \leq G_t}, 1_{(\tau_{\{d\}} \leq t)} ) = \mathds{1}(t)\Tilde{Q}^{X_{G_t}}(\cdot)
\end{align}
$\mathbb{P}$-a.s.
\end{proposition}
\noindent This is an almost sure equality between probability measures on $(\mathcal{M},\mathcal{H})$. We note that the proof does not use many properties of our model for $X$ and $G_t$, and similar arguments may thus be used to show strong Markov properties for other 
Markov processes and random times provided that an independence assumption similar to Assumption~\ref{assumption:IndependenceG} is imposed.

\begin{proof}
The proof is inspired by the proof of Theorem 7.5.1 in~\citet{Jacobsen:2006} with the necessary changes to adjust for the fact that Assumption~\ref{assumption:IndependenceG} is different than the usual Markov independence assumption. 

As in~\citet{Jacobsen:2006}, we note that showing Equation~\eqref{eq:strongMarkovResult} is equivalent to showing for $n \geq 1$ and all measurable and bounded function $f_i : (0,\infty] \times \overline{E}  \rightarrow \mathbb{R}$ that the following holds:
\begin{align} \label{eq:InductionDesired}
   \mathds{1}(t)\mathbb{E} \left[ \prod_{i=1}^n f_i(T_{G_t,i},X_{G_t,i})  \mid (X_s)_{s \leq G_t}, 1_{(\tau_{\{d\}} \leq t)} \right] = \mathds{1}(t)\Tilde{E}^{X_{G_t}}\left[ \prod_{i=1}^n f_i(\tau_i,\eta_i) \right]
\end{align}
which we will prove by induction on $n$. The proof consists of four steps:
\begin{enumerate}[label=(\roman*)]
    \item For a discretization of $G_t$ to $G_t(M)$ with $M \in \mathbb{N}$, show the result (\ref{eq:strongMarkovResult}) for $X$ stopped at $T_{G_t(M),n-1}$ for any $n \geq 1$:
    \begin{align*}
        &1_{(T_{G_t(M),n-1} < \infty)}  \mathds{1}(t)\mathbb{P}(\theta_{T_{ G_t(M),n-1}}\mu \in \cdot \mid (X_s)_{s \leq T_{G_t(M),n-1} }, X_{G_t},1_{(\tau_{\{d\}} \leq t)})  \\
        & = 1_{(T_{G_t(M),n-1} < \infty)}  \mathds{1}(t)\Tilde{Q}^{X_{T_{G_t(M),n-1}}}(\cdot).
    \end{align*}
    \item Show convergence for $M \rightarrow \infty$:
    \begin{align*}
       & \lim_{M \rightarrow \infty} 1_{(T_{G_t(M),n-1} < \infty)} \mathds{1}(t) \Tilde{E}^{X_{T_{G_t(M),n-1}}}[f_n(\tau_1,\eta_1)] \\
       & = 1_{(T_{G_t,n-1} < \infty)} \mathds{1}(t) \Tilde{E}^{X_{T_{G_t,n-1}}}[f_n(\tau_1,\eta_1)].
    \end{align*}
    \item Discretize $G_t$ to $G_t(M)$ and use dominated convergence with (i) and (ii) to conclude: 
    $$ \mathds{1}(t)\mathbb{E}[f_n(T_{G_t,n},X_{G_t,n}) \mid (X_s)_{s \leq T_{G_t,n-1}},X_{G_t}, 1_{(\tau_{\{d\}} \leq t)}] = \mathds{1}(t)\Tilde{E}^{X_{T_{G_t,n-1}}}[f_n(\tau_1,\eta_1)].$$
    \item Use (iii) and induction over $n$ to finish the proof.
\end{enumerate}

\bigskip \noindent (i) : Let $H \in \mathcal{H}$ be given. Define
\begin{align*}
G_t(M) &= \sum_{m=1}^\infty t_{Mm} 1_{(t_{M(m-1)}\leq G_t < t_{Mm})}
\end{align*}
for $t_{Mm} = m2^{-M}$. Had we not known that $G_t < \infty$ almost surely, we would have to add $\infty 1_{(G_t = \infty)}$ here and show the result on $(G_t < \infty)$. 

Take an $F$ on the form 
$$F=( (X_s)_{s \leq T_{G_t(M),n-1}} \in B ) \cap (X_{G_t} \in C) \cap (1_{(\tau_{\{d\}} \leq t)} \in D )$$
and note that the collection of such sets constitutes an intersection-stable generator for 
$$\sigma((X_s)_{s \leq T_{G_t(M),n-1}}) \vee \sigma(X_{G_t}) \vee \sigma(1_{(\tau_{\{d\}} \leq t)})$$ 
containing $\Omega$. Write $F_{M,m} = F \cap ( G_t(M) = t_{Mm} )$ and note that 
\begin{align*}
    F_{M,m} \in \sigma((X_s)_{s \leq T_{t_{Mm},n-1}} ) \vee \sigma(X_{G_t}) \vee \sigma(1_{(\tau_{\{d\}} \leq t)}).
\end{align*}
Here we used that $G_t(M)$ is $G_t$-measurable and that $X_{G_t}$ contains $G_t$ as a coordinate. Now we get
\begin{align*}
    &\int_{F_{M,m}}  1_{(T_{G_t(M),n-1} < \infty)} \mathds{1}(t) \Tilde{Q}^{X_{T_{G_t(M),n-1}}}(H) \: d \mathbb{P} \\
    &= \int_{F_{M,m}} 1_{(T_{t_{Mm},n-1} < \infty)} \mathds{1}(t) \Tilde{Q}^{X_{T_{t_{Mm},n-1}}}(H) \: d \mathbb{P} \\
    &= \int_{F_{M,m}} 1_{(T_{t_{Mm},n-1} < \infty)} \mathds{1}(t) \mathbb{P}( \theta_{T_{t_{Mm},n-1}} \mu \in H \mid (X_s)_{s \leq T_{t_{Mm},n-1}}, X_{G_t}, 1_{(\tau_{\{d\}} \leq t)} )  \: d \mathbb{P} \\
     & =\int_{F_{M,m}} 1_{(T_{t_{Mm},n-1} < \infty)}  \mathds{1}(t) 1_{(\theta_{T_{t_{Mm},n-1}} \mu \in H)}  \: d \mathbb{P} \\     &= \mathbb{P} \left( F_{M,m} \cap (T_{t_{Mm},n-1} < \infty) \cap (\textnormal{RBNS}_t) \cap (\theta_{T_{t_{Mm},n-1}} \mu \in H)  \right) \\
     &= \mathbb{P} \left( F_{M,m} \cap (T_{G_t(M),n-1} < \infty) \cap (\textnormal{RBNS}_t) \cap (\theta_{T_{G_t(M),n-1}} \mu \in H)  \right),
\end{align*}
where the second equality is Lemma~\ref{lemmaStoppingtime} with $T=T_{t_{Mm},n-1}$. The third equality uses $F_{M,m} \cap (T_{t_{Mm},n-1}<\infty) \cap (\textnormal{RBNS}_t) \in \sigma((X_s)_{s \leq T_{t_{Mm},n-1}}) \vee \sigma(X_{G_t}) \vee \sigma(1_{(\tau_{\{d\}} \leq t)})$; recall in particular that $(\textnormal{RBNS}_t)=(T_{\{2\}} \leq t, T_{\{5\}} > t)$ and $(T_{\{2\}} \leq t)$ is known from $\sigma(X_{G_t})$ since $(T_{\{2\}} \leq t) = (G_t > 0)$. Summing over $m \geq 1$ and using $(G_t < \infty)$ almost surely gives
\begin{align*}
    &\int_F  1_{(T_{G_t(M),n-1} < \infty)} \mathds{1}(t) \Tilde{Q}^{X_{T_{G_t(M),n-1}}}(H) \: d \mathbb{P} \\
    &= \mathbb{P} \left( F \cap (T_{G_t(M),n-1} < \infty) \cap  (\textnormal{RBNS}_t) \cap (\theta_{T_{G_t(M),n-1}} \mu \in H)  \right).
\end{align*}
As the left- and right-hand side are finite measures, uniqueness of finite measures on intersection stable classes containing $\Omega$ gives that the equation holds for any $F \in \sigma( (X_s)_{s \leq T_{G_t(M),n-1}}) \vee \sigma(X_{G_t}) \vee \sigma(1_{(\tau_{\{d\}} \leq t)})$. As $1_{(T_{G_t(M),n-1} < \infty)} \mathds{1}(t) \Tilde{Q}^{X_{T_{G_t(M),n-1}}}(H)$ is also $\sigma((X_s)_{s \leq T_{G_t(M),n-1}}) \vee \sigma(X_{G_t}) \vee \sigma(1_{(\tau_{\{d\}} \leq t)})$-measurable, it satisfies the defining properties of
\begin{align*}
    &\mathbb{E}[ 1_{(T_{G_t(M),n-1} < \infty)} 1_{(\textnormal{RBNS}_t)} 1_{(\theta_{T_{G_t(M),n-1}} \mu \in H)} \mid (X_s)_{s \leq T_{G_t(M),n-1}}, X_{G_t}, 1_{(\tau_{\{d\}} \leq t)}],
\end{align*}
so we conclude 
\begin{align*}
 &1_{(T_{G_t(M),n-1} < \infty)} \mathds{1}(t) \mathbb{P}(\theta_{T_{G_t(M),n-1}} \mu \in H \mid (X_s)_{s \leq T_{G_t(M),n-1}}, X_{G_t},1_{(\tau_{\{d\}} \leq t)}) \\
 &= \mathbb{E}[ 1_{(T_{G_t(M),n-1} < \infty)} 1_{(\textnormal{RBNS}_t)} 1_{(\theta_{T_{G_t(M),n-1}} \mu \in H)} \mid (X_s)_{s \leq T_{G_t(M),n-1}}, X_{G_t}, 1_{(\tau_{\{d\}} \leq t)}] \\
 &= 1_{(T_{(G_t)(M),n-1} < \infty)} \mathds{1}(t) \Tilde{Q}^{X_{T_{(G_t)(M),n-1}}}(H).
\end{align*}

\bigskip \noindent (ii) : Note that we can write 
\begin{align*}
    &1_{(T_{G_t,n-1} < \infty)} \Tilde{E}^{X_{T_{G_t,n-1}}}[f_n(\tau_1,\eta_1)] \\
    &= 1_{(T_{G_t,n-1} < \infty)} \int_{( T_{G_t,n-1} , \infty ]} \int_{\overline{E}} f_n(v,y) \: \Tilde{r}_v(\phi_{T_{G_t,n-1} v}(X_{T_{G_t,n-1}}), dy) \Tilde{F}_{X_{T_{G_t},n-1}}(dv)
\end{align*}
with a similar expression for $1_{(T_{G_t(M),n-1} < \infty)} \Tilde{E}^{X_{T_{G_t(M),n-1}}}[f_n(\tau_1,\eta_1)]$.

On $(\overline{N}_{T_{G_t,n-1}} = \overline{N}_{T_{G_t(M),n-1}})$ and $(T_{G_t(M),n-1} < \infty)$, it holds for $v \geq T_{G_t(M),n-1}$, since $T_{G_t(M),n-1} \geq T_{G_t,n-1}$, that
\begin{align*}
    \phi_{T_{G_t,n-1} v}(X_{T_{G_t,n-1}}) &= \phi_{T_{G_t(M),n-1} v}(\phi_{T_{G_t,n-1} T_{G_t(M),n-1} }(X_{T_{G_t,n-1}})) \\
    &= \phi_{T_{G_t(M),n-1} v}(X_{T_{G_t(M),n-1}})
\end{align*}
and
\begin{align*}
    \overline{\Tilde{F}}_{X_{T_{G_t,n-1}}}(v) &= \overline{\Tilde{F}}_{X_{T_{G_t,n-1}}}( T_{G_t(M),n-1} ) \overline{\Tilde{F}}_{X_{T_{G_t(M),n-1}}}( v ).
\end{align*}
Therefore, it holds on $(\overline{N}_{T_{G_t,n-1}} = \overline{N}_{T_{G_t(M),n-1}})$ that
\begin{align*}
   &1_{(T_{G_t(M),n-1} < \infty)} \Tilde{E}^{X_{T_{G_t(M),n-1}}}[f_n(\tau_1,\eta_1)] \\
   &= 1_{(T_{G_t(M),n-1} < \infty)} \times \\
   & \qquad \int_{ (T_{(G_t(M),n-1},\infty] } \int_{\overline{E}} f_n(v,y) \: \Tilde{r}_v( \phi_{ T_{G_t(M),n-1} v } ( X_{ T_{G_t(M),n-1}  }  ), dy   )  \Tilde{F}_{T_{G_t(M),n-1} }(dv) \\
   &= 1_{(T_{G_t(M),n-1} < \infty)} \times \frac{1}{\overline{\Tilde{F}}_{X_{T_{G_t,n-1} }} (T_{G_t(M),n-1} ) } \times \\
   & \qquad  \int_{ (T_{G_t(M),n-1},\infty] } \int_{\overline{E}} f_n(v,y) \: \Tilde{r}_v( \phi_{ T_{G_t,n-1} v } ( X_{ T_{G_t,n-1}  }  ), dy   )  \Tilde{F}_{T_{G_t,n-1} }(dv) \\
   & \stackrel{M \rightarrow \infty}{\rightarrow}  1_{(T_{G_t,n-1} < \infty)} \Tilde{E}^{X_{T_{G_t,n-1}}}[f_n(\tau_1,\eta_1)].
\end{align*}
Note that
$$1_{\left(\overline{N}_{T_{G_t,n-1}} = \overline{N}_{T_{G_t(M),n-1}}\right)} \stackrel{M \rightarrow \infty}{\rightarrow} 1,$$
since $\overline{N}$ is right-continuous and $T_{G_t(M), n-1} \downarrow T_{G_t,n-1}$. We thus get
\begin{align*}
    &\lim_{M \rightarrow \infty} 1_{(T_{G_t(M),n-1} < \infty)} \Tilde{E}^{X_{T_{G_t(M),n-1}}}[f_n(\tau_1,\eta_1)] \\
    &= \lim_{M \rightarrow \infty} 1_{\left(\overline{N}_{T_{G_t,n-1}} = \overline{N}_{T_{G_t(M),n-1}}\right)} 1_{(T_{G_t(M),n-1} < \infty)} \Tilde{E}^{X_{T_{G_t(M),n-1}}}[f_n(\tau_1,\eta_1)]  \\
    &= \lim_{M \rightarrow \infty} 1_{\left(\overline{N}_{T_{G_t,n-1}} = \overline{N}_{T_{G_t(M),n-1}}\right)} 1_{(T_{G_t,n-1} < \infty)} \Tilde{E}^{X_{T_{G_t,n-1}}}[f_n(\tau_1,\eta_1)] \\
    &= 1_{(T_{G_t,n-1} < \infty)} \Tilde{E}^{X_{T_{G_t,n-1}}}[f_n(\tau_1,\eta_1)].
\end{align*}
We are now in a position to show the result using dominated convergence.

\bigskip\noindent (iii) :  Take $F$ on the form 
$$F=((X_s)_{s \leq T_{G_t,n-1}} \in B) \cap (X_{G_t} \in C) \cap ( 1_{(\tau_{\{d\}} \leq t)} \in D).$$ 
Such sets constitute an intersection stable generator of $\sigma((X_s)_{s \leq T_{G_t,n-1}}) \vee \sigma(X_{G_t}) \vee \sigma(1_{(\tau_{\{d\}} \leq t)})$ containing $\Omega$. Now let 
$$F_M = ( (X_s)_{ s \leq T_{G_t(M),n-1}} \in B) \cap (X_{G_t} \in C) \cap ( 1_{(\tau_{\{d\}} \leq t)} \in D)$$
as then 
$$F_M \in \sigma((X_s)_{s \leq T_{G_t(M),n-1}}) \vee \sigma(X_{G_t}) \vee \sigma(1_{(\tau_{\{d\}} \leq t)})$$ 
and $1_{F_M} \stackrel{M \rightarrow \infty}{\rightarrow} 1_F$ by using that $X$ is right-continuous, that $T_{G_t(M),n-1} \geq T_{G_t,n-1}$, and that $\lim_{M \rightarrow \infty} T_{G_t(M),n-1} = T_{G_t,n-1}$. We now see using dominated convergence and the results of (i) and (ii):
\begin{align*}
    &\int_F 1_{(T_{G_t,n-1} < \infty)} \mathds{1}(t) \Tilde{E}^{X_{T_{G_t,n-1}}}[f_n(\tau_1,\eta_1)] d\mathbb{P} \\
    &\stackrel{\textnormal{(ii)}}{=} \lim_{M \rightarrow \infty} \int_{F_M} 1_{(T_{G_t(M),n-1} < \infty)} \mathds{1}(t) \Tilde{E}^{X_{T_{G_t(M),n-1}}}[f_n(\tau_1,\eta_1)] d\mathbb{P} \\
    &\stackrel{\textnormal{(i)}}{=} \lim_{M \rightarrow \infty} \int_{F_M} 1_{(T_{G_t(M),n-1} < \infty)} \mathds{1}(t) \\
    & \qquad \qquad \qquad \times \mathbb{E}[f_n(T_{G_t(M),n},X_{G_t(M),n}) \mid (X_s)_{s \leq T_{G_t(M),n-1}}, X_{G_t}, 1_{( \tau_{\{d\}} \leq t)}] d\mathbb{P} \\
    &= \lim_{M \rightarrow \infty} \int_{F_M} 1_{(T_{G_t(M),n-1} < \infty)} \mathds{1}(t) f_n(T_{G_t(M),n},X_{G_t(M),n})  d\mathbb{P} \\
    &= \int_{F} 1_{(T_{G_t,n-1} < \infty)} \mathds{1}(t) f_n(T_{G_t,n},X_{G_t,n})  d\mathbb{P}.
\end{align*}
In the last equality, we used that $(T_{G_t(M),n},X_{G_t(M),n})(\omega) = (T_{G_t,n},X_{G_t,n})(\omega) $ for $M=M(\omega)$ sufficiently large.
 Now since the left-hand side and right-hand side are finite measures, when seen as a function of $F$, that are equal on an intersection stable generator including $\Omega$, we can conclude that they are equal on all of $\sigma((X_s)_{s \leq T_{G_t,n-1}}) \vee \sigma(X_{G_t}) \vee \sigma(1_{( \tau_{\{d\}} \leq t)})$. Since it also holds that $1_{(T_{G_t,n-1} < \infty)} \mathds{1}(t) \Tilde{E}^{X_{T_{G_t,n-1}}}[f_n(\tau_1,\eta_1)]$ is $\sigma((X_s)_{s \leq T_{G_t,n-1}}) \vee \sigma(X_{G_t}) \vee \sigma(1_{( \tau_{\{d\}} \leq t)})$-measurable, we conclude 
\begin{align*} 
    &1_{(T_{G_t,n-1} < \infty)} \mathds{1}(t) \mathbb{E}[ f_n(T_{G_t,n},X_{G_t,n}) \mid 
 (X_s)_{s \leq T_{G_t,n-1}}, X_{G_t}, 1_{( \tau_{\{d\}} \leq t)}] \\
 &=  \mathbb{E}[ 1_{(T_{G_t,n-1} < \infty)} \mathds{1}(t) f_n(T_{G_t,n},X_{G_t,n}) \mid 
 (X_s)_{s \leq T_{G_t,n-1}}, X_{G_t}, 1_{( \tau_{\{d\}} \leq t)}] \\ 
 &= 1_{(T_{G_t,n-1} < \infty)} \mathds{1}(t) \Tilde{E}^{X_{T_{G_t,n-1}}}[f_n(\tau_1,\eta_1)].
\end{align*}
We can further strengthen this to
\begin{align} \label{eq:InductionResult}
\begin{split}
    \mathds{1}(t) \mathbb{E}[ f_n(T_{G_t,n},X_{G_t,n}) \mid  (X_s)_{s \leq T_{G_t,n-1}}, X_{G_t},1_{(\tau_{\{d\}} \leq t)}] =  \mathds{1}(t) \Tilde{E}^{X_{T_{G_t,n-1}}}[f_n(\tau_1,\eta_1)]
\end{split}    
\end{align}
since also 
\begin{align*}
    &1_{(T_{G_t,n-1} = \infty)} \mathbb{E}[ f_n(T_{G_t,n},X_{G_t,n}) \mid (X_s)_{s \leq T_{G_t,n-1}}, X_{G_t}, 1_{(\tau_{\{d\}} \leq t)}] \\
    &= 1_{(T_{G_t,n-1} = \infty)} f_n(\infty,\nabla) \\
    &=  1_{(T_{G_t,n-1} = \infty)} \Tilde{E}^{X_{T_{G_t,n-1}}}[f_n(\tau_1,\eta_1)].
\end{align*}

Now we are in a position to use induction over $n$ in Equation~\eqref{eq:InductionDesired}.

\bigskip \noindent (iv) : Using Equation~\eqref{eq:InductionResult} with $n=1$ gives the result from Equation~\eqref{eq:InductionDesired} for $n=1$. For $n \geq 2$, assume the result holds for $n-1$ and observe
\begin{align*} 
       & \mathds{1}(t) \mathbb{E} \left[ \prod_{i=1}^n f_i(T_{G_t,i},X_{G_t,i})  \: \Big\vert \: (X_s)_{s \leq G_t}, 1_{( \tau_{\{d\}} \leq t)} \right] \\
       &= \mathds{1}(t) \mathbb{E} \Bigg[ \mathbb{E}[f_n(T_{G_t,n},X_{G_t,n}) \mid (X_s)_{s \leq T_{G_t,n-1}},X_{G_t},1_{( \tau_{\{d\}} \leq t)} ]  \\ 
       & \qquad \qquad \times \prod_{i=1}^{n-1} f_i(T_{G_t,i},X_{G_t,i}) \: \Big\vert \: (X_s)_{s \leq G_t}, 1_{( \tau_{\{d\}} \leq t)} \Bigg] \\
        &\leftstackrel{\textnormal{(iii)}}{=} \mathds{1}(t) \mathbb{E} \left[ \Tilde{E}^{X_{T_{G_t,n-1}}}[f_n(\tau_1,\eta_1)] \times \prod_{i=1}^{n-1} f_i(T_{G_t,i},X_{G_t,i})   \: \Big\vert \: (X_s)_{s \leq G_t}, 1_{( \tau_{\{d\}} \leq t)} \right] \\
        &= \mathds{1}(t) \Tilde{E}^{X_{G_t}} \left[ \Tilde{E}^{\eta_{n-1}}[f_n(\tau_1,\eta_1)] \times \prod_{i=1}^{n-1} f_i(\tau_i,\eta_i)  \right]
\end{align*}
where the last equality follows by the induction hypothesis. Continuing the calculations, we see
\begin{align*} 
       \mathds{1}(t) \Tilde{E}^{X_{G_t}} \left[ \prod_{i=1}^{n-1} f_i(\tau_i,\eta_i) \times \Tilde{E}^{\eta_{n-1}}[f_n(\tau_1,\eta_1)]  \right] = \mathds{1}(t) \Tilde{E}^{X_{G_t}} \left[ \prod_{i=1}^{n} f_i(\tau_i,\eta_i)  \right],
\end{align*}
by writing out the expectations using the Markov kernels. Hence we have shown
$$ \mathds{1}(t)\mathbb{E} \left[ \prod_{i=1}^n f_i(T_{G_t,i},X_{G_t,i})  \mid (X_s)_{s \leq G_t}, 1_{( \tau_{\{d\}} \leq t)}  \right] = \mathds{1}(t) \Tilde{E}^{X_{G_t}} \left[ \prod_{i=1}^{n} f_i(\tau_i,\eta_i)  \right]$$
so Equation~\eqref{eq:InductionDesired} hold for all $n \in \mathbb{N}$ by induction, which proves the desired result.
\end{proof}

\noindent For the proof of Proposition~\ref{prop:strongMarkov}, we needed Lemma~\ref{lemmaStoppingtime}, which corresponds to Proposition~\ref{prop:strongMarkov} if the random time $G_t$ was replaced by a $\mathcal{F}^X$-stopping time in some of the terms.
\begin{lemma} (Strong Markov type property at stopping time.) \label{lemmaStoppingtime} \\
Under Assumption~\ref{assumption:IndependenceG} it holds that
\begin{align} \label{eq:StrongMarkovLemma}
    1_{(T < \infty)}1_{(G_t \leq T)} \mathds{1}(t)  \mathbb{P}(\theta_{T} \mu \in \cdot \mid (X_s)_{s \leq T}, X_{G_t}, 1_{( \tau_{\{d\}} \leq t)}) = 1_{(T < \infty)}1_{(G_t \leq T)} \mathds{1}(t) \Tilde{Q}^{X_T}(\cdot)
\end{align}
for any $\mathcal{F}^X$-stopping time $T$.
\end{lemma}

\noindent The proof is very similar to the proof of Proposition~\ref{prop:strongMarkov}.

\begin{proof}
This is equivalent to showing for all $n \geq 1$ and all measurable bounded functions $f_i$, $1 \leq i \leq n$ that
\begin{align*}
    &1_{(T < \infty)}1_{(G_t \leq T)} \mathds{1}(t) \mathbb{E}\left[ \prod_{i=1}^n f_i(T_{T,i},X_{T,i}) \: \Big\vert \: (X_s)_{s \leq T}, X_{G_t},1_{( \tau_{\{d\}} \leq t)}  \right] \\
    &= 1_{(T < \infty)} 1_{(G_t \leq T)} \mathds{1}(t) \Tilde{E}^{X_T} \left[ \prod_{i=1}^n f_i(\tau_i, \eta_i) \right].
\end{align*}
Define
\begin{align*}
T(M) &= \sum_{m=1}^\infty t_{Mm} 1_{(t_{M(m-1)}\leq T < t_{Mm})} + \infty 1_{(T=\infty)}
\end{align*}
with $t_{Mm} = m2^{-M}$. Note that $(T(M) < \infty) = (T < \infty)$. We partition the proof into four parts corresponding to the four parts of the proof of Proposition~\ref{prop:strongMarkov}.

\bigskip \noindent (i) : First we show the result~(\ref{eq:StrongMarkovLemma}) with the discrete random time $T(M)$ in place of $T$, i.e.
\begin{align} \label{eq:StoppingTimeResultTM}
\begin{split}
        &1_{(T < \infty)}1_{(G_t \leq T(M))} \mathds{1}(t)  \mathbb{P}(\theta_{T(M)} \mu \in \cdot \mid  (X_s)_{s \leq T(M)}, X_{G_t},1_{( \tau_{\{d\}} \leq t)}  ) \\
    &= 1_{(T < \infty)}1_{(G_t \leq T(M))}  \mathds{1}(t) \Tilde{Q}^{X_{T(M)}}(\cdot).
\end{split}
\end{align}
Take $H \in \mathcal{H}$ and $F$ on the form 
$$F = ( (X_s)_{s \leq T(M)} \in B) \cap (X_{G_t} \in C) \cap (1_{(\tau_{\{d\}} \leq t)} \in D).$$ 
Note that the collection of such sets constitutes an intersection-stable generator for $\sigma((X_s)_{s \leq T(M)}) \\ \vee \sigma(X_{G_t}) \vee \sigma(1_{(\tau_{\{d\}} \leq t)})$ containing $\Omega$. Write $F_{M,m} = F \cap (T(M) = t_{Mm})$, and note that $F_{M,m} \in \mathcal{F}^X_{t_{Mm}} \vee \sigma(X_{G_t}) \vee \sigma(1_{(\tau_{\{d\}} \leq t)})$ and also $(\textnormal{RBNS}_t) \cap (G_t \leq t_{Mm}) \in \mathcal{F}^X_{t_{Mm}} \vee \sigma(X_{G_t}) \vee \sigma(1_{(\tau_{\{d\}} \leq t)})$. Then
\begin{align*}
    & \int_{F_{M,m}} 1_{(G_t \leq T(M))} \mathds{1}(t) \Tilde{Q}^{X_{T(M)}}(H) d\mathbb{P} \\
    &= \int_{F_{M,m}} 1_{(G_t \leq t_{Mm})} \mathds{1}(t) \Tilde{Q}^{X_{t_{Mm}}}(H) d\mathbb{P} \\
    & = \int_{F_{M,m}} 1_{(G_t \leq t_{Mm})} \mathds{1}(t) \mathbb{P}(\theta_{t_{Mm}} \mu \in H \mid X_{t_{Mm}}, 1_{(\tau_{\{d\}} \leq t)}) d\mathbb{P} \\
    & =  \int_{F_{M,m}} 1_{(G_t \leq t_{Mm})} \mathds{1}(t) \mathbb{P}(\theta_{t_{Mm}} \mu \in H \mid \mathcal{F}^X_{t_{Mm}} \vee \sigma(X_{G_t}) \vee \sigma(1_{(\tau_{\{d\}} \leq t)}))  d\mathbb{P} \\
    &= \int_{F_{M,m}}  1_{(G_t \leq t_{Mm})} \mathds{1}(t) 1_{(\theta_{t_{Mm}} \mu \in H)} d\mathbb{P} \\
    &= \mathbb{P}(F_{M,m}  \cap (G_t \leq T(M)) \cap (\textnormal{RBNS}_t) \cap (\theta_{T(M)} \mu \in H)).
\end{align*}
The third equality is Assumption~\ref{assumption:IndependenceG} with $v=t_{Mm}$. Summing over $m \geq 1$ and using uniqueness of finite measures on intersection-stable generators gives the result for $T(M)$ on $(T(M) < \infty)=(T < \infty)$. This now gives Equation~\eqref{eq:StoppingTimeResultTM}.

\bigskip \noindent (ii) : Note that for any bounded measurable function $f$, we have
\begin{align*}
    \Tilde{E}^{X_T}[f(\tau_1,\eta_1)] &= \int_{(T,\infty]} \int_{\overline{E}} f(v,y) \: \Tilde{r}_v(\phi_{Tv}(X_T),\diff y) \: \Tilde{F}_{X_T}(\diff v)
\end{align*}
with a similar expression for $\Tilde{E}^{X_{T(M)}}[f(\tau_1,\eta_1)]$. On $(\overline{N}_T = \overline{N}_{T(M)})$, it holds for $v \geq T(M)$, since also $T(M) \geq T$, that
\begin{align*}
    \phi_{T v} (X_T) &= \phi_{T(M) v}(\phi_{T T(M)}(X_T)) \\
    &= \phi_{T(M) v}(X_{T(M)})
\end{align*}
and 
$$\overline{\Tilde{F}}_{X_T}(v) = \overline{\Tilde{F}}_{X_T}(T(M)) \overline{\Tilde{F}}_{X_{T(M)}}(v).$$
Hence, we have on $(\overline{N}_T = \overline{N}_{T(M)}) \cap (T < \infty)$
\begin{align*}
    \Tilde{E}^{X_{T(M)}}[f(\tau_1,\eta_1)] &= \int_{(T(M),\infty]} \int_{\overline{E}} f(v,y) \: \Tilde{r}_v( \phi_{T(M) v}(X_{T(M)}),\diff y) \: \Tilde{F}_{X_{T(M)}} (\diff v) \\
    &= \frac{1}{\overline{\tilde{F}}_{X_T}(T(M))} \int_{(T(M),\infty]} \int_{\overline{E}} f(v,y) \: \Tilde{r}_v(\phi_{Tv}(X_T),\diff y) \: \Tilde{F}_{X_T}(\diff v) \\
    & \stackrel{M \rightarrow \infty}{\rightarrow} \int_{(T,\infty]} \int_{\overline{E}} f(v,y) \: \Tilde{r}_v(\phi_{Tv}(X_T),\diff y) \: \Tilde{F}_{X_T}(\diff v) \\
    &= \Tilde{E}^{X_T}[f(\tau_1,\eta_1)].
\end{align*}
Since also $1_{(\overline{N}_{T(M)} = \overline{N}_T)} \rightarrow 1$ for $M \rightarrow \infty$, we have
\begin{align*}
    \lim_{M \rightarrow \infty} 1_{(T < \infty)} \Tilde{E}^{X_{T(M)}}[f(\tau_1,\eta_1)] &= \lim_{M \rightarrow \infty} 1_{(\overline{N}_{T(M)} = \overline{N}_T)} 1_{(T < \infty)} \Tilde{E}^{X_{T(M)}}[f(\tau_1,\eta_1)] \\
    &= \lim_{M \rightarrow \infty} 1_{(\overline{N}_{T(M)} = \overline{N}_T)} 1_{(T < \infty)} \Tilde{E}^{X_T}[f(\tau_1,\eta_1)] \\
    &= 1_{(T < \infty)} \Tilde{E}^{X_T}[f(\tau_1,\eta_1)].
\end{align*}
We are now in a position to use dominated convergence.

\bigskip \noindent (iii) : Take $F$ on the form $F = ( (X_s)_{s \leq T} \in B) \cap (X_{G_t} \in C) \cap ( 1_{(\tau_{\{d\}} \leq t)} \in D)$ and write 
$$F_M =  ( (X_s)_{s \leq T(M)} \in B) \cap (X_{G_t} \in C) \cap ( 1_{(\tau_{\{d\}} \leq t)} \in D).$$
Note that $F_M \in \sigma((X_s)_{s \leq T(M)}) \vee \sigma(X_{G_t}) \vee \sigma(1_{(\tau_{\{d\}} \leq t)})$ and $\lim_{M \rightarrow \infty} 1_{F_M} = 1_F$ since $X$ is right-continuous and $T(M) \downarrow T$. Similarly, $\lim_{M \rightarrow \infty} 1_{(G_t \leq T(M))} = 1_{(G_t \leq T)}$ since the indicator $s \mapsto 1_{(G_t \leq s)}$ is right-continuous and $T(M) \downarrow T$.
By dominated convergence and the results from (i) and (ii), we obtain:
\begin{align*}
    &\int_F 1_{(T < \infty)}1_{(G_t \leq T)} \mathds{1}(t) \Tilde{E}^{X_T}[f(\tau_1,\eta_1)] d\mathbb{P} \\
    & \leftstackrel{\textnormal{(ii)}}{=} \lim_{M \rightarrow \infty} \int_{F_M}  1_{(T < \infty)}1_{(G_t \leq T(M))}\mathds{1}(t) \Tilde{E}^{X_{T(M)}}[f(\tau_1,\eta_1)] d\mathbb{P} \\
    &\leftstackrel{\textnormal{(i)}}{=} \lim_{M \rightarrow \infty} \int_{F_M} 1_{(T < \infty)}1_{(G_t \leq T(M))} \mathds{1}(t) \\
    & \qquad \qquad \qquad \times \mathbb{E}[f( T_{T(M),1}, X_{T(M),1}) \mid (X_s)_{s \leq T(M)}, X_{G_t},1_{(\tau_{\{d\}} \leq t)}  ] d\mathbb{P} \\
    &= \lim_{M \rightarrow \infty} \int_{F_M} 1_{(T < \infty)}1_{(G_t \leq T(M))} \mathds{1}(t) f(T_{T(M),1}, X_{T(M),1}) d\mathbb{P} \\
    &= \int_{F} 1_{(T < \infty)}1_{(G_t \leq T)} \mathds{1}(t) f(T_{T,1}, X_{T,1}) d\mathbb{P},
\end{align*}
where the third equality follows by 
$$F_M \cap (T < \infty) \cap (G_t \leq T(M)) \cap (\textnormal{RBNS}_t) \in \sigma((X_s)_{s \leq T(M)} \vee \sigma(X_{G_t}) \vee \sigma(1_{(\tau_{\{d\}} \leq t)})$$ 
 and the last equality follows since $\lim_{M \rightarrow \infty} 1_{F_M} = 1_F$.
By uniqueness of finite measures on intersection stable generators, this shows that
\begin{align*}
    &1_{(T < \infty)} 1_{(G_t \leq T)} \mathds{1}(t) \Tilde{E}^{X_T}[f(\tau_1,\eta_1)] \\
    &=  \mathbb{E}[1_{(T < \infty)} 1_{(G_t \leq T)}  \mathds{1}(t) f(T_{T,1},X_{T,1}) \mid (X_s)_{s \leq T},X_{G_t},1_{(\tau_{\{d\}} \leq t)} ] \\
    &= 1_{(T < \infty)} 1_{(G_t \leq T)}  \mathds{1}(t) \mathbb{E}[f(T_{T,1},X_{T,1}) \mid (X_s)_{s \leq T},X_{G_t},1_{(\tau_{\{d\}} \leq t)} ].
\end{align*}
The result also holds when removing the indicator $1_{(T < \infty)}$ since 
\begin{align*}
    &1_{(T = \infty)} 1_{(G_t \leq T)} \mathds{1}(t) \Tilde{E}^{X_T}[f(\tau_1,\eta_1)] \\
    &= 1_{(T = \infty)} 1_{(G_t \leq T)} \mathds{1}(t) f(\infty,\nabla) \\
    &= 1_{(T = \infty)} 1_{(G_t \leq T)}  \mathds{1}(t) \mathbb{E}[f(T_{T,1},X_{T,1}) \mid (X_s)_{s \leq T},X_{G_t},1_{(\tau_{\{d\}} \leq t)} ].
\end{align*}

\bigskip \noindent (iv) : Using the result of (iii) with $f=f_1$ gives the base case of the induction. For $n \geq 2$, assume that the result hold for $n-1$. Observe then
\begin{align*}
    &1_{(T < \infty )} 1_{(G_t \leq T)} \mathds{1}(t) \mathbb{E}\left[ \prod_{i=1}^n f_i(T_{T,i},X_{T,i}) \mid (X_s)_{s \leq T}, X_{G_t},1_{(\tau_{\{d\}} \leq t)} \right] \\
    &=\mathbb{E}\Bigg[  1_{(T < \infty )} 1_{(G_t \leq T)} \mathds{1}(t) \mathbb{E}[f_n(T_{T,n},X_{T,n}) \mid (X_s)_{s \leq T_{T,n-1}}, X_{G_t},1_{(\tau_{\{d \}} \leq t)}] \times  \\
    & \qquad \quad \prod_{i=1}^{n-1} f_i(T_{T,i},X_{T,i}) \mid (X_s)_{s \leq T}, X_{G_t},1_{(\tau_{\{d\}} \leq t)} \Bigg] \\
    &\leftstackrel{\textnormal{(iii)}}{=}  \mathbb{E}\bigg[ 1_{(T < \infty )} 1_{(G_t \leq T)} \mathds{1}(t)  \Tilde{E}^{X_{T_{T,n-1}}}[f_n(\tau_1,\eta_1)] \prod_{i=1}^{n-1} f_i(T_{T,i},X_{T,i}) \mid (X_s)_{s \leq T},X_{G_t},1_{(\tau_{\{d\}} \leq t)} \bigg] \\
    &= 1_{(T < \infty )} 1_{(G_t \leq T)} \mathds{1}(t) \mathbb{E}\left[ \Tilde{E}^{X_{T_{T,n-1}}}[f_n(\tau_1,\eta_1)] \prod_{i=1}^{n-1} f_i(T_{T,i},X_{T,i})  \mid (X_s)_{s \leq T},X_{G_t},1_{(\tau_{\{d\}} \leq t)} \right]  \\
    &= 1_{(T < \infty )} 1_{(G_t \leq T)} \mathds{1}(t) \Tilde{E}^{X_T}\left[ \prod_{i=1}^{n-1} f_i(\tau_i,\eta_i) \Tilde{E}^{\eta_{n-1}}[f_n(\tau_1,\eta_1) ] \right] 
\end{align*}
where the first equality follows by the tower property and the second follows by (iii) with $f=f_n$ and the $\mathcal{F}^X$-stopping time $T_{T,n-1}$ using also that $(G_t \leq T) \subseteq (G_t \leq T_{T,n-1})$. The last equality follows by the induction hypothesis. We finally see, 
\begin{align*}
    &1_{(T < \infty )} 1_{(G_t \leq T)} \mathds{1}(t) \Tilde{E}^{X_T}\left[ \prod_{i=1}^{n-1} f_i(\tau_i,\eta_i)  \Tilde{E}^{\eta_{n-1}}[f_n(\tau_1,\eta_1) ] \right] \\
    &= 1_{(T < \infty )} 1_{(G_t \leq T)} \mathds{1}(t) \Tilde{E}^{X_T}\left[ \prod_{i=1}^{n} f_i(\tau_i, \eta_i)   \right]  
\end{align*}
by writing out the expectations using the Markov kernels. This concludes the induction and the proof.

\end{proof}

\end{appendices}


\begin{thebibliography}{29}
\providecommand{\natexlab}[1]{#1}
\providecommand{\url}[1]{\texttt{#1}}
\expandafter\ifx\csname urlstyle\endcsname\relax
  \providecommand{\doi}[1]{doi: #1}\else
  \providecommand{\doi}{doi: \begingroup \urlstyle{rm}\Url}\fi




\bibitem[{Andersen et~al.(1993)Andersen, Borgan, Gill \&
  Keiding}]{Andersen.etal:1993}
\textsc{Andersen, P.~K.}, \textsc{Borgan, O.}, \textsc{Gill, R.~D.} \& \textsc{Keiding, N.} (1993).
\textit{{Statistical Models Based on Counting Processes}}.
\newblock New York: Springer.

\bibitem[{Antonio et~al.(2016)Antonio, Godecharle, and Van Oirbeek}]{Antonio.etal:2016}
\textsc{Antonio, K.}, \textsc{Godecharle, E.} \& \textsc{Van Oirbeek, R.} (2016).
\newblock {A multi-state approach and flexible payment distributions for micro-level reserving in general insurance}.
\newblock Preprint. Available at SSRN 2777467.



\bibitem[{Antonio and Plat(2014)}]{Antonio:Plat:2014}
\textsc{Antonio, K.} \& \textsc{Plat, R.} (2014).
\newblock {Micro-level stochastic loss reserving for general insurance}.
\newblock \textit{Scandinavian Actuarial Journal} \textbf{2014}, 649--669.




\bibitem[{Avanzi et~al.(2026)Avanzi, Richman, Wong, W{\"u}thrich, and Xie}]{Avanzi.etal:2026}
\textsc{Avanzi, B.}, \textsc{Richman, R.}, \textsc{Wong, B.}, \textsc{W{\"u}thrich, M.} \& \textsc{Xie, Y.} (2026).
\newblock {Reinforcement Learning for Micro-Level Claims Reserving}.
\newblock arXiv preprint arXiv:2601.07637.


\bibitem[{Badescu et~al.(2016)Badescu, Lin, and Tang}]{Badescu.etal:2016}
\textsc{Badescu, A.~L.}, \textsc{Lin, X.~S.} \& \textsc{Tang, D.} (2016).
\newblock {A marked Cox model for the number of IBNR claims: Theory}.
\newblock \textit{Insurance: Mathematics and Economics} \textbf{69}, 29--37.



\bibitem[{Bettonville et~al.(2021)Bettonville, d'Oultremont, Denuit, Trufin and Van Oirbeek}]{Bettonville.etal:2021}
\textsc{Bettonville, C.}, \textsc{d'Oultremont, L.}, \textsc{Denuit, M.}, \textsc{Trufin, J.} \& \textsc{Van Oirbeek, R.} (2021).
\newblock {Matrix calculation for ultimate and 1-year risk in the Semi-Markov individual loss reserving model}.
\newblock \textit{Scandinavian Actuarial Journal} \textbf{2021}, 380--407.


\bibitem[{Boumezoued and Devineau(2017)}]{Boumezoued:Devineau:2017}
\textsc{Boumezoued, A.} \& \textsc{Devineau, L.} (2017).
\newblock {Individual claims reserving: a survey}.
\newblock Preprint. Available online at https://hal.archives-ouvertes.fr/hal-01643929.


\bibitem[{Buchardt et~al.(2023)Buchardt, Furrer, and Sandqvist}]{Buchardt.etal:2023a}
\textsc{Buchardt, K.}, \textsc{Furrer, C.} \& \textsc{Sandqvist, O.~L.} (2023).
\newblock {Transaction time models in multi-state life insurance}.
\newblock \textit{Scandinavian Actuarial Journal} \textbf{2023}, 974--999.

\bibitem[Buchardt et~al.(2025)Buchardt, Furrer, and Sandqvist]{Buchardt.etal:2023b}
\textsc{Buchardt, K.}, \textsc{Furrer, C.} \& \textsc{Sandqvist, O.~L.} (2025).
\newblock Estimation for multistate models subject to reporting delays and incomplete event adjudication.
\newblock \textit{The Annals of Applied Statistics} \textbf{19},  2782--2802.

\bibitem[{Buchardt et~al.(2015)Buchardt, M{\o}ller, and Schmidt}]{Buchardt.etal:2015}
\textsc{Buchardt, K.}, \textsc{M{\o}ller, T.} \& \textsc{Schmidt, K.~B.} (2015).
\newblock {Cash flows and policyholder behaviour in the semi-Markov life insurance setup}.
\newblock \textit{Scandinavian Actuarial Journal} \textbf{2015}, 660--688.

\bibitem[{B{\"u}cher and Rosenstock(2024)}]{Bucher:Rosenstock:2024}
\textsc{B{\"u}cher, A.} \& \textsc{Rosenstock, A.} (2024).
\newblock {Combined modelling of micro-level outstanding claim counts and individual claim frequencies in non-life insurance}.
\newblock \textit{{European Actuarial Journal}}, 1--33.

\bibitem[{Christiansen(2012)}]{Christiansen:2012}
\textsc{Christiansen, M.} (2012).
\newblock {Multistate models in health insurance}.
\newblock \textit{AStA Advances in Statistical Analysis} \textbf{96}, 155--186. 

\bibitem[{Christiansen and Furrer(2021)}]{Christiansen:Furrer:2021}
\textsc{Christiansen, M.~C.} \& \textsc{Furrer, C.} (2021).
\newblock {Dynamics of state-wise prospective reserves in the presence of non-monotone information}.
\newblock \textit{{Insurance: Mathematics and Economics}} \textbf{97}, 81--98.

\bibitem[{Crevecoeur et~al.(2022b)Crevecoeur, Antonio, Desmedt, and Masquelein}]{Crevecoeur.etal:2022b}
\textsc{Crevecoeur, J.}, \textsc{Antonio, K.}, \textsc{Desmedt, S.} \& \textsc{Masquelein, A.} (2022b).
\newblock {Bridging the gap between pricing and reserving with an occurrence and development model for non-life insurance claims}.
\newblock \textit{ASTIN Bulletin} \textbf{53}, 185--212.

\bibitem[{Dabrowska(1995)}]{Dabrowska:1995}
\textsc{Dabrowska, D.} (1995).
\newblock {Estimation of transition probabilities and bootstrap in a semiparametric Markov renewal model}.
\newblock \textit{Journal of Nonparametric Statistics} \textbf{5}, 237--259. 


\bibitem[{Delong et~al.(2022)Delong, Lindholm, and W{\"u}thrich}]{Delong.etal:2022}
\textsc{Delong, {\L}.}, \textsc{Lindholm, M.} \& \textsc{W{\"u}thrich, M.~V.} (2022).
\newblock {Collective reserving using individual claims data}.
\newblock \textit{Scandinavian Actuarial Journal} \textbf{2022}, 1--28.



\bibitem[{Furrer and Sandqvist(2026)Furrer, Christian \& Sandqvist, Oliver}]{Furrer:Sandqvist:2026}
\textsc{Furrer, C.} \& \textsc{Sandqvist, O.~L.} (2026).
\newblock {Loss of earning capacity in Denmark--an actuarial perspective}.
\newblock arXiv preprint arXiv:2501.11578.


\bibitem[{Haberman and Pitacco(1998)}]{Haberman:Pitacco:1998}
\textsc{Haberman, S.} \& \textsc{Pitacco, E.} (1998).
\textit{{Actuarial models for disability insurance}}.
\newblock Boca Raton: Chapman \& Hall.

\bibitem[{Helwich(2008)}]{Helwich:2008}
\textsc{Helwich, M.} (2008).
\newblock \emph{Durational effects and non-smooth semi-{M}arkov models in life
  insurance}.
\newblock PhD thesis, University of Rostock.

\bibitem[{Hoem(1972)}]{Hoem:1972}
\textsc{Hoem, J.} (1972).
\newblock {Inhomogeneous semi-Markov processes, select actuarial tables, and duration-dependence in demography}.
\newblock In: \textit{Population Dynamics}.
\newblock {New York: Academic Press, pp.\ 251--296}. 

\bibitem[{Hougaard(2000)}]{Hougaard:2000}
\textsc{Hougaard, P.} (2000).
\textit{{Analysis of multivariate survival data}}.
\newblock New York: Springer.

\bibitem[{Jacobsen(2006)}]{Jacobsen:2006}
\textsc{Jacobsen, M.} (2006).
\textit{{Point process theory and applications: Marked point and piecewise deterministic processes}}.
\newblock Boston: Birkhauser.

\bibitem[{Janssen(1966)}]{Janssen:1966}
\textsc{Janssen, J.} (1966).
\newblock {Application des processus semi-markoviens {\`a} un probl{\'e}me d’invalidit{\'e}}.
\newblock \textit{Bulletin de l’Association Royale des Actuaries Belges} \textbf{63}, 35--52. 

\bibitem[{Lagakos et~al.(1978)Lagakos, Sommer, and Zelen}]{Lagakos.etal:1978}
\textsc{Lagakos, S.~W.}, \textsc{Sommer, C.~J.} \& \textsc{Zelen, M.} (1978).
\newblock {Semi-Markov models for partially censored data}.
\newblock \textit{{Biometrika}} \textbf{65}, 311--317.


\bibitem[Lopez et~al.(2018)Lopez, Milhaud, and Th{\'e}rond]{Lopez.etal:2018}
\textsc{Lopez, O.}, \textsc{Milhaud, X.} \& \textsc{Th{\'e}rond, P.~E.} (2018).
\newblock Micro-level VS macro-level reserving in non-life insurance: why and when?
\newblock Preprint. Available online at https://hal.science/hal-01868744.


\bibitem[{Lopez et~al.(2019)Lopez, Milhaud, and Th{\'e}rond}]{Lopez.etal:2019}
\textsc{Lopez, O.}, \textsc{Milhaud, X.} \& \textsc{Th{\'e}rond, P.~E.} (2019).
\newblock {A tree-based algorithm adapted to microlevel reserving and long development claims}.
\newblock \textit{ASTIN Bulletin: The Journal of the IAA} \textbf{49}, 741--762.


\bibitem[{Mack(1993)}]{Mack:1993}
\textsc{Mack, T.} (1993).
\newblock {Distribution-free calculation of the standard error of chain ladder reserve estimates}.
\newblock \textit{ASTIN Bulletin: The Journal of the IAA} \textbf{23}, 213--225.  


\bibitem[{Mack(1999)}]{Mack:1999}
\textsc{Mack, T.} (1999).
\newblock {The standard error of chain ladder reserve estimates: Recursive calculation and inclusion of a tail factor}.
\newblock \textit{ASTIN Bulletin: The Journal of the IAA} \textbf{29}, 361--366.  

\bibitem[{Norberg(1990)}]{Norberg:1990}
\textsc{Norberg, R.} (1990).
\newblock {Payment measures, interest, and discounting: an axiomatic approach with applications to insurance}.
\newblock \textit{Scandinavian Actuarial Journal} \textbf{1990}, 14--33.  

\bibitem[{Norberg(1993)}]{Norberg:1993}
\textsc{Norberg, R.} (1993).
\newblock {Prediction of Outstanding Liabilities in Non-Life Insurance}.
\newblock \textit{ASTIN Bulletin} \textbf{23}, 95--115. 

\bibitem[{Norberg(1999)}]{Norberg:1999}
\textsc{Norberg, R.} (1999).
\newblock {Prediction of Outstanding Liabilities II. Model Variations and Extensions}.
\newblock \textit{ASTIN Bulletin} \textbf{29}, 5--25. 


\bibitem[{Okine et~al.(2022)Okine, Frees, and Shi}]{Okine.etal:2022}
\textsc{Okine, A.~N.}, \textsc{Frees, E.~W.} \& \textsc{Shi, P.} (2022).
\newblock {Joint model prediction and application to individual-level loss reserving}.
\newblock \textit{ASTIN Bulletin: The Journal of the IAA} \textbf{52}, 91--116.


\bibitem[{{R Development Core Team}(2023)}]{R:2023}
\textsc{{R Development Core Team}} (2023).
\newblock \textit{R: A Language and Environment for Statistical Computing}.
\newblock Vienna, Austria: R Foundation for Statistical Computing.
\newblock {ISBN} 3-900051-07-0, http://www.R-project.org.

\bibitem[{Schervish(1995)}]{Schervish:1995}
\textsc{Schervish, M.~J.} (1995).
\textit{{Theory of statistics}}.
\newblock New York: Springer.


\bibitem[{Shah and Peters(2020)}]{Shah:Peters:2020}
\textsc{Shah, R.~D.} \& \textsc{Peters, J.} (2020).
\newblock {The hardness of conditional independence testing and the generalised covariance measure}.
\newblock \textit{The Annals of Statistics} \textbf{48}, 1514--1538.


\bibitem[{Spitoni et~al.(2012)Spitoni, Verduijn, and Putter}]{Spitoni.etal:2012}
\textsc{Spitoni, C.}, \textsc{Verduijn, M.} \& \textsc{Putter, H.} (2012).
\newblock {Estimation and asymptotic theory for transition probabilities in Markov renewal multi-state models}.
\newblock \textit{{The International Journal of Biostatistics}} \textbf{8}, 23.

\bibitem[{Vanegas et~al.(2025)Vanegas, Badescu, and Lin}]{Vanegas.etal:2025}
\textsc{Vanegas, S.~C.}, \textsc{Badescu, A.~L.} \& \textsc{Lin, X.~S.} (2025).
\newblock {Claim reserving via inverse probability weighting: A micro-level chain-ladder method}.
\newblock \textit{European Actuarial Journal} \textbf{15}, 89--135.

\bibitem[{W{\"u}thrich(2018)}]{Wuthrich:2018}
\textsc{W{\"u}thrich, M.~V.} (2018).
\newblock {Machine learning in individual claims reserving}.
\newblock \textit{Scandinavian Actuarial Journal} \textbf{2018}, 465--480.


\bibitem[{Yackel(1968)}]{Yackel:1968}
\textsc{Yackel, J.} (1968).
\newblock {A random time change relating semi-Markov and Markov processes}.
\newblock \textit{The Annals of Mathematical Statistics} \textbf{39}, 358--364.  

\end{thebibliography}
\end{document}


\def\spacingset#1{\renewcommand{\baselinestretch}%
{#1}\small\normalsize} \spacingset{1}

\markboth{O.L.~Sandqvist}{A multistate approach to disability insurance reserving with information delays}

\title{Supplementary material for “A multistate approach to disability insurance reserving with information delays”}

  \author{O.L.~Sandqvist$^{\textnormal{a,b}}$ \vspace{0.5cm}\\
    $^\textnormal{a}$PFA Pension, Sundkrogsgade 4, DK-2100 Copenhagen \O, Denmark \\
    $^\textnormal{b}$Department of Mathematical Sciences, University of Copenhagen, \\ Universitetsparken 5, DK-2100 Copenhagen \O, Denmark
    }   

\maketitle

The supplementary material is organized as follows. Supplement~\ref{sec:extraDisscussions} contains additional discussions on the modeling assumptions, the extension to stochastic interest rates is given in Supplement~\ref{sec:stochasticInterest}, and additional details on the estimation procedure are given in Supplement~\ref{sec:EstimationAppendix}.

\section{Additional discussions}
\label{sec:extraDisscussions}

\subsection{Valid time state space for disability insurance contract} \label{subsec:statespace} 

In the multistate modeling literature, it is common to allow for a general finite state space. We restrict our attention to the particular state space depicted in Figure~\ref{fig:Ystatespace} because, as was noted in the introduction, the relation between the valid time and transaction time reserves can be highly model-specific. The state space is intended to be general enough to capture most common disability insurance schemes. The hierarchical structure is imposed to simplify the transaction time model construction by making it so that there is only one disability and reactivation time to keep track of, as well as making the implementation in Section~\ref{sec:application} easier since we can avoid implementing the semi-Markov Kolmogorov forward differential equations known from~\citet{Buchardt.etal:2015} and instead use Thiele's differential equations successively. 

Modeling disability insurance contracts using the model from Figure~\ref{fig:Ystatespace} implies that at most
one disability can occur, that the disability type does not change after disablement, and that a reactivation of this disability is permanent. For contracts where it is important to model temporary reactivations, a possibility is to use a non-hierarchical illness-death model where reactivations are modeled as jumps back into state $a$ instead of into the separate reactivated state $r$, see e.g.\ Figure 3 in~\citet{Helwich:2008} or Example 2.1 in~\citet{Christiansen:2012}.

When coverage periods are short, as is usually the case for disability insurance schemes, and the disability hazard is small, ignoring the possibility of several disabilities can be reasonable. Even with a non-hierachical model, it can be complicated to allow for several distinct disabilities if the insurance contract includes a coverage period. To see this, consider the situation where the insured becomes disabled within the coverage period of a disability annuity, reactivates, and becomes disabled again outside of the coverage period. Whether the insured is qualified for disability payments for the second disability depends on whether or not it was caused by the disability event in the coverage period. 

The most natural way to capture this is to choose the disability event times to be those that lead to payout when estimating the disability hazard or to extend $X$ such that it contains information about which disability event is causing the current disablement. An alternative would be to let the payment rate in the disabled state be the average payment rate conditional on the historical development of $X$, see Section~\ref{subsec:SimplifyingAssumptions} for more details. These approaches would all lead to models with an intricate dependence on the past coverage periods and the historical development of $X$. In Section~\ref{subsec:SimplifyingAssumptions} we propose ways to obtain consistent reserves in situations where there may be several disabilities and/or transition between the disability types without having to use the non-hierarchical state space.

\subsection{Weakening of the independence assumptions.} \label{subsec:Weakening}
In the absence of Assumption~\ref{assumption:Independence} and~\ref{assumption:IndependenceG}, the valid time hazards and reserves could be influenced by additional transaction time information such as the disability reporting delay $U_{\mathcal{I}}$. This could be relevant if e.g.\ longer reporting delays were indicative of a more serious disability such that the intensity of $N_{ir}$ with respect to the filtration $t \mapsto \mathcal{F}^X_t \vee \sigma(U_{\mathcal{I}})$ was $\lambda_{ir}(t) = 1_{(Y_{t-} = i)} \mu_{ir}(t,U_{t-},U_{\mathcal{I}})$ with $\mu_{ir}$ being decreasing as a function of the last argument. Imposing an assumption like Assumption~\ref{assumption:Independence} when also conditioning on $U_{\mathcal{I}}$ would then for example lead to the award-term of the RBNSi reserve becoming
\begin{align*}
    \mathbb{E}[P(G_t) \mid \mathcal{F}_t^\mathcal{Z}, X_{G_t}=(G_t,i,0)] =  V_i(G_t,0,U_{\mathcal{I}}; (\tau^d > t))
\end{align*}
with obvious notation. We briefly note that it is not a priori clear whether long reporting delays would be indicative of more or less severe disabilities. One hypothesis is that people with severe disabilities would find it more demanding to submit insurance claims. On the other hand, they might not need as much time to collect medical evidence if it is self-evident that they will be approved for disability benefits.

An alternative way to weaken Assumption~\ref{assumption:Independence} would be to incorporate more of the transaction time information in the valid time model e.g.\ by using different disabled states for different disability severities instead of a single disabled state. If the reporting delays only affect our estimate of the future trajectory of the valid time process $X$ through the information they give us about the severity of the disability, Assumption~\ref{assumption:Independence} would be satisfied in the larger valid time model that incorporates information about the severity of the disability.

\subsection{On the simplifying assumptions} \label{subsec:SimplifyingAssumptions} 

 Some notable simplifications that have been made in order to arrive at tractable transaction time reserves are the independence assumptions (which were discussed in Remark~\ref{remark:Violation} and Appendix~\ref{subsec:Weakening}), that there can be at most one disability event in the coverage period which reaches the payout stage, that the time and type of the disability event is completely known once benefits have started, and that the insurer cannot retract disability benefits. While it is true that none of these assumptions fully hold in practice, we believe that they are not seriously violated, and they only contribute with second-order effects compared to the main effects that have been included. 
 
 For example, the probability of experiencing $q$ disabilities with independent causes is roughly equal to the disability hazard raised to the $q$'th power, so while more than one disability may occur in reality it is very uncommon. Similarly, there might be situations where the insurer and insured do not agree on the time of disablement leading to the time of disablement being changed to an earlier date after the insurer has started paying benefits. There may also be situations where retraction of disability benefits occurs if the insured willfully withheld information about their reactivation. However, the changes to the event times are probably sufficiently small and infrequent that the effect is negligible compared to other sources of error stemming from modeling, estimation, and forecasting error related to the biometric and financial model constituents. Accommodating transitions between the disabled states could be important depending on what the disability types represent. For example, they could represent different severities. If the severity of a disability changes often and different severities lead to substantially different payouts then this would be important to include in the model. We hence discuss possible remedies.
    
    If the aforementioned effects are sufficiently large to warrant explicit modeling, our models may still serve as a starting point. To accommodate the possibility of the time of disability changing after benefits have been awarded, we could for example add a term to the award-part of the RBNSi reserve equal to the the probability that the insured will be awarded benefits from an earlier time than $G_t$ multiplied with the average additionally awarded amount for such cases. To accommodate multiple disabilities, a pragmatic approach could be to place the disability periods end-to-end in the estimation phase, such that an additional disability was treated as an annulment of a reactivation. This skews the time value of the cash flow but otherwise results in consistent reserves. It will however likely make the individual predictions less precise since the duration dependence will stem from a more heterogeneous population. For example, a long disability duration could stem from a single disability from which the insured has not reactivated, or it could be that they have reactivated from a long disability but then recently become disabled again. 
    
    If transitions between the disabled states were possible in the valid time state space, we could formulate a transaction time model similar to the one specified towards the end of Example 5.8 in~\citet{Buchardt.etal:2023a}. This would however result in a substantially more complicated model. We instead propose to let disability type $i_k$ represent disabilities that start out as type $i_k$ and estimate the reactivation, disabled mortality, and reactivated mortality hazards consistently with this. If the valid time disability payments also depend on the disability type, we propose to model these payments conditional on the initial disability type and the disability duration. By the tower property, this leads to the same reserves but in practice requires that we also estimate these conditional payments, which brings the approach closer to the non-life insurance literature where the benefit sizes also have to be modeled. This approach may also be useful in other situations where the benefit size depends on more than the state and duration process; it could for example depend on whether the insured is receiving benefits from the government or other insurance companies. Taking the conditional expectation of the payments given the state and duration brings the problem back into something that can be represented in the usual semi-Markov framework.
    
    As illustrated here, a benefit of having interpretable analytically tractable expressions for the reserves is that it is possible to reason about how to adjust the model when the underlying assumptions change.

\section{Stochastic interest rate} \label{sec:stochasticInterest}

We consider the extension of the results to models with stochastic interest rates. Assume that $\kappa(t)=\exp\left(\int_{(0,t]} r(v) \diff v\right)$ with $r : \Omega \times \mathbb{R}_+ \mapsto \mathbb{R}$ being a stochastic process. We assume that there exists an equivalent martingale measure for the financial market which we denote $\mathbb{Q}$. Define the time-$t$ forward interest rate as 
$$f(t,s) = -\frac{\partial \log P(t,s)}{\partial s}$$ 
where for $0 \leq t \leq s$ we have $P(t,s) = \left. \mathbb{E}^{\mathbb{Q} } \left[ \kappa(t)/\kappa(s) \; \right\vert \mathcal{F}^r_t \right]$ is the price at time $t$ of a zero-coupon bond paying one unit at time $s$. Then $P(t,s) = \exp\left( \int_{(t,s]} f(t,v) \diff v \right)$. We introduce
$$\kappa_u(t) = \exp\left(\int_{(0,t]} r(u,v) \diff v\right)$$
for \[
    r(u,v) := \left\{\begin{array}{lr}
        r(v) , & v \leq u \\
        \vspace{-0.3cm} \\
        f(u,v), & v > u
        \end{array} \right.
  \]
being the realized interest rate before time $u$ and the forward interest rate after time $u$. Note for $t \leq u \leq s$ that 
\begin{align*}
    \mathbb{E}^{\mathbb{Q} }\left[ \frac{\kappa(s)}{\kappa(t)} \: \Big\vert \: \mathcal{F}^r_u\right] &= \frac{\kappa(u)}{\kappa(t)} \mathbb{E}^{\mathbb{Q} }\left[ \frac{\kappa(s)}{\kappa(u)} \: \Big\vert \: \mathcal{F}^r_u\right] \\
    &= \exp\left(\int_{(t,u]} r(v) \diff v\right) \exp\left(\int_{(u,s]} f(u,v) \diff v\right) \\
    &= \frac{\kappa_u(s)}{\kappa_u(t)}.
\end{align*}
The extension of our results to a stochastic interest rate is simple if there is no dependence or conditional dependence between the filtrations $\mathcal{F}^r$ and $\mathcal{F}^\mathcal{Z}$. Redefine the transaction time reserve as 
\begin{align*}
 \mathcal{V}(t) &= \mathbb{E}[\mathcal{P}(t) \mid \mathcal{F}^\mathcal{Z}_t \vee \mathcal{F}^r_t]
\end{align*}
where now $\mathbb{E} = \mathbb{E}^{\mathbb{P} \otimes \mathbb{Q}}$.
By Theorem 5.4 of~\citet{Buchardt.etal:2023a}, we see
\begin{align*}
    \mathcal{V}(t) &= \mathbb{E}\Big[ \int_{[0,\infty)} \frac{\kappa(t)}{\kappa(s)} B(\diff s) \: \big\vert \: \mathcal{F}^\mathcal{Z}_t \vee \mathcal{F}^r_t\Big] \\
    & \quad-\mathbb{E}\bigg[ \sum_{j=1}^J \int_{[0,t]} \frac{\kappa(t)}{\kappa(s)} 1_{(Y^t_{s-}=j)} B_{j,s-U^t_{s-}}(\diff s) +  \sum_{\substack{j,k=1 \\ j \neq k} }^J \int_{[0,t]} \frac{\kappa(t)}{\kappa(s)} b_{jk}(s,U^t_{s-}) N^t_{jk}(\diff s)  \: \big\vert \: \mathcal{F}^\mathcal{Z}_t \vee \mathcal{F}^r_t \bigg].
\end{align*}
Using the independence between $\mathcal{F}^r$ and $\mathcal{F}^\mathcal{Z}$, standard arguments now give 
\begin{align*}
    \mathcal{V}(t) &= \mathbb{E}\Big[ \int_{[0,\infty)} \frac{\kappa_t(t)}{\kappa_t(s)} B(\diff s) \: \big\vert \: \mathcal{F}^\mathcal{Z}_t \vee \mathcal{F}^r_t\Big] \\
    & \quad-\mathbb{E}\bigg[ \sum_{j=1}^J \int_{[0,t]} \frac{\kappa_t(t)}{\kappa_t(s)} 1_{(Y^t_{s-}=j)} B_{j,s-U^t_{s-}}(\diff s) +  \sum_{\substack{j,k=1 \\ j \neq k} }^J \int_{[0,t]} \frac{\kappa_t(t)}{\kappa_t(s)} b_{jk}(s,U^t_{s-}) N^t_{jk}(\diff s)  \: \big\vert \: \mathcal{F}^\mathcal{Z}_t \vee \mathcal{F}^r_t \bigg].
\end{align*}
This expression is identical to the expression for $\mathcal{V}$ in the case of a deterministic interest rate, but where $\kappa$ is replaced with $\kappa_t$. Now since $\mathcal{F}^r$ is completely exogeneous to $\mathcal{F}^\mathcal{Z}$, all the calculations in Section~\ref{sec:reserving} remain valid with $\kappa$ replaced by $\kappa_t$, and hence we conclude that the results presented in Theorem~\ref{thm:CBNRReserve},~\ref{thm:RBNSiReserve}, and~\ref{thm:RBNSrReserve} still hold true if we substitute $\kappa_t$ for $\kappa$ in all terms.

\begin{remark} (Run-off plots.) \label{rmk:Runoff} \\
Run-off plots are a common graphical tool used to validate the reserves. The intuition is that if the reserves are correctly specified, the reserve should be converted to payments such that the sum of the two stays constant over time, see Figure~\ref{fig:runoff} for an illustration. We show that if interest is handled appropriately, this approach is justified.
\tikzset{every picture/.style={line width=0.75pt}}       
\begin{figure}[H] \label{fig:runoff}
    \centering
    \scalebox{1}{
    \begin{tikzpicture}
    \begin{axis}[
        stack plots=y,
        area style,
        enlarge x limits=false,
        xticklabels={,,},
        yticklabels={,,},
        xlabel={Time},
        ylabel={Capital},
        width=11cm,
        height=5cm,
        tick style={semithick,color=black},
        xtick pos=left,
        ytick pos=left,
        axis lines = middle,
        x label style={at={(axis description cs:1,-0.05)},anchor=north},
        y label style={at={(axis description cs:-0.1,0.85)},anchor=south},
        legend pos=outer north east,
        legend cell align=left,
        axis line style={shorten >=-10pt, shorten <=-10pt}
    ]
        \addplot [fill=black!20] coordinates
            {(0,0) (1,1) (2,2) (3,3) (4,4) (5,5) (6,6)}
                \closedcycle;
        \addplot [fill=black!40] coordinates
            {(0,6) (1,8) (2,8.5) (3,8.25) (4,7.625) (5,6.8125) (6,5.90625)}
                \closedcycle;
        \addplot [fill=black!70] coordinates
            {(0,6) (1,3) (2,1.5) (3,0.75) (4,0.375) (5,0.1875) (6,0.09375)}
                \closedcycle;
        \legend{Payments,RBNS,IBNR};
        \end{axis}
    \end{tikzpicture}}
\caption{Illustration of a run-off plot with origin at the end of the coverage period, where reserves consist of only IBNR and RBNS contributions.}
\end{figure}
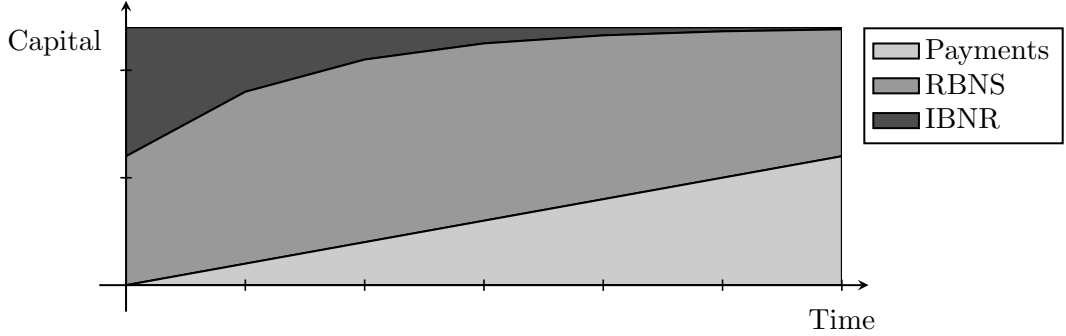

\noindent Assume that reserves are calculated with a fixed frequency (e.g.\ monthly) which we take to be the unit of the time scale. Let $\mathcal{V}^{\textnormal{pc}}(t)= \mathbb{E}[\mathcal{P}(t) \mid \mathcal{F}^\mathcal{Z}_t \vee \mathcal{F}^r_{t-1}]$ be the reserve calculated with the previous interest rate curve. Since increasing the number of policies does not diversify the market risk, we focus on validating the non-financial model elements. For one policy, straightforward calculations give
\begin{align} \label{eq:increment}
    \frac{\kappa_{t-1}(t-1)}{\kappa_{t-1}(t)} \mathcal{V}^{\textnormal{pc}}(t) -    \mathcal{V}(t-1) =-\int_{(t-1,t]} \frac{\kappa_{t-1}(t-1)}{\kappa_{t-1}(s)} \mathcal{B}(\diff s)+\mathbb{E}[\mathcal{P}(t-1) \mid \mathcal{F}^\mathcal{Z}_t \vee \mathcal{F}^r_{t-1}] - \mathcal{V}(t-1).
\end{align}
Hence, the change in the reserves where the interest rate curve is kept fixed at the previous interest rate curve and discounted one time unit is equal to the realized payments discounted to the previous time unit plus a term which has mean zero conditional on $\mathcal{F}^\mathcal{Z}_{t-1} \vee \mathcal{F}^r_{t-1}$. In other words, the latter term has mean zero if the non-financial part of the model is correctly specified no matter the prior financial and non-financial developments. Since this non-financial risk can be diversified away by increasing the size of the portfolio, it becomes negligible when summing Equation~\eqref{eq:increment} over many policies. Hence, the height of the stacked curves of $t \mapsto \mathcal{V}(0)+\sum_{m=1}^t  \kappa_{m-1}(m-1)/\kappa_{m-1}(m) \times \mathcal{V}^{\textnormal{pc}}(m) -    \mathcal{V}(m-1)$ and $t \mapsto \mathcal{B}(0) + \sum_{m=1}^t \int_{(m-1,m]} \kappa_{m-1}(m-1) /\kappa_{m-1}(s) \: \mathcal{B}(\diff s)$ summed over the policies is excepted to stay constant throughout when the number of policies used in the sample is sufficiently large.
\demormk{}
\end{remark}

\section{Estimation} \label{sec:EstimationAppendix}

This section details how to embed the statistical problem of this paper into the one from~\citet{Buchardt.etal:2023b} and provides a more accessible exposition of their estimation procedure applied to the current setting.

\subsection{Statistical model}
 
The disability and reactivation events may both be affected by reporting delays, but the death event is not. The disability reporting delay is $U_{\mathcal{I}} : \Omega \mapsto \mathbb{R}_+$ with $U_{\mathcal{I}} = 1_{(\tau_{\mathcal{I}} < \infty)}(T_{\{2\}}-\tau_{\mathcal{I}})$. Estimating the IBNR-factor $I_i(s,t)$ when $s < \infty$ is equivalent to estimating the disability reporting delay distribution since
\begin{align*}
    I_i(s,t) = \mathbb{P}(\textnormal{CBNR}_t \mid \tau_{\mathcal{I}} = s, Y_{\tau_{\mathcal{I}}} = i) = \mathbb{P}(U_{\mathcal{I}} > t-s \mid \tau_{\mathcal{I}} = s, Y_{\tau_{\mathcal{I}}} = i).
\end{align*}
Let $R : \Omega \mapsto \mathbb{R}_+$ with $R = \inf \{ s \geq 0 : G_s = G_\infty\}$ be the final time where payments are stopped. The reactivation reporting delay $U_r : \Omega \mapsto \mathbb{R}_+$ is $U_r = 1_{(\tau_{\{r\}} < \infty)}(R-\tau_{\{r\}})$. Note that there is no reporting delay when the insurer terminates the running payments, but a jump of $Z^{(1)}$ from state $2$ to state $3$ or state $5$ which triggers backpay may lead to a reactivation reporting delay. Estimation of the reactivation reporting delay distribution is not needed for reserve calculation but is needed for our proposed estimator of the valid time hazards.

Both the disability and reactivation events also have non-trivial adjudications while death events do not. We can calculate the adjudication probability $\mathbb{P}(X_{G_t} = (G_t,Z^{(3)}_t,W_t) \mid \mathcal{F}_t^\mathcal{Z})$
as an absorption probability on the state space $\mathcal{J}^\omega=\{1,2,3,4,5\}$ depicted in Figure~\ref{fig:Adjudication} with a set of $\mathcal{F}^\mathcal{Z}$-predictable transition intensities $\omega_{jk} : \Omega \times \mathbb{R}_+  \mapsto \mathbb{R}_+$ $(j,k \in \mathcal{J}^\omega,j \neq k)$. Disability benefits are awarded if and only if the process is absorbed in state 3 or 5. The multistate model of Figure~\ref{fig:Adjudication} starts each time a disability or reactivation event is reported. A disability adjudication starts in state 1, while a reactivation adjudication starts in state 2. Note that since the adjudication hazards are $\mathcal{F}^\mathcal{Z}$-predictable, they can and will be different when the adjudication pertains to a disability or a reactivation event. The shared notation for the adjudication hazards and the state space is chosen for parsimony.

Let $\sigma(t) = \inf\{ v \geq 0 : W_v > W_t \}$ be the next time where disability benefits are awarded after time $t$. The multistate model of Figure~\ref{fig:Adjudication} corresponds to modeling $s \mapsto \mathcal{Z}_s$ on the interval $\left(t,\sigma(t)\right]$ using the self-exciting filtration $\mathcal{F}^\mathcal{Z}$ except that the mark at time $\sigma(t)$ is the reduced mark $(Z_{\sigma(t)},1_{(W_{\sigma(t)} > W_t)})$ as opposed to the full mark $\mathcal{Z}_{\sigma(t)}$. When the object of interest is the adjudication probability, there is no need to model the full mark, which would entail modeling how much backpay was awarded, and as a consequence also whether the insured receives running benefits after the award or not.

\begin{figure}[H] \label{fig:Adjudication}
\centering
\scalebox{0.8}{
   \begin{tikzpicture}[node distance=8em, auto]
	\node[punkt] (g) {$3$};
	\node[punkt, right=2cm of g] (i1) {$1$};
        \node[punkt, right=2cm of i1] (i2) {$2$};
        \node[, right=0.8cm of i1] (a) {};
        \node[punkt, below=1.5cm of a] (a2) {$4$};
        \node[punkt, left=1.5cm of a2] (a3) {$5$};
         \path (i1) edge [pil] node [above=0.155cm]  {$\omega_{13}$} (g)
	;
        \path (i1) edge [pil, bend right=20] node [below=0.15cm]  {$\omega_{12}$} (i2)
	;
         \path (i2) edge [pil, bend right=20] node [above=0.1cm]  {$\omega_{21}$} (i1)
	;
         \path (i1) edge [pil] node [left=0.15cm]  {$\omega_{14}$} (a2)
	;
         \path (i2) edge [pil] node [right=0.15cm]  {$\omega_{24}$} (a2)
	;
          \path (i1) edge [pil] node [right=0.15cm]  {$\omega_{15}$} (a3)
	;
    \end{tikzpicture}
}
\caption{Multistate model for adjudications. Active report is 1, inactive report is 2, awarded is 3, dead without award is 4, and dead with award is 5.}
\end{figure}
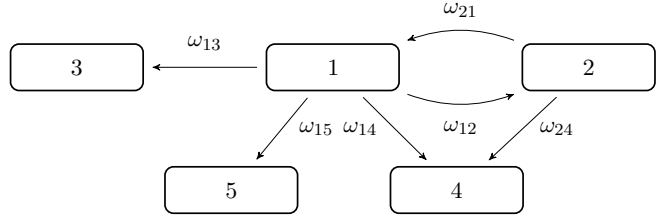

\noindent Let the observation window be $[0,\eta]$ and let the valid time process $X$ be subject to independent left-truncation and right-censoring, hence being observed on a random interval $(V,C] \subseteq [0,\eta]$ even in the absence of reporting delays and adjudication processes.  Events that occurred in $(V,C]$ can be reported and adjudicated up until time $\eta$. Baseline covariates, meaning covariates that are known at time $0$ in the valid time and transaction time filtrations, can easily be incorporated here and in the rest of the paper by conditioning on them throughout. Define the parameter spaces $\mathbb{G}$, $\mathbb{F}$, and $\Theta$ being subsets of Euclidean spaces. The statistical model is denoted $\mathcal{P} = \{ P_{(g,f,\theta)} : (g,f,\theta) \in \mathbb{G} \times \mathbb{F} \times \Theta\}$. The adjudication intensities under $P_{(g,f,\theta)}$ only depend on $(g,f,\theta)$ through $g$ and are denoted $t \mapsto \omega_{jk}(t;g)$. The reporting delay distributions under $P_{(g,f,\theta)}$ only depend on $(g,f,\theta)$ through $f$ and are denoted $t \mapsto \mathbb{P}(U_{\mathcal{I}} \leq t \mid \tau_{\mathcal{I}}, Y_{\tau_{\mathcal{I}}}, f)$ and $t \mapsto \mathbb{P}(U_r \leq t \mid \tau_{\mathcal{I}}, Y_{\tau_{\mathcal{I}}},  \tau_{\{r\}}, f)$. The valid time hazards under $P_{(g,f,\theta)}$ only depend on $(g,f,\theta)$ through $\theta$ and are denoted $(t,u) \mapsto \mu_{jk}(t,u;\theta)$. We now describe the proposed estimators.

\subsection{Adjudication probabilities}

 Let $N^{\omega}_{jk} : \Omega \times \mathbb{R}_+ \mapsto \mathbb{N}_0$ denote the counting processes on $\mathcal{J}^\omega$ where jumps due to starting a new adjudication are excluded. The objective function is the log-likelihood, which for one insured is 
$$\ell_\omega(g) = \sum_{\substack{j,k \in \mathcal{J}^\omega \\ j \neq k}} \int_{(0,\eta]} \log \{\omega_{jk}(t;g)\} \: N^\omega_{jk}(\diff t)-\int_{(0,\eta]} \omega_{jk}(t;g) \diff t.$$ 
For $n$ i.i.d.\ insured, the log-likelihood which we denote $\ell_\omega^{(n)}(g)$ is a sum over $n$ such terms.  The estimator is the maximum likelihood estimator $\hat{g}_n = \argmax_{g \in \mathbb{G}} \ell_\omega^{(n)}(g)$. The argmax can be found using existing \texttt{glm} software packages, see Section D in the supplementary material of~\citet{Buchardt.etal:2023b} for details.

\subsection{Reporting delay distribution}

The reporting delays are right-truncated since only jumps reported before $\eta$ are part of the sample. We first discuss disability reporting delays and subsequently reactivation reporting delays. Since we for disability reporting delays condition on $\tau_{\mathcal{I}}$, only disability claims that will be awarded should be included in the estimation of the reporting delay distribution for disability events. We accommodate this by weighting the relevant objective function with the adjudication probability. Assume that $U_{\mathcal{I}}$ given $(\tau_{\mathcal{I}},Y_{\tau_{\mathcal{I}}})$ has density with respect to a common reference measure across $\mathcal{P}$. We informally denote this density evaluated at a point $u$ by $\frac{\diff }{\diff u}\mathbb{P}(U_{\mathcal{I}} \leq u \mid \tau_{\mathcal{I}}, Y_{\tau_{\mathcal{I}}}, f)$. For one subject, the objective function for an observed disability reporting delay $u$ is 
\begin{align*}
    \ell_{U_\mathcal{I}}(f;\hat{g}_n) &= \mathbb{P}(X_{Z^{(2)}_\eta} = (Z^{(2)}_\eta,Z^{(3)}_\eta,0) \mid \mathcal{F}^\mathcal{Z}_\eta, \hat{g}_n) \times \log \bigg\{ \frac{ \frac{\diff }{\diff u}\mathbb{P}(U_\mathcal{I} \leq u \mid \tau_{\mathcal{I}}=Z^{(2)}_\eta, Y_{\tau_{\mathcal{I}}}=Z^{(3)}_\eta, f)}{\mathbb{P}(U_\mathcal{I} \leq \eta-\tau_{\mathcal{I}} \mid \tau_{\mathcal{I}}=Z^{(2)}_\eta, Y_{\tau_{\mathcal{I}}}=Z^{(3)}_\eta, f)}  \bigg\}.
\end{align*} 
Here the first term is the adjudication probability $\mathbb{P}(X_{G_\eta} = (G_\eta,Z^{(3)}_\eta,0) \mid \mathcal{F}^\mathcal{Z}_\eta, \hat{g}_n)$ on $(\textnormal{RBNSi}_\eta)$ and is $1$ if the disability has been awarded. The objective function for $n$ i.i.d.\ subjects is a sum of $n$ such terms and we denote this by $\ell_{U_{\mathcal{I}}}^{(n)}(f;\hat{g}_n)$. 

 Similarly, assume $U_r$ given $(\tau_{\mathcal{I}},Y_{\tau_{\mathcal{I}}},\tau_{\{r\}})$ has density with respect to a common reference measure across $\mathcal{P}$ which we informally denote $\frac{\diff}{\diff w} \mathbb{P}(U_r \leq w \mid \tau_{\mathcal{I}}, Y_{\tau_{\mathcal{I}}}, \tau_{\{r\}}, f)$. For one subject, the objective function for an observed reactivation reporting delay $w$ is 
\begin{align*}
    \ell_{U_r}(f;\hat{g}_n)&= (1-\mathbb{P}(X_{G_\eta} = (G_\eta,Z^{(3)}_\eta,W_\eta) \mid \mathcal{F}_\eta^\mathcal{Z},\hat{g}_n)) \\
    & \quad \times \log \bigg\{ \frac{ \frac{\diff}{\diff w} \mathbb{P}(U_r \leq w \mid \tau_{\mathcal{I}}=Z^{(2)}_\eta, Y_{\tau_{\mathcal{I}}}=Z^{(3)}_\eta, \tau_{\{r\}}=G_\eta, f)}{\mathbb{P}(U_r \leq \eta - \tau_{\{r\}} \mid  \tau_{\mathcal{I}}=Z^{(2)}_\eta, Y_{\tau_{\mathcal{I}}}=Z^{(3)}_\eta, \tau_{\{r\}}=G_\eta, f)} \bigg\}.
\end{align*}
Let the objective function for $n$ i.i.d.\ subjects be denoted $\ell_{U_r}^{(n)}(f;\hat{g}_n)$. The estimator is then $\hat{f}_n = \argmax_{f \in \mathbb{F}} \ell_{U_{\mathcal{I}}}^{(n)}(f;\hat{g}_n)+\ell_{U_r}^{(n)}(f;\hat{g}_n)$. 

\begin{remark} (Chain ladder estimator.) \\
    The reporting delay distribution can alternatively be estimated using chain ladder, see for example Section 5 in~\citet{Bucher:Rosenstock:2024}, but the large-sample properties are then not a special case of~\citet{Buchardt.etal:2023b} since they consider parametric estimators. Note that chain ladder in this case does not estimate a reserve but rather an element of the individual reserve, namely the IBNR-factor. As seen in Section~\ref{sec:reserving}, an assumption like Assumption~\ref{assumption:Independence} is needed for the IBNR reserve to approximately decompose into an IBNR-factor-adjusted frequency multiplied with a classic valid time disability reserve as the associated claim size.
    \demormk
\end{remark}

\subsection{Valid time hazards} \label{sec:estHaz}
  
Make a partition of the observation window $0=t_0<t_1<\dots<t_L=\eta$ and let $O_{jk}(\ell) = N^\eta_{jk}(t_{\ell+1})-N^\eta_{jk}(t_{\ell})$ and $E_j(\ell) = \int_{(t_\ell, t_{\ell + 1}]} 1_{(Y^\eta_s = j)} \diff s$ $(\ell=0,\dots,L-1)$ be the occurrences and exposures based on $H^\eta$ for a single insured. To describe how the duration process is affected by accepting or rejecting a disability claim, we introduce the auxiliary durations $U_j(\ell) = U^\eta_{t_\ell}$ for $j \in \mathcal{J}$. Note that there is initially no dependence on $j$.

If there is an unadjudicated disability at time $\eta$, the occurrences, exposures, and durations are modified:
\begin{align*}
    O_{ai}(\ell) &\leftarrow \mathbb{P}(X_{G_\eta} = (G_\eta,i,0) \mid \mathcal{F}^\mathcal{Z}_\eta, \hat{g}_n),& G_\eta \in (t_\ell,t_{\ell + 1}], \\  
    E_{a}(\ell) &\leftarrow (t_{\ell+1}-t_\ell) \times (1-\mathbb{P}(X_{G_\eta} = (G_\eta,i,0) \mid \mathcal{F}^\mathcal{Z}_\eta, \hat{g}_n)), & G_\eta \in (0,t_\ell], \\
    E_{i}(\ell) &\leftarrow (t_{\ell+1}-t_\ell) \times \mathbb{P}(X_{G_\eta} = (G_\eta,i,0) \mid \mathcal{F}^\mathcal{Z}_\eta, \hat{g}_n), & G_\eta \in (0,t_\ell], \\
    U_i(\ell) &\leftarrow t_\ell - G_\eta, & G_\eta \in (0,t_\ell],
\end{align*}
for $i = Z^{(3)}_\eta$. If there is an unadjudicated reactivation at time $\eta$, the occurrences, exposures, and durations are modified:
\begin{align*}
    O_{ir}(\ell) &\leftarrow (1-\mathbb{P}(X_{G_\eta} = (G_\eta,i,W_\eta) \mid \mathcal{F}_\eta^\mathcal{Z},\hat{g}_n)),& G_\eta \in (t_\ell,t_{\ell + 1}],\\ 
    E_{i}(\ell) &\leftarrow (t_{\ell+1}-t_\ell) \times \mathbb{P}(X_{G_\eta} = (G_\eta,i,W_\eta) \mid \mathcal{F}_\eta^\mathcal{Z},\hat{g}_n), & G_\eta \in (0,t_\ell],\\  
    E_{r}(\ell) &\leftarrow (t_{\ell+1}-t_\ell) \times (1-\mathbb{P}(X_{G_\eta} = (G_\eta,i,W_\eta) \mid \mathcal{F}_\eta^\mathcal{Z},\hat{g}_n)), & G_\eta \in (0,t_\ell], \\
    U_i(\ell) &\leftarrow t_\ell - Z^{(2)}_\eta, & G_\eta \in (0,t_\ell],
\end{align*}
for $i = Z^{(3)}_\eta$. Finally, in all cases, the exposures are modified with the reporting delay distribution:
\begin{align*}
    E_{ai}(\ell) &\leftarrow E_a(\ell) \times \mathbb{P}(U_\mathcal{I} \leq \eta-\tau_{\mathcal{I}} \mid \tau_{\mathcal{I}}=t_\ell, Y_{\tau_{\mathcal{I}}}=i, \hat{f}_n), \\
    E_{ir}(\ell) &\leftarrow E_i(\ell) \times \mathbb{P}(U_r \leq \eta-\tau_{\{r\}} \mid \tau_{\mathcal{I}}=Z^{(2)}_\eta, Y_{\tau_{\mathcal{I}}}=i, \tau_{\{r\}} = t_\ell, \hat{f}_n),
\end{align*}
for all $i \in \mathcal{I}$. For the remaining transitions set $E_{jk}(\ell) \leftarrow E_j(\ell)$.

The objective function $\ell(\theta;\hat{g}_n,\hat{f}_n)$ is the log-likelihood resulting from assuming that $(O_{jk}(\ell))_{j,k,\ell}$ are independent Poisson distributed random variables with mean values \\ $(\mu_{jk}(t_\ell, U_{j}(\ell);\theta)E_{jk}(\ell))_{j,k,\ell}$. Let the objective function for $n$ i.i.d.\ subjects be denoted $\ell^{(n)}(\theta;\hat{g}_n,\hat{f}_n)$. The estimator is $\hat{\theta}_n = \argmax_{\theta \in \Theta} \ell^{(n)}(\theta;\hat{g}_n,\hat{f}_n)$, which corresponds to the Poisson approximation from~\citet{Buchardt.etal:2023b}. The argmax can be found using existing \texttt{glm} software packages when the modified occurrences, exposures, and durations have been constructed. Note that if two occurrences have the same mean value, these occurrences and their corresponding exposures can be summed without changing the objective function. This aggregation may save memory and speed up computations.

\subsection{Asymptotic properties} \label{subsec:Asymp}

The Poisson approximation introduces bias that does not vanish asymptotically, but which is small when the hazards and reporting delays are small, see Section 3.3 of~\citet{Buchardt.etal:2023b} or Section B.2 of their supplementary material. The approximation error for the disability hazard is expected to be negligible since the hazard appears to be smaller than $10^{-2}$ by some margin. There is also a small approximation bias coming from implementing the Poisson approximation via occurrences and exposures, see Section D in the supplementary material of~\citet{Buchardt.etal:2023b} for details. Denote by $B_n = \hat{\theta}_n-\hat{\theta}^{\textnormal{full}}_n$ the approximation bias, where $\hat{\theta}^{\textnormal{full}}_n$ is the estimator based on the full procedure. We use $\overset{\textnormal{a}}{\sim}$ to denote asymptotic distribution. The asymptotic properties of the estimators are given in Proposition~\ref{prop:asymp}.

\begin{proposition} \label{prop:asymp}
     Under Assumptions 1-3 and 5-8 from~\citet{Buchardt.etal:2023b}, which are standard integrability and smoothness conditions,
     \begin{align*}
         (\hat{g}_n,\hat{f}_n,\hat{\theta}_n) \overset{\textnormal{a}}{\sim} \mathcal{N}( (0,0,B_n), \Sigma ) 
     \end{align*}
     for a non-singular covariance matrix $\Sigma$.
\end{proposition}
\begin{proof}
    This follows directly from Theorem 1 of~\citet{Buchardt.etal:2023b}.
\end{proof}
